\documentclass[english,12pt]{article}

\usepackage{hyperref}

\usepackage[T1]{fontenc}
\usepackage[latin1]{inputenc}
\usepackage{geometry}
\usepackage{float}
\usepackage{mathrsfs}
\usepackage{amsmath}
\usepackage{amssymb}
\usepackage{graphicx}

\hypersetup{
    colorlinks=true,
    linkcolor=blue,
    filecolor=magenta,      
    urlcolor=cyan,
    citecolor = red,
}

\makeatletter

\usepackage{algorithm,algpseudocode}

\usepackage{amsthm}

\usepackage{mathrsfs}

\usepackage{amsfonts}

\usepackage{epsfig}

\usepackage{bm}

\usepackage{mathrsfs}

\usepackage{enumerate}

\@ifundefined{definecolor}{\@ifundefined{definecolor}
 {\@ifundefined{definecolor}
 {\usepackage{color}}{}
}{}
}{}

\usepackage{subfig}\usepackage[all]{xy}

\newtheorem{theorem}{Theorem}[section]
\newtheorem{lem}{Lemma}[section]
\newtheorem{rem}{Remark}[section]

\newcounter{hypA}

\newcounter{hypB}

\newcounter{hypD}
\newenvironment{hypD}{\refstepcounter{hypD}\begin{itemize}
 \item[({\bf D\arabic{hypD}})]}{\end{itemize}}

\usepackage{babel}\date{}

\usepackage{babel}

\makeatother

\usepackage{babel}

\usepackage{mathtools}

\DeclarePairedDelimiter\floor{\lfloor}{\rfloor}

\usepackage[title]{appendix}

\begin{document}

\begin{center}

{\Large \textbf{Score-Based Parameter Estimation for a Class of \\[0.2cm] Continuous-Time State Space Models}}

\vspace{0.5cm}

BY ALEXANDROS BESKOS$^{1}$, DAN CRISAN$^{2}$, AJAY JASRA$^{3}$, NIKOLAS KANTAS$^{2}$ \& HAMZA RUZAYQAT$^{3}$

{\footnotesize $^{1}$Department of Statistical Science, University College London, London, WC1E 6BT, UK.}
{\footnotesize E-Mail:\,} \texttt{\emph{\footnotesize a.beskos@ucl.ac.uk}}\\
{\footnotesize $^{2}$Department of Mathematics, Imperial College London, London, SW7 2AZ, UK.}
{\footnotesize E-Mail:\,} \texttt{\emph{\footnotesize d.crisan@imperial.ac.uk, n.kantas@imperial.ac.uk}}\\
{\footnotesize $^{3}$Computer, Electrical and Mathematical Sciences and Engineering Division, King Abdullah University of Science and Technology, Thuwal, 23955, KSA.}
{\footnotesize E-Mail:\,} \texttt{\emph{\footnotesize ajay.jasra@kaust.edu.sa, hamza.ruzayqat@kaust.edu.sa}}
\end{center}


\begin{abstract}
We consider the problem of parameter estimation for a class of continuous-time state space models (SSMs).
In particular, we explore the case of a partially observed diffusion, 
with data also arriving according to a diffusion process. Based upon a standard identity of the score function, we consider two particle
filter based methodologies to estimate the score function. Both methods rely on an
online estimation algorithm for the score function, as described, e.g., in~\cite{backward}, of 
 $\mathcal{O}(N^2)$ cost, with $N\in\mathbb{N}$
the number of particles. The first approach employs a simple Euler discretization
and standard particle smoothers 
and is of cost $\mathcal{O}(N^2 + N\Delta_l^{-1})$ per unit time, where \mbox{$\Delta_l=2^{-l}$}, $l\in\mathbb{N}_0$, is the time-discretization step. 
The second approach is new and based upon a novel diffusion bridge construction.
It yields a new backward type Feynman-Kac formula in continuous-time for the score function and is presented 
along with a particle method for its approximation. Considering a time-discretization, the cost is $\mathcal{O}(N^2\Delta_l^{-1})$ per unit time. 
To improve computational costs, we then consider multilevel methodologies for the score function. We illustrate our parameter estimation method via stochastic gradient approaches in several numerical examples.\\
\noindent \textbf{Keywords}: Score Function, Parameter Estimation, Particle Filter, Diffusion Bridges.\\
\noindent\textbf{AMS subject classifications}: 65C05, 65C35, 60G35, 60J60, 60J65, 60H10, 60H35, 91G60
\end{abstract}

\section{Introduction}

We consider the problem of parameter estimation for continuous-time SSMs. These are models comprising stochastic differential equations (SDEs)
describing a hidden dynamic state and their observations.
Such models are ubiquitous in
a large number of practical applications in science, engineering, finance, and economics, see \cite{handbook} for an overview.
Inference in SSMs, also known as hidden Markov models, hinges upon computing conditional probability distributions of 
the dynamic hidden state given the acquired observations and unknown static parameters. This is referred to as the stochastic filtering problem, 
which is in general intractable, but reliable 
numerical approximations are routinely available \cite{crisan_bain,handbook}. 
The problem of inferring the unknown parameters is more challenging.
In this paper we focus on maximum likelihood inference and gradient methods that are performed in an online manner. The offline or iterative case given a fixed batch of observations can also be treated using our proposed methods.
The approach considered in this article is to make use of the gradient of the log-likelihood, commonly referred to as the \emph{score function}, 
within a stochastic gradient algorithm 
(see e.g.~\cite{backward,forward_smoothing}). 
Intrinsically, there are several challenges arising with such an approach. Firstly,
when one adopts a continuous-time model and 
assumes access to arbitrarily high frequency observations, one does not observe 
in practice truly in continuous-time, therefore some sort of time-discretization is required.
Secondly, in both discrete-time and continuous-time formulations, there are very few cases when the score function is analytically available. 
Both these issues imply that numerical approximations are required.

We consider two different approaches for the numerical approximation of the score function. 
The first is to simply time-discretize a representation of the score function and then apply 
discrete-time numerical approximation schemes \cite{backward, forward_smoothing, paris}. 
The second, is to develop a numerical approximation scheme directly on continuous-time path-space to estimate the score function
and then (necessarily) discretize the algorithm in time. The order of designing the estimation method and time-discretization can be rather important. Often  the second approach
is preferable in terms of both performance and robustness as the discretization mesh vanishes, 
see e.g.~\cite{omiros_ISdiffusions, omiros_DA}. 
We then use the score estimate for implementing recursive maximum likelihood, where the parameters are updated at unit time intervals (\cite{legland1997,poyia}). 
The particular choice of time interval length is without loss of generality, and allows the score to accumulate sufficient information from the observations before updating the parameters.

In the first approach, we consider a well-known expression for the score function, for instance as given in \cite{campillo}.
Given this formula, one can produce an Euler discretized version of the score and work in discrete-time but with high frequency. 
The score is an expectation 
of an additive functional of the hidden state path conditioned on the available observations, which is commonly referred to as the smoothing distribution. 
In this context, many well-known particle smoothing schemes can now be adopted, such as the ones described in  \cite{backward, paris}. 
These latter approaches are simulation-based schemes whose convergence is based upon the number of samples $N\in\mathbb{N}$.
We prove some technical results for the discretized problem, which together with the work in \cite{backward} 
allow us  conjecture that to obtain a mean square error (MSE) of $\mathcal{O}(\epsilon^2)$,  for given $\epsilon>0$,  we require a computational cost of 
$\mathcal{O}(\epsilon^{-4})$ per unit time. The latter derives from an algorithmic cost of $\mathcal{O}(N^2 + N\Delta_l^{-1})$ per unit time, 
where $\Delta_l=2^{-l}$ is the time-discretization. As we explain later in the article, this corresponds to a best case scenario, 
due to the intrinsic nature of the algorithm. In particular, we start with a continuous-time formula and time-discretize it, but
the deduced numerical algorithm can be problematic in terms of computational complexity as $l\rightarrow\infty$. 
Whilst in some scenarios one does not observe any issues, examples can be found where the variance of the method can explode as $l$ grows (\cite{shouto}), 
thus putting into question the validity of the conjecture on the cost to achieve an MSE of $\mathcal{O}(\epsilon^2)$.

This motivates the introduction of our second approach, where we build upon a change of measure technique proposed in \cite{schauer}.
This latter approach has so far been used in very different contexts than the present paper, namely related to discretely
observed diffusions for Bayesian inference and Markov chain Monte Carlo (MCMC) \cite{vd_meulen_guided_mcmc, whitaker_guided} or smoothing for potentially non-linear observation functions   \cite{vd_meulen_smoothing}.
Our approach is a data augmentation scheme, whereby, at unit time intervals, the end points for the hidden state are sampled and the path is connected using
diffusion bridges. Then, starting again from the formula for the score function in \cite{campillo}, we will use this change of measure associated to
a diffusion bridge and its driving Brownian noise; see also \cite{shouto} where a related approach is used for a different class of problems. 
Based upon this change of measure we develop a new backward type Feynman-Kac formula in continuous-time. 
This new formula facilitates an adaptation of the method in \cite{backward} in true continuous-time, albeit one cannot apply it in practice. 
We time-discretize the algorithm and conjecture that to obtain a MSE of $\mathcal{O}(\epsilon^2)$, for given $\epsilon>0$, we need a cost of 
$\mathcal{O}(\epsilon^{-6})$ per unit time, which derives from an algorithmic cost of  $\mathcal{O}(N^2\Delta_l^{-1})$ per unit time. 
We note however, that this computational complexity will not explode with increasing $l$ as may be the case in the first approach. 
To improve the cost required for a given MSE, we develop a novel  multilevel Monte Carlo extension that can, in some cases, achieve a 
MSE of $\mathcal{O}(\epsilon^2)$  at a cost per unit time of $\mathcal{O}(\epsilon^{-4})$. 
We remark that although our MSE-cost statements are based upon conjectures, they are verified numerically. 
Direct proofs of these require substantial technical results that will be the topic of future work.


\subsection{Contributions and Organization}
We conclude this introduction by emphasizing that our contributions are aimed to deal with both continuous-time observations and hidden states. 
As mentioned earlier this poses very particular challenges relative to earlier works that deal with filtering and smoothing 
when discrete-time observations/models are used as in \cite{botha_sde_pmcmc,etienne_bwd, glo, vd_meulen_smoothing, whitaker_guided,  sarkka_pf_girsanov}.
None of these works look at continuous-time observations. Similarly online likelihood estimation of the parameters using the score function 
is considered  in \cite{etienne_bwd, glo} only for the case of discrete-time observations of hidden diffusions.

The contributions of this paper can be summarized as follows:
\begin{itemize}
 \item We investigate the efficiency and accuracy of two fundamentally different numerical approximations of the score function on its own and 
 when used for the purpose of recursive maximum likelihood. Both methods rely on fairly standard tools such as changes of measure, particle smoothing and Euler time-discretization.
 \item We provide a detailed discussion on the computational complexity of each method. We illustrate the performance and computational cost for several models in numerical examples that consider estimation of the score function and parameter estimation.
 \item The second approach is a novel method that operates directly on the path-space. The approach improves performance and is robust
 to arbitrarily small time-discretization at the expense of additional computational cost. The latter is reduced using a new Multilevel Particle Filter; see \cite{mlpf, high_freq_ml} for some existing approaches.
\end{itemize}

This article is structured as follows. In Section \ref{sec:problem} the basic problem is formulated in continuous-time.
In Section \ref{sec:method1} we consider our first method for online score estimation and explain the various features associated to it. In Section \ref{sec:method2} our second method for online score estimation is developed.
In Section \ref{sec:numerics} our numerical results are presented. 

\subsection{Notation}

Let $(\mathsf{X},\mathcal{X})$ be a measurable space.
We write $\mathcal{B}_b(\mathsf{X})$ for the set of bounded measurable functions,
 $\varphi:\mathsf{X}\rightarrow\mathbb{R}^d$, $d\in\mathbb{N}$, and $\mathcal{C}(\mathsf{X})$ for the  continuous ones.
Let $\varphi:\mathbb{R}^d\rightarrow\mathbb{R}$; $\textrm{Lip}_{\|\cdot\|_2}(\mathbb{R}^{d})$ denotes the collection of real-valued functions that are Lipschitz w.r.t. the Euclidean distance $\|\cdot\|_2$ 
. That is, $\varphi\in\textrm{Lip}_{\|\cdot\|_2}(\mathbb{R}^{d})$ if there exists a $C<+\infty$ such that for any $(x,y)\in\mathbb{R}^{2d}$,
$$
|\varphi(x)-\varphi(y)| \leq C\|x-y\|_2.
$$
$\mathcal{N}_s(\mu,\Sigma)$ 
denotes an $s$-dimensional Gaussian law of mean $\mu$ and covariance $\Sigma$; if $s=1$ we omit subscript $s$. For a vector/matrix $X$, $X^*$ denotes the transpose of $X$.
For $A\in\mathcal{X}$, $\delta_A(du)$ denotes the Dirac measure of $A$, and if $A=\{x\}$ with $x\in \mathsf{X}$, we write $\delta_x(du)$. 
For a vector-valued function in $d$-dimensions $\varphi(x)$ (resp.~$d$-dimensional vector~$x$) we write the $i^{\textrm{th}}$-component, $1\le i\le d$, as $\varphi^{(i)}(x)$ (resp.~$x^i$). For a $d\times q$ matrix $x$ we write the $(i,j)^{\textrm{th}}$-entry as $x^{(ij)}$. $\mathbb{N}=\{1,2,\dots\}$ and $\mathbb{N}_0=\mathbb{N}\cup\{0\}$.
\section{Problem Formulation}\label{sec:problem}

\subsection{Preliminaries}
We consider the parameter space $\theta\in\Theta\subset\mathbb{R}^{d_{\theta}}$, with $\Theta$ being compact, $d_{\theta}\in\mathbb{N}$. 
The  stochastic processes $\{Y_t\}_{t\geq 0}$, $\{X_t\}_{t\geq 0}$ of interest are defined upon the probability triple $(\Omega,\mathcal{F},\mathbb{P}_{\theta})$, 
with $Y_t\in\mathbb{R}^{d_y}$, $X_t\in\mathbb{R}^{d_x}$, $d_y, d_x \in\mathbb{N}$, initial conditions $X_0=x_*\in\mathbb{R}^{d_x}$, $Y_0=y_*\in\mathbb{R}^{d_y}$, and are determined as the 
solution of the system of SDEs: 
\begin{align}
dY_t & =  h_{\theta}(X_t)dt + dB_t; \label{eq:obs}\\
dX_t & =  b_{\theta}(X_t)dt + \sigma(X_t)dW_t. \label{eq:state}
\end{align}
Here,  for each $\theta\in\Theta$, $h_{\theta}:\mathbb{R}^{d_x}\rightarrow\mathbb{R}^{d_y}$, $b_{\theta}:\mathbb{R}^{d_x}\rightarrow\mathbb{R}^{d_x}$, $\sigma:\mathbb{R}^{d_x}\rightarrow\mathbb{R}^{d_x\times d_x}$, with $\sigma$ being of full rank, and $\{B_t\}_{t\geq 0}, \{W_t\}_{t\geq 0}$
are independent standard Brownian motions of dimension $d_y$, $d_x$ respectively.

To minimize technical difficulties, the following assumptions are made throughout the paper:
\begin{hypD}
	\label{hyp_diff:1}
	\begin{itemize}
		\item[(i)] $\sigma$ is continuous, bounded;  $a(x):=\sigma(x)\sigma(x)^*$ is uniformly elliptic; \vspace{0.1cm}
		\item[(ii)] for each $\theta$, $h_\theta$ and $b_{\theta}$ are bounded, measurable;  $h_{\theta}^{(i)}\in\textrm{Lip}_{\|\cdot\|_2}(\mathbb{R}^{d_x})$, $1\le i \le d_y$;  \vspace{0.1cm}
		\item[(iii)] 
		the gradients $\nabla_{\theta}h_\theta:\mathbb{R}^{d_x}\rightarrow\mathbb{R}^{d_y\times d_{\theta}}$ and  $\nabla_{\theta}b_\theta:\mathbb{R}^{d_x}\rightarrow\mathbb{R}^{d_x\times d_{\theta}}$ exist, and are  continuous,  bounded, measurable; 
		$\nabla_{\theta}h_{\theta}^{(ij)}\in\textrm{Lip}_{\|\cdot\|_2}(\mathbb{R}^{d_x})$, $1\le i\le d_y$, \mbox{$1\le j\le d_{\theta}$};
		\item[(iv)] let $\phi_{\theta}(x)=(\nabla_{\theta}b_{\theta}(x))^*\,a(x)^{-1}\sigma(x)$; for any  $\theta$, $\phi_{\theta}^{(ij)}\in\textrm{Lip}_{\|\cdot\|_2}(\mathbb{R}^{d_x})$, $1\le i \le d_{\theta}$, \mbox{$1\le j\le d_{x}$}.
	\end{itemize}
\end{hypD}
%

%

\noindent 
We  introduce the probability measure $\overline{\mathbb{P}}_{\theta}$, defined via the Radon-Nikodym derivative:
\begin{equation}
\label{eq:Z}
Z_{t,\theta}:=\frac{d\mathbb{P}_{\theta}}{d\overline{\mathbb{P}}_{\theta}}\Big|_{\mathcal{F}_{_t}} = \exp\Big\{\int_{0}^t h_{\theta}(X_s)^*dY_s - \tfrac{1}{2}\int_{0}^t h_{\theta}(X_s)^*h_{\theta}(X_s)ds\Big\},
\end{equation}
with $\mathcal{F}_t=\sigma(\{X_s,Y_s\}_{0\le s\le t})$. 
Henceforth, $\overline{\mathbb{E}}_{\theta}$ denotes expectation 
w.r.t.~$\overline{\mathbb{P}}_{\theta}$, so that under  $\overline{\mathbb{P}}_{\theta}$, the process $\{X_t\}_{t\geq 0}$ follows the dynamics in~\eqref{eq:state}, whereas $\{Y_t\}_{t\geq 0}$ is a Brownian motion
independent of $\{X_t\}_{t\geq 0}$.  We define,  for $\varphi\in\mathcal{B}_b(\mathbb{R}^{d_x})$:
\begin{equation*}
\gamma_{t,\theta}(\varphi) :=
 \overline{\mathbb{E}}_{\theta}\,\big[\,\varphi(X_t)Z_{t,\theta}\,\big|\,\mathcal{Y}_t\,\big],
\end{equation*}
where $\mathcal{Y}_t$ is the filtration generated by the process $\{Y_s\}_{0\leq s \leq t}$. 
Our objective is to produce estimates of the gradient of the score function $\nabla_{\theta}\log(\gamma_{T,\theta}(1))$.
%
 %
\begin{rem}
To connect the changes of measures with standard likelihood derivations,  notice that -- via Girsanov's  theorem -- $Z_{t,\theta}$ is the density (w.r.t.~to a Wiener measure) of the distribution of $\{Y_{s}\}_{0\le s\le t}$ conditionally 
on $\{X_{s}\}_{0\le s\le t}$. Then,  $\gamma_{T,\theta}(1)$ integrates out $\{X_{t}\}_{0\le t\le T}$,  thus corresponds to the marginal density -- i.e., the likelihood -- of the observations  $\{Y_{t}\}_{0\le t\le T}$.
\end{rem} 

In our setting, the score function writes as
 (see e.g.~\cite{campillo}):
\begin{equation}\label{eq:gll}
\nabla_{\theta}\log(\gamma_{T,\theta}(1)) = \frac{\overline{\mathbb{E}}_{\theta}\,[\,\lambda_{T,\theta}Z_{T,\theta}\,|\,\mathcal{Y}_T\,]}{\overline{\mathbb{E}}_{\theta}\,[\,Z_{T,\theta}\,|\,\mathcal{Y}_T\,]},
\end{equation}
where we have defined: 
\begin{align*}
\lambda_{T,\theta} := \int_{0}^T(\nabla_{\theta}b_{\theta}(X_t))^* &a(X_t)^{-1}\sigma(X_t)dW_t
\\ &+ \int_{0}^T(\nabla_{\theta}h_{\theta}(X_t))^* dY_t - \int_{0}^T(\nabla_{\theta}h_{\theta}(X_t))^*h_{\theta}(X_t)dt.
\end{align*}
 For completeness, a derivation of \eqref{eq:gll} can be found in \autoref{sec:Deriv_Eq4}.
We remark that one can derive a formula for the score function when $\sigma$ depends upon $\theta$, which is given in Section \ref{sec:method2}.
We will assume throughout that $T\in\mathbb{N}$. Note also that an application of Bayes' rule gives that, almost surely:
\vspace{-8pt}
\begin{equation}\label{eq:gll1} 
\frac{\overline{\mathbb{E}}_{\theta}\,[\,\lambda_{T,\theta}Z_{T,\theta}\,|\,\mathcal{Y}_T\,]}{\overline{\mathbb{E}}_{\theta}\,[\,Z_{T,\theta}\,|\,\mathcal{Y}_T\,]} =
\mathbb{E}_{\theta}\,[\,\lambda_{T,\theta}\,|\,\mathcal{Y}_T\,], 
\end{equation}
where $\mathbb{E}_{\theta}$ denotes expectation w.r.t.~$\mathbb{P}_{\theta}$.

\subsection{Parameter Estimation}

In the offline case suppose one has obtained data $\{Y_t\}_{0\le t\le T}$. Then it is possible to perform standard gradient descent using \eqref{eq:gll1} and updating $\theta$ iteratively:
\begin{align}\label{eq:theta_update_offline}
\theta^{m+1}=\theta^{m}+\alpha_m\,\mathbb{E}_{\theta^m}\,[\,\lambda_{T,\theta^m}\,|\,\mathcal{Y}_T\,],
\end{align}
where $\alpha_m\in\mathbb{R}^+$, $m\in\mathbb{N}_{0}$, are decreasing step-sizes. 
Instead here we will mainly focus on an online gradient estimation procedure. To obtain this one can aim to maximize the following limiting \emph{average log-likelihood}, 
$$\mathcal{L}(\theta)=\lim_{t\rightarrow\infty}\frac{1}{t}\int_0^t\log\gamma_{s,\theta}(1)ds.$$ 
Let the filter be denoted as $\pi_{t,\theta}(\varphi)=\mathbb{E}_\theta[\varphi(X_t)\vert \mathcal{Y}_s]$ and using standard arguments (e.g. Lemma 3.29 p. 67 \cite{crisan_bain}) one can re-write 
$\gamma_{t,\theta}(1)$ as 
$$\log\gamma_{t,\theta}(1)=\int_{0}^{t}\pi_{s,\theta}(h)^{T}dY_s-\frac{1}{2}\int_{0}^{t}\pi_{s,\theta}(h)^{T}\pi_{s,\theta}(h)ds.$$
Under appropriate stability and regularity conditions for both $Y_t$ and $\pi_{t,\theta}$ (see \cite{surace18} for more details), then both $\mathcal{L}(\theta)$ and $\nabla_\theta\mathcal{L}(\theta)$ are ergodic averages.
This means one can implement stochastic gradient ascent using
either $\nabla_\theta\log\gamma_{t,\theta}(1)$ or $\nabla_\theta\log\gamma_{t_n,\theta}(1)-\nabla_\theta\log\gamma_{t_{n-1},\theta}(1)$ for any $t_{n-1}<t_n$ as estimates of $\mathcal{L}(\theta)$.
Given an initial $\theta_0\in\Theta$, as we obtain the observation path continuously in time, we will update $\theta$ at times $T\in\mathbb{N}$ using the following recursion:
\begin{align}\label{eq:theta_update}
\theta_T = \theta_{T-1} + \alpha_T\Big(\nabla_{\theta}\log(\gamma_{T,\theta_{T-1}}(1))-
\nabla_{\theta}\log(\gamma_{T-1,\theta_{T-1}}(1))
\Big)
\end{align}
where, for $T\in\mathbb{N}$, $\alpha_T\in\mathbb{R}^+$ is a collection of step-sizes that satisfy $\sum_{T\in\mathbb{N}}\alpha_T = \infty$, $\sum_{T\in\mathbb{N}}\alpha_T^2<\infty$ 
to ensure convergence of the estimation as $T\rightarrow\infty$; see \cite{BMP90,legland1997} for details. This scheme can provide an online estimate for the parameter vector as data arrive.
Steps are performed at $\mathcal{O}(1)$ times to ensure that enough information has accumulated to update the parameter. The adoption of unit times is made only for notational convenience.
As both recursions \eqref{eq:theta_update_offline} and \eqref{eq:theta_update}  cannot be computed exactly, we focus upon methodologies that
approximate the score function $\nabla_{\theta}\log(\gamma_{T,\theta}(1))$.


\section{Direct Feynman-Kac Formulation}\label{sec:method1}

\subsection{Discretized Model}\label{sec:disc_model}

In practice, we will have to work with a discretization of the model in \eqref{eq:obs}-\eqref{eq:state}.  
%
%
We assume access to path of the data $\{Y_t\}_{0\leq t \leq T}$ which is available up-to an (almost) arbitrarily fine level of time discretization. 
This would be a very finely discretized path, as accessing the actual continuous path of observation is not possible; this point is discussed later on. 
One could focus on a time-discretization of either side of \eqref{eq:gll1}, however, as is conventional
in the literature (e.g.~\cite{crisan_bain,high_freq_ml}) we focus on the left hand side.

Let $l\in\mathbb{N}_0$ and consider an Euler-Maruyama time-discretization with step-size $\Delta_l=2^{-l}$. 
That is, for $k\in\{1,2,\dots,T/\Delta_l\}$:
\begin{align}
\widetilde{X}_{k\Delta_l}  =  \widetilde{X}_{(k-1)\Delta_l} + b_{\theta}(\widetilde{X}_{(k-1)\Delta_l})\Delta_l + \sigma(\widetilde{X}_{(k-1)\Delta_l})(W_{k\Delta_l}-W_{(k-1)\Delta_l}),
\quad \widetilde{X}_{0}=x_*.
\label{eq:disc_state}
\end{align}
Note that the Brownian motion in \eqref{eq:disc_state} is the same as in \eqref{eq:state} under both $\mathbb{P}_{\theta}$ and $\overline{\mathbb{P}}_{\theta}$.
We set:
\begin{align}
\lambda_{T,\theta}^l(x_0,&x_{\Delta_l},\dots,x_T) :=  \sum_{k=0}^{T/\Delta_l-1}\Big\{\,(\nabla_{\theta} b_{\theta}(x_{k\Delta_l}))^*a(x_{k\Delta_l})^{-1}\sigma(x_{k\Delta_l})(W_{(k+1)\Delta_l}-W_{k\Delta_l}) \nonumber  \\  & + (\nabla_{\theta}h_{\theta}(x_{k\Delta_l}))^*(Y_{(k+1)\Delta_l}-Y_{k\Delta_l})
-  (\nabla_{\theta} h_{\theta}(x_{k\Delta_l}))^*h_{\theta}(x_{k\Delta_l})\Delta_l
\,\Big\}.
\label{eq:defL}
\end{align}
We remark that $\lambda_{T,\theta}^l$ is a function also of the observations, but this  dependence is suppressed from the notation. 
%
For $k\in\{0,1,\dots,T/\Delta_l-1\}$, we define:
\vspace{-10pt}
\begin{align*}
g_{k,\theta}^l(x_{k\Delta_l}) := \exp\Big\{h_{\theta}(x_{k\Delta_l})^*(y_{(k+1)\Delta_l}-y_{k\Delta_l})-\tfrac{\Delta_l}{2}h_{\theta}(x_{k\Delta_l})^*h_{\theta}(x_{k\Delta_l})\Big\}.
\end{align*}
Note that:
\vspace{-10pt}
\begin{align*}
Z_{T,\theta}^l(x_0,x_{\Delta_l},\dots,x_T) :=& \prod_{k=0}^{T/\Delta_l-1}g_{k,\theta}^l(x_{k\Delta_l}) \\
=& \exp\Big\{\sum_{k=0}^{T/\Delta_l-1}\big[\,h_{\theta}(x_{k\Delta_l})^*(y_{(k+1)\Delta_l}-y_{k\Delta_l})-\tfrac{\Delta_l}{2}h_{\theta}(x_{k\Delta_l})^*h_{\theta}(x_{k\Delta_l})\,\big]\Big\}
\end{align*}
is a time-discretization of $Z_{T,\theta}$. 
We thus obtain  the discretized approximation of the score function  $\nabla_{\theta}\log(\gamma_{T,\theta}(1))$:
\vspace{-10pt}
\begin{align}\label{eq:grad_est_discrete_time}
\nabla_{\theta}\log(\gamma_{T,\theta}^l(1)) := \frac{\overline{\mathbb{E}}_{\theta}\,[\,\lambda_{T,\theta}^l(\widetilde{X}_0,\widetilde{X}_{\Delta_l},\dots,\widetilde{X}_T)
\,Z_{T,\theta}^l(\widetilde{X}_0,\widetilde{X}_{\Delta_l},\dots,\widetilde{X}_T)\,|\,\mathcal{Y}_T\,]}{\overline{\mathbb{E}}_{\theta}\,[\,Z_{T,\theta}^l(\widetilde{X}_0,\widetilde{X}_{\Delta_l},\dots,\widetilde{X}_T)\,|\,\mathcal{Y}_T\,]}.
\end{align}

We have the following result which establishes the convergence of our Euler approximation. Below $\|\cdot\|_2$ is the $L_2-$norm for vectors. The proof is given in \autoref{proof_thm_bias}. 
\begin{theorem}\label{thm:bias}
Assume (D\ref{hyp_diff:1}). Then for any $(r,T)\in[1,\infty)\times\mathbb{N}$ there exists a $C<+\infty$ such that for any $l\in\mathbb{N}_0$
$$
 \mathbb{E}_{\theta}\left[\left\|\nabla_{\theta}\log(\gamma_{T,\theta}(1))  - \nabla_{\theta}\log(\gamma_{T,\theta}^l(1))\right\|_2^r\right]^{1/r} \leq C\Delta_l^{1/2}.
$$
\end{theorem}

%
\noindent
The result is fairly standard, but we note it is not a simple application of results on SDEs in filtering (e.g.~\cite{picard,talay}) and this is reflected in the proof. The 
rate of convergence of the approximation will be very relevant for some of our subsequent discussions.

\subsection{Backward Feynman-Kac Model and Particle Smoothing}
\label{sec:disc_grad_backward}
From herein the $\widetilde{X}$ notation is dropped for simplicity. Consider the time interval $[k,k+1]$ and the $k$-th update of \eqref{eq:theta_update}.
We define the discrete-time approximation (at level $l$) as: 
\vspace{-8pt}
\begin{align*}
u_{k,l}=(x_{k+\Delta_l}, \dots, x_{k+1})\in E_{l}:=(\mathbb{R}^{d_x})^{\Delta_{l}^{-1}}, \quad k\in\mathbb{N}_0.
\end{align*}
Recall (\ref{eq:Z}). A discrete-time approximation of $p_{\theta}(\{Y_t\}_{k\le t\le k+1}|\{X_t\}_{0\le t\le T})$ is:
\vspace{-10pt}
\begin{align*}
G_{k,\theta}^l(u_{k-1,l},u_{k,l}) & :=  \prod_{p=0}^{\Delta_{l}^{-1}-1} g_{k+p,\theta}^l(x_{k+p\Delta_l}),
\end{align*}
where we set $u_{-1,l}=x_{*}$, for each $l\in \mathbb{N}_0$.
We denote by $m_{\theta}^l$ the Euler transition density induced by time-discretisation \eqref{eq:disc_state}, and then write the initial distribution and Markov transition kernel 
for the discrete-time process with $k\in\mathbb{N}$ as follows:
\begin{align*}
\eta_{0,\theta}^l(du_{0,l}) & =  \prod_{p=1}^{\Delta_l^{-1}} m_{\theta}^l(x_{(p-1)\Delta_l},x_{p\Delta_l})dx_{p\Delta_l}; \\
\vspace{-5pt}
M_{\theta}^{l}(u_{k-1,l},du_{k,l}) & = \prod_{p=1}^{\Delta_l^{-1}} m_{\theta}^l(x_{k+(p-1)\Delta_l},x_{k+p\Delta_l})dx_{k+p\Delta_l}.
\end{align*}
\vspace{-12pt}
\begin{rem}
\hspace{-4pt}
The definition of  $G_{k,\theta}^l(u_{k-1,l},u_{k,l})$, $M_{\theta}^{l}(u_{k-1,l},du_{k,l})$ implies: i)  $G_{k,\theta}^l(u_{k-1,l},u_{k,l})$
involves $u_{k-1,l}$ only via its very last element, $x_{k}$;
ii) the dynamics 
of 	$u_{k,l}$ conditionally on $u_{k-1,l}$  depend only on the very last element, $x_{k}$, 
of $u_{k-1,l}$.
\end{rem}
\noindent We can now state the discrete-time filtering distribution for $k\in\mathbb{N}_0$:
\begin{equation}
\label{eq:FKmodel}
\pi_{k,\theta}^l\big(d(u_{0,l},\dots,u_{k,l})\big) :=  \frac{\big(\prod_{p=0}^{k} G_{p,\theta}^l(u_{p-1,l},u_{p,l})\big)\,\eta_{0,\theta}^l(du_{0,l})\prod_{p=1}^k M_{\theta}^{l}(u_{p-1,l},du_{p,l})}
{\int_{E_l^{k+1}}\big(\prod_{p=0}^{k} G_{p,\theta}^l(u_{p-1,l},u_{p,l})\big)\,\eta_{0,\theta}^l(du_{0,l})\prod_{p=1}^k M_{\theta}^{l}(u_{p-1,l},du_{p,l})}.
\end{equation}
That is, $\pi_{k,\theta}^l\big(d(u_{0,l},\dots,u_{k,l})\big)$ is a discrete-time approximation of the filtering distribution:
\begin{equation*}
\pi_{k,\theta}\big(d(\{X_t\}_{0\le t\le k})\big) := \mathbb{P}_{\theta}(d\{X_t\}_{0\le t\le k}|\{Y_t\}_{0\le t\le k}).
\end{equation*}
Expression (\ref{eq:FKmodel}) corresponds to a standard Feynman-Kac model  (see e.g.~\cite{FK}), thus one can approximate  the involved filtering distributions via the corresponding Monte Carlo methodology. 


We develop a Monte Carlo method for the approximation of the discretised score function in (\ref{eq:grad_est_discrete_time}).
This is accomplished  by presenting a backward formula for \eqref{eq:grad_est_discrete_time}. We define for any $p\in\mathbb{N}_0$:
\begin{align*}
f_{\theta}^l(x_{p\Delta_l},x_{(p+1)\Delta_l})  := &  (\nabla_{\theta}b_{\theta}(x_{p\Delta_l}))^*a(x_{p\Delta_l})^{-1}\sigma(x_{p\Delta_l})(W_{(p+1)\Delta_l}-W_{p\Delta_l})  \\[0.2cm]  &\quad + (\nabla_{\theta}h_{\theta}(x_{p\Delta_l}))^*(Y_{(p+1)\Delta_l}-Y_{p\Delta_l})
-  (\nabla_{\theta} h_{\theta}(x_{p\Delta_l}))^*h_{\theta}(x_{p\Delta_l})\Delta_l.
\end{align*}
and let:
\vspace{-18pt}
\begin{align}
\Lambda_{k,\theta}^l(u_{k-1,l},u_{k,l}) &:= \sum_{p=0}^{\Delta_l^{-1}-1}f_{\theta}^l(x_{k+p\Delta_l},x_{k+(p+1)\Delta_l});
\nonumber \\
F_{T,\theta}^l(u_{0,l},\dots,u_{T-1,l}) &:= \sum_{k=0}^{T-1}\Lambda_{k,\theta}^l(u_{k-1,l},u_{k,l})\,\,\Big( \equiv \lambda_{T,\theta}^l(x_0,x_{\Delta_l},\dots,x_T)  \Big), \label{eq:defF}
\end{align}
for $\lambda_{T,\theta}^l(x_0,x_{\Delta_l},\dots,x_T)$ as defined in (\ref{eq:defL}) 
and used in the score function approximation (\ref{eq:grad_est_discrete_time}).
Thus:
\begin{equation}
\label{eq:backward_grad}
\nabla_{\theta}\log(\gamma_{T,\theta}^l(1)) = \int_{E_l^{T}}
F_{T,\theta}^l(u_{0,l},\dots,u_{T-1,l})\, \mathbb{Q}_{T-1,\theta}^l\big(d(u_{0,l},\dots,u_{T-1,l}) \big),
\end{equation}
where $\mathbb{Q}_{T-1,\theta}^l\big(d(u_{0,l},\dots,u_{T-1,l}) \big)$ 
is a time-discretisation of the smoothing law:
\begin{equation*}
\mathbb{Q}_{T-1,\theta}(d\{X_t\}_{0\le t\le T}):= 
\mathbb{P}_{\theta}(d\{X_t\}_{0\le t\le T}|\{Y_t\}_{0\le t\le T}).
\end{equation*} 
Now, by the time-reversal formula for hidden Markov models (see e.g.~\cite{backward,forward_smoothing}) one has:
%
%
%
\begin{align*}
\mathbb{Q}_{T-1,\theta}^l\big(d(u_{0,l},\dots,u_{T-1,l}) \big) := \pi_{T-1,\theta}^l(du_{T-1,l})
\prod_{k=1}^{T-1}  B_{k,\theta,\pi_{k-1,\theta}^l}^l(u_{k,l},du_{k-1,l}),
\end{align*}
for the backward Markov kernel:
\begin{align}
\label{eq:BACK}
B_{k,\theta,\pi_{k-1,\theta}^l}^l(u_{k,l}, &du_{k-1,l}) :=  
\frac{\pi_{k-1,\theta}^l(du_{k-1,l})\, 
G_{k,\theta}^l(u_{k-1,l},u_{k,l}) m_{\theta}^l(u_{k-1,l},u_{k,l})} {\pi_{k-1,\theta}^l(
G_{k,\theta}^l(\cdot,u_{k,l}) m_{\theta}^l(\cdot,u_{k,l}))},
\end{align}
under the standard notation:
$$\pi_{k-1,\theta}^l(
G_{k,\theta}^l(\cdot,u_{k,l}) m_{\theta}^l(\cdot,u_{k,l})) = \int_{E_l}\pi_{k-1,\theta}^l(du_{k-1,l}) 
G_{k,\theta}^l(u_{k-1,l},u_{k,l}) m_{\theta}^l(u_{k-1,l},u_{k,l}).
$$
%
%
\begin{rem}
\label{rem:cancel}	
 The model structure gives important cancellations in (\ref{eq:BACK}), so that:
$$
B_{k,\theta,\pi_{k-1,\theta}^l}^l(u_{k,l}, du_{k-1,l}) \equiv \frac{\pi_{k-1,\theta}^l(du_{k-1,l}) g_{k,\theta}^l(x_k) m_{\theta}^l(x_k,x_{k+\Delta_l})}{\int_{E_l}\pi_{k-1,\theta}^l(du_{k-1,l}) g_{k,\theta}^l(x_k) m_{\theta}^l(x_k,x_{k+\Delta_l})}.
$$
\end{rem}
The objective now is to approximate the  right hand side of \eqref{eq:backward_grad} using particle approximations. Our online particle approximation of the gradient of the log-likelihood in \eqref{eq:backward_grad}, for a given $l\in\mathbb{N}_0$ is presented in \autoref{alg:online_disc}. Our estimates are given in \eqref{eq:grad_est}
and \eqref{eq:grad_est_1} in \autoref{alg:online_disc}. The approach is the method introduced in \cite{backward,forward_smoothing}.


\begin{algorithm}[h]
\caption{Online Score Function Estimation for a given $l\in\mathbb{N}_0$.} \label{alg:online_disc}
\begin{enumerate}
\item{For $i\in\{1,\dots,N\}$, sample $u_{0,l}^i$ i.i.d.~from $\eta_{0,\theta}^l(\cdot)$. The estimate of
$\nabla_{\theta}\log(\gamma_{1,\theta}^l(1))$
 is:
\begin{equation}
\label{eq:grad_est}
\widehat{\nabla_{\theta}\log(\gamma_{1,\theta}^l(1))} := \frac{\sum_{i=1}^N G_{0,\theta}^l(x_{*},u_{0,l}^i)\Lambda_{0,\theta}^l(x_{*},u_{0,l}^i)}{\sum_{i=1}^N G_{0,\theta}^l(x_{*},u_{0,l}^i)}.
\end{equation}
Set $k=1$, and for $i\in\{1,\dots,N\}$, $\check{u}_{-1,l}^i=x_{*}$.}
\item{For $i\in\{1,\dots,N\}$ sample $\check{u}_{k-1,l}^i$ from:
$$
\sum_{i=1}^N\frac{G_{k-1,\theta}^l(\check{u}_{k-2,l}^i,u_{k-1,l}^i)}{\sum_{j=1}^NG_{k-1,\theta}^l(\check{u}_{k-2,l}^j,u_{k-1,l}^j)}\delta_{\{u_{k-1,l}^i\}}(\cdot).
$$
If $k=1$, for $i\in\{1,\dots,N\}$, set $F_{k-1,\theta}^{l,N}(\check{u}_{0,l}^i)=\Lambda_{0,\theta}^l(x_{*},\check{u}_{0,l}^i)$.}
\item{For $i\in\{1,\dots,N\}$, sample $u_{k,l}^i$ from $M_{\theta}^l(\check{u}_{k-1,l}^i,\cdot)$. For $i\in\{1,\dots,N\}$, compute:
\begin{equation}
\label{eq:f_rec_def}
F_{k,\theta}^{l,N}(u_{k,l}^i) = 
\frac{\sum_{j=1}^N 
g_{k,\theta}^l(\check{x}_k^j)m_{\theta}^l(\check{x}_k^j,x_{k+\Delta_l}^i)\{F_{k-1,\theta}^{l,N}(\check{u}_{k-1,l}^j)
+ \Lambda_{k,\theta}^l(\check{u}_{k-1,l}^j,u_{k,l}^i)
\}}
{\sum_{j=1}^N 
g_{k,\theta}^l(\check{x}_k^j)m_{\theta}^l(\check{x}_k^j,x_{k+\Delta_l}^i)
}.
\end{equation}
The estimate of
$\nabla_{\theta}\log(\gamma_{k+1,\theta}^l(1))$
 is:
\begin{equation}
\label{eq:grad_est_1}
\widehat{\nabla_{\theta}\log(\gamma_{k+1,\theta}^l(1))} := 
\frac{\sum_{i=1}^N G_{k,\theta}^l(\check{u}_{k-1,l}^i,u_{k,l}^i)F_{k,\theta}^{l,N}(u_{k,l}^i)}{\sum_{i=1}^N G_{k,\theta}^l(\check{u}_{k-1,l}^i,u_{k,l}^i)}.
\end{equation}
Set $k=k+1$ and return to the start of 2..}
\end{enumerate}
\end{algorithm}

\subsection{Discussion of \autoref{alg:online_disc}}\label{sec:first_alg_disc}

There are several remarks  worth making, before proceeding. Firstly, the cost of this algorithm
 per unit time is $\mathcal{O}(N\Delta_l^{-1}+N^2)$. In detail, 
 the cost of the particle filter is $\mathcal{O}(N\Delta_l^{-1})$.
 The  cost of calculating 
 $F_{k,\theta}^{l,N}(u_{k,l}^i)$, $1\le i \le N$, in (\ref{eq:f_rec_def}), 
 is $\mathcal{O}(N^2)$; $\Delta_l$ is not involved here due to cancellations -- see \autoref{rem:cancel}. 
 The cost of \eqref{eq:grad_est_1} is $\mathcal{O}(N)$, 
 given the particle filter has already been executed.
  %
%
 There are several implications of this remark. Based upon the results in \cite{backward} and \autoref{thm:bias}, in a sequel work we prove, under appropriate assumptions, we will have the following MSE, for $(k,N)\in\mathbb{N}^2$:
\begin{equation}\label{eq:mse_bound}
\mathbb{E}_{\theta}\Big[\,\Big\|\,\widehat{\nabla_{\theta}\log(\gamma_{k,\theta}^l(1))}-
\nabla_{\theta}\log(\gamma_{k,\theta}(1))\, 
\Big\|_2^2\,\Big] \leq C\Big(\frac{1}{N}+\Delta_l\Big),
\end{equation}
where $C$ is a constant that does not depend on $N$ or $l$. To achieve an MSE of $\mathcal{O}(\epsilon^2)$ for some
$\epsilon>0$ given, one sets $l$ so that $\Delta_l=\epsilon^2$ (i.e.~$l=\mathcal{O}(|\log(\epsilon)|)$) and $N=\mathcal{O}(\epsilon^{-2})$. The cost per unit time of doing this is then $\mathcal{O}(\epsilon^{-4})$. 
We note that to choose $l$ as specified, one has to have access to an appropriately finely observed data set and this is assumed throughout.
Typically, one could use a multilevel Monte Carlo method, as in \cite{high_freq_ml}, to reduce the cost to achieve an MSE of $\mathcal{O}(\epsilon^2)$. However,
in this case as the $\mathcal{O}(N^2)$ cost dominates and does not depend on $l$, one can easily check that such a variance reduction method \emph{will not improve} our particle method. 
To understand this, one can prove a bound on the MSE, for instance of the type conjectured later in this article \eqref{eq:ml_var_thm}, and then try to minimize the cost, by selecting the appropriate number of samples on each level 
to obtain a given MSE. This latter problem leads to a constrained minimization problem that can be solved using Lagrange multipliers (as in e.g.~\cite{lan_mu}), but one can show that this yields that the order of the cost
to achieve an MSE of $\mathcal{O}(\epsilon^2)$ is still $\mathcal{O}(\epsilon^{-4})$.

Secondly, it is important to note that the method in \cite{paris} can reduce the cost of online score estimation to $\mathcal{O}(N\Delta_l^{-1})$ per unit time. However, in order to do so, one requires that $m_{\theta}^l(x,x')$ is uniformly lower-bounded in $(x,x')$, which does not typically occur for Euler-discretized diffusion densities. As a result, we only use the approach shown in \autoref{alg:online_disc}. Note that \cite{glo}, in a different but related context, considers using unbiased and non-negative estimates of the transition density in the approach in \cite{paris}, but such estimates are not always available.

Thirdly and rather importantly, there is a potential issue related to the construction of the algorithm. 
We have started with a continuous-time formula, discretized it and applied what are essentially discrete-time methods for smoothing of additive functionals. 
A serious caveat is that the algorithm is not well-defined as $l\rightarrow\infty$, which is what we mean by saying it has \emph{no (Wiener) path-space formulation}.
The source of the issue is related to using approximations of the transition density $m_{\theta}^l(\check{x}_k^j,x_{k+\Delta_l}^i)$ 
in \eqref{eq:f_rec_def}, which can degenerate when $l$ is high. This will result in increasing Monte Carlo variance and computational cost 
and may mean that $C$ in \eqref{eq:mse_bound} explodes exponentially in $l$. We refer the reader to \cite[Figure 1.1]{shouto} for a numerical example.
This issue has manifested itself in MCMC schemes for inferring fully observed SDEs (see e.g.~\cite{omiros_DA}), but 
in the context of particle smoothing and \autoref{alg:online_disc} there are additional considerations.
The resampling operation introduces \emph{discontinuities}. Often such terminology refers to the lack of continuity of the transition density: 
$$\theta \rightarrow \widehat{\nabla_{\theta}\log(\gamma_{k,\theta}^l(1))},$$ 
but here we are interested  in the behavior of $m_{\theta}^l(\check{x}_k^j,x_{k+\Delta_l}^i)$ when we combine points $\check{x}_k^j$ and $x_{k+\Delta_l}^i$ 
that are intrinsically not obtained in a continuous way as $\Delta_l$ diminishes.
This issue has not received attention in earlier numerical studies, but remains a concern. 
As a result, we now consider defining an algorithm that is robust to the size of the time-discretization mesh and hence has a path-space formulation.

\section{Path-Space Feynman-Kac Formulation}\label{sec:method2}

\subsection{Data Augmentation using Bridges}


We begin this section with a review of the method in \cite{schauer,vd_meulen_guided_mcmc}. For simplicity we consider the case $t\in[0,1]$ and let
$\mathbf{X}:=\{X_{t}\}_{t\in[0,1]}$, and $\mathbf{W}:=\{W_{t}\}_{t\in[0,1]}$. Let also $p_\theta(x,t;x',1)$ denote the unknown transition density from time $t$ to $1$ associated to \eqref{eq:state} 
and let also $p_\theta(x,x'):=p_\theta(x,0;x',1)$.
Suppose one could sample from $p_\theta$ to obtain $(x, x')\in \mathbb{R}^{2d_x}$. Then one can interpolate these points by using a bridge process which has a drift given by $b_{\theta}(x)+a(x)\nabla_x\log{p}_\theta(x,t;x',1)$, as we will explain below. Let $\overline{\mathbb{P}}_{\theta,x,x'}$ denote the law of the solution of the SDE \eqref{eq:state}, on $[0,1]$, started at $x_*=x$ and conditioned to hit $x'$ at time $1$.

As $p_\theta$ is intractable in general, we consider a user-specified auxiliary process $\{\tilde{X}_t\}_{t\in[0,1]}$ following:
\begin{align}
\label{eq:aux_SDE_tilde}
d \tilde X_{t} = \tilde b_{\theta}(t,\tilde X_{t})dt + \tilde \sigma(t,\tilde X_{t})dW_{t}, \quad t\in[0,1],\quad~\tilde{X}_0 =x, 
\end{align}
where for each parameter value $\theta\in\Theta$, $\tilde b_{\theta}:[0,1]\times\mathbb{R}^{d_x}\rightarrow\mathbb{R}^{d_x}$ and $\tilde \sigma:\mathbb{R}^{d_x}\rightarrow\mathbb{R}^{d_x\times d_x}$ is
such that $\tilde a (1,x'):= \tilde \sigma(1,x') \tilde \sigma(1,x')^* \equiv a(x')$. Most importantly, (\ref{eq:aux_SDE_tilde}) is chosen so that its transition density 
is available. 
To avoid confusion -- as the specification of process (\ref{eq:aux_SDE_tilde}) can involve parameter $\theta$ and a given ending position $x'$ -- we note that the transition density of  (\ref{eq:aux_SDE_tilde}) from time $t$ to time~$1$ corresponds to a mapping  $z\to\tilde{p}_{\theta,x'}(x, t; z, 1)$. We also use the notation $\tilde{p}_{\theta,x'}(x, z):=\tilde{p}_{\theta,x'}(x, 0; z, 1)$.
One possible choice is to use an Ornstein-Uhlenbeck process (e.g.~obtained using linearizations or variational inference with \eqref{eq:state} \cite[Section 1.3]{schauer}); see also \cite[Section 2.2]{schauer} for technical conditions on $\tilde b_{\theta}, \tilde a, \tilde p_{\theta,x'}$. 
The main purpose of $\{\tilde X_t\}_{t\in[0,1]}$ is to construct another process $\{X_t^\circ\}_{t\in[0,1]}$ conditioned to hit a given $x'$ at $t=1$. The latter will form an importance proposal for $\{X_t\}_{t\in[0,1]}$. Let:
\begin{align}
\label{eq:aux_SDE}
d X^\circ_{t} = b_{\theta}^{\circ}(t,X^{\circ}_{t}; x')dt + \sigma(X^{\circ}_{t})dW_{t}, \quad t\in[0,1],\quad~X^{\circ}_0 =x, 
\end{align}
where:
$$b_{\theta}^{\circ}(t,x;x')=b_{\theta}(x)+a(x)\nabla_x\log\tilde{p}_{\theta,x'}(x,t;x',1),$$ 
and denote by
 $\mathbb{P}^\circ_{\theta,x,x'}$ the probability law of the solution of (\ref{eq:aux_SDE}). 
The SDE in (\ref{eq:aux_SDE}) gives rise to a function: 
\begin{equation}
\mathbf{W}\rightarrow C_{\theta}(x,\mathbf{W},x'),
\label{eq:map}
\end{equation} 
mapping the driving Wiener noise $\mathbf{W}$ to the solution of (\ref{eq:aux_SDE}), so we have effectively reparameterized the problem from $\mathbf{X}$ to $(\mathbf{W},x')$.

Now, following \cite{schauer}, the two measures $\overline{\mathbb{P}}_{\theta,x,x'}$ and $\mathbb{P}^\circ_{\theta,x,x'}$ are absolutely continuous w.r.t.~each other, with Radon-Nikodym derivative: 
\begin{equation}
\frac{d\overline{\mathbb{P}}_{\theta,x,x'} }{d \mathbb{P}^\circ_{\theta,x,x'} }(\mathbf{X})=
\exp\Big\{  \int_{0}^{1} L_{\theta}(t,X_t)dt\Big\} \times \frac{ \tilde{p}_{\theta,x'}(x,x')}{{p}_{\theta}(x,x')},
\label{eq:density}
\end{equation}
where: 
\begin{align*}
L_{\theta}&(t,x):=\big(b_\theta(x)- \tilde{b}_\theta(t,x)\big)^*\, \nabla_x\log\tilde{p}_{\theta,x'}(x,t;x',1)\\
&\!\!\!\!\!\!\!-\tfrac{1}{2}\textrm{Tr}\,\Big\{\,\big[ a(x)-\tilde{a}(t,x)\big] \big[ -\nabla_x^2\log\tilde{p}_{\theta,x'}(x,t;x',1)-\nabla_x\log\tilde{p}_{\theta,x'}(x,t;x',1)\nabla_x\log\tilde{p}_{\theta,x'}(x,t;x',1)^* \big]\,\Big\} 
\end{align*}
with $\textrm{Tr}(\cdot)$ denoting the trace of a squared matrix.
Note that, in the case when $\sigma=\sigma(x)$ is not a constant function,  then, typically, $x'\rightarrow \tilde{p}_{\theta,x'}(x,x')$ will not integrate to $1$ and will give rise to a non-trivial distribution to sample from.
As the complete algorithm will require being able to sample from the transition density, we rewrite:
\begin{equation}
\frac{d\overline{\mathbb{P}}_{\theta,x,x'} }{d\widetilde{\mathbb{P}}_{\theta,x,x'} }(\mathbf{X})=
\exp\Big\{  \int_{0}^{1} L_{\theta}(t,X_t)dt\Big\} \times \frac{\tilde{p}_{\theta,x'}(x,x')}{{p}_{\theta}(x,x')\hat{{p}}_{\theta}(x,x')}\times 
\hat{{p}}_{\theta}(x,x'),
\label{eq:density2}
\end{equation}
where an arbitrary, tractable and easy to sample density $\hat{{p}}_{\theta}(x,x')$ is used to sample $x'$.

\subsection{Estimation of Score in Continuous-Time}
We return to the expression of the score function in  \eqref{eq:gll} and use the alternative change of measure described above. Consider the processes:
\begin{align*}
\mathbf{X}_{k}:=\{X_{t}\}_{t\in[k,k+1]}, \quad \mathbf{Y}_k:=\{Y_{t}\}_{t\in[k,k+1]},
\quad k\in\mathbb{N}_0.
\end{align*}
We introduce  the following notation:
\begin{align*}
\Psi_{\theta}(\mathbf{X}_k) & =  \int_{k}^{k+1} L_{\theta}(t,X_t)dt;\\[0.3cm]
J_{k,\theta}(\mathbf{X}_k,\mathbf{Y}_k) & =  \int_{k}^{k+1} h_{\theta}(X_{t})^{*}dY_{t} - \tfrac{1}{2}
\int_{k}^{k+1}h_{\theta}(X_t)^{*}h_{\theta}(X_t)dt;\\[0.3cm]
\Lambda_{k,\theta}(\mathbf{X}_k,\mathbf{Y}_k) & = 
\int_{k}^{k+1}(\nabla_{\theta}b_{\theta}(X_t))^*a(X_t)^{-1}(dX_t-b_{\theta}(X_t)dt) \\ &\qquad\qquad + \int_{k}^{k+1}(\nabla_{\theta}h_{\theta}(X_t))^* dY_{t} - \int_{k}^{k+1}(\nabla_{\theta}h_{\theta}(X_t))^*h_{\theta}(X_t)dt;\\[0.3cm]
\Phi_{k,\theta}(\mathbf{X}_k,\mathbf{Y}_k) & =  J_{k,\theta}(\mathbf{X}_k,\mathbf{Y}_k) + \Psi_{\theta}(\mathbf{X}_k) + \log\frac{\tilde{p}_{\theta,x_{k+1}}(x_k,x_{k+1})}{\hat{{p}}_{\theta}(x_{k},x_{k+1})}.
\end{align*}
Note all the integrands can be computed point wise.
\begin{rem}
The Wiener process in (\ref{eq:map}) is defined on the time interval $[0,1]$, 
thus so is the transform $C_{\theta}(x,\mathbf{W},x')$. In the derivations below, 
one needs to calculate $J_{k,\theta}(C_{\theta}(x_{k},\mathbf{W}_{k},x_{k+1})
,\mathbf{Y}_k)$ and $\Lambda_{k,\theta}(C_{\theta}(x_{k},\mathbf{W}_k,x_{k+1}),\mathbf{Y}_k)$,
for $\mathbf{W}_k$'s that correspond to samples from the Wiener measure on $[0,1]$.
With some abuse of notation, it is to be understood that the path $C_{\theta}(x_{k},\mathbf{W}_k,x_{k+1})$ is `shifted' from $[0,1]$ to $[k,k+1]$, so all quantities below agree with the notation introduced above. Also, the calculation of $\Psi_{\theta}(C_{\theta}(x_{k},\mathbf{W}_{k},x_{k+1}))$ will be required, but this should create no confusion.
\end{rem}
Under these definitions, for any $T\in\mathbb{N}$ the score function in \eqref{eq:gll} can be rewritten as:
\begin{align*}
\nabla_{\theta}\log(\gamma_{T,\theta}(1))  \equiv 
\frac{ 
\overline{\mathbb{E}}_{\theta}\Big[
\,\big(\sum_{k=0}^{T-1}\Lambda_{k,\theta}(\mathbf{X}_k,\mathbf{Y}_k)\,\big)
\exp\big\{
\sum_{k=0}^{T-1} J_{k,\theta}(\mathbf{X}_k,\mathbf{Y}_k)
\big\}
\,\big|\,\mathcal{Y}_T\,\Big]
}
{
\overline{\mathbb{E}}_{\theta}\,\Big[\,
\exp\big\{
\sum_{k=0}^{T-1} J_{k,\theta}(\mathbf{X}_k,\mathbf{Y}_k)
\big\}
\,\big|\,\mathcal{Y}_T\,\Big]
}.
\end{align*}
%
Making use of the transform (\ref{eq:map}) and the density expression in \eqref{eq:density}, 
we can equivalently write:
\begin{align}
&\nabla_{\theta}\log(\gamma_{T,\theta}(1)) =
\label{eq:path}
\\[0.4cm] 
&\!\!\!\!\!\!\!\!\!\!= \frac{\widetilde{\mathbb{E}}_{\theta}\,
\Big[\,
\big(\sum_{k=0}^{T-1}\Lambda_{k,\theta}(C_{\theta}(x_{k},\mathbf{W}_{k},x_{k+1})
,\mathbf{Y}_k)\big)
\exp\big\{
\sum_{k=0}^{T-1}\Phi_{k,\theta}(C_{\theta}(x_{k},\mathbf{W}_k,x_{k+1}),\mathbf{Y}_k)
\big\}
\,\big|\,\mathcal{Y}_T\,\Big]}
{
\widetilde{\mathbb{E}}_{\theta}\,\Big[\,
\exp\big\{
\sum_{k=0}^{T-1}\Phi_{k,\theta}(C_{\theta}(x_{k},\mathbf{W}_k,x_{k+1}),\mathbf{Y}_k)
\big\}
\,\big|\,\mathcal{Y}_T\,\Big]
} ,\nonumber
\end{align}
where, the expectation $\widetilde{\mathbb{E}}_{\theta}\,[\,\cdot\,|\,\mathcal{Y}_T\,]$ 
is considered under the probability measure:
\begin{align}
\widetilde{\mathbb{P}}_{\theta}\,\big(\,d(\mathbf{W}_{0}, x_1,\ldots,\mathbf{W}_{T-1}, x_{T})\,\big) := 
\bigotimes_{k=0}^{T-1} \big[\, \mathbb{W}(d\mathbf{W}_{k}) \otimes \hat{p}_{\theta}(x_{k},x_{k+1})dx_{k+1}\,\big], \label{eq:Markov}
\end{align}
independently of $\mathcal{Y}_T$; here, $\mathbb{W}$ is the standard Wiener measure on $[0,1]$ and $x_0=x_{*}$.
\begin{rem}
The approach that has been adopted here can also be used if $\sigma$ depends upon $\theta$. Assuming the formula is
well-defined, one would have a score function with an expression of the type:
$$
\frac{\widetilde{\mathbb{E}}_{\theta}\,
\Big[\,
\big(\sum_{k=0}^{T-1} \Xi_{k,\theta}(C_{\theta}(x_{k},\mathbf{W}_{k},x_{k+1}) \big)
\exp\big\{
\sum_{k=0}^{T-1}\Phi_{k,\theta}(C_{\theta}(x_{k},\mathbf{W}_k,x_{k+1}),\mathbf{Y}_k)
\big\}
\,\big|\,\mathcal{Y}_T\,\Big]}
{
\widetilde{\mathbb{E}}_{\theta}\,\Big[\,
\exp\big\{
\sum_{k=0}^{T-1}\Phi_{k,\theta}(C_{\theta}(x_{k},\mathbf{W}_k,x_{k+1}),\mathbf{Y}_k)
\big\}
\,\big|\,\mathcal{Y}_T\,\Big]
},
$$
where:
$$
\Xi_{k,\theta}(C_{\theta}(x_{k},\mathbf{W}_{k},x_{k+1}) = \{\nabla_{\theta}\Phi_{k,\theta}(C_{\theta}(x_{k},\mathbf{W}_{k},x_{k+1})
,\mathbf{Y}_k)\}+ \nabla_{\theta}\log\{\hat{p}_{\theta}(x_{k},x_{k+1})\},
$$
and with an appropriate modification of the approach to allow $\sigma$ to depend on $\theta$.
To keep consistency with the ideas in Section \ref{sec:method1} we do not consider the formula from herein, but remark that extension of the forthcoming methodology to this case is straightforward.
\end{rem}
\begin{rem}
Recent advances in \cite{bierkens2020} have extended the construct of the auxiliary bridge process -- developed via (\ref{eq:aux_SDE_tilde}), (\ref{eq:aux_SDE}) herein -- to the setting of \emph{hypoelliptic} SDEs. Though we do not pursue this direction here for the purpose of easing the exposition,  we remark that, given these new developments,  one can now, in principle, obtain  score function estimates -- thus, also carry out parameter inference -- in the hypoelliptic regime along the same steps we follow in the current work. 
\end{rem}


The expression in \eqref{eq:path}, together with the (trivially) Markovian dynamics 
for the process in \eqref{eq:Markov} allow one to construct a backward Feynman-Kac type formula
as in \eqref{eq:backward_grad}. To better connect the approach here and that in Section \ref{sec:method1}, define
$u_k=(\mathbf{W}_{k},x_{k+1})$ for $k\in\mathbb{N}_0$, $u_{-1}=x_{*}$ and:
\begin{align*}
\Lambda_{k,\theta}^C(u_{k-1},u_{k}) & =  \Lambda_{k,\theta}(C_{\theta}(x_{k},\mathbf{W}_{k},x_{k+1}),\mathbf{Y}_k);\\
\Phi_{k,\theta}^C(u_{k-1},u_{k}) & =  \Phi_{k,\theta}(C_{\theta}(x_{k},\mathbf{W}_k,x_{k+1}),\mathbf{Y}_k);\\
F_{T,\theta}(u_0,\dots,u_{T-1})   &=  \sum_{k=0}^{T-1}\Lambda_{k,\theta}^C(u_{k-1},u_{k}).
\end{align*}
Superscript $C$ is motivated by the well-posedness of the formula in \emph{continuous-time} path-space.
Set:
$$
\pi_{k,\theta}\big(\,d(u_0,\dots,u_{k})\,\big) := 
\frac{\Big(\prod_{p=0}^k \exp\{\Phi_{p,\theta}^C(u_{p-1},u_{p})\}\Big)
\widetilde{\mathbb{P}}_{\theta}\,\big(\,d(u_{0},\ldots,u_k)\,\big)}{
\int_{E^{k+1}}
\Big(\prod_{p=0}^k \exp\{\Phi_{p,\theta}^C(u_{p-1},u_{p})\}\Big)
\widetilde{\mathbb{P}}_{\theta}\,\big(\,d(u_{0},\ldots,u_k)\,\big)
},
$$
with $E=C([0,1],\mathbb{R}^{d_x})$.
Then, we have the representation:
\begin{equation}\label{eq:backward_grad_new}
\nabla_{\theta}\log(\gamma_{T,\theta}(1)) = \int_{E^T} F_{T,\theta}(u_0,\dots,u_{T-1})\, \mathbb{Q}_{T-1,\theta}\big(\,d(u_0,\dots,u_{T-1})\,\big),
\end{equation}
where:
$$
\mathbb{Q}_{T-1,\theta}\big(\,d(u_0,\dots,u_{T-1})\,\big) = \pi_{T-1,\theta}(du_{T-1}) \prod_{k=1}^{T-1} B_{k,\theta,\pi_{k-1,\theta}}(u_k,du_{k-1}),
$$
and:
$$
B_{k,\theta,\pi_{k-1,\theta}}(u_k,du_{k-1}) := \frac{\pi_{k-1,\theta}(du_{k-1}) \exp\{\Phi_{k,\theta}^C(u_{k-1},u_{k})\}\hat{p}_{\theta}(x_{k},x_{k+1})}{\pi_{k-1,\theta}(\exp\{\Phi_{k,\theta}^C(\cdot,u_{k})\}\hat{p}_{\theta}(\cdot,x_{k+1}))
}.
$$
We remark that formula \eqref{eq:backward_grad_new} is a type of backward Feynman-Kac formula in continuous-time, which to our knowledge is new.

As in the case of \autoref{alg:online_disc}, an effective Monte Carlo approximation of such a smoothing expectation \eqref{eq:backward_grad_new} is given in \autoref{alg:Forward-PF-infinite}. The estimates of the score function
are given in equations  \eqref{eq:grad_est_cont}-\eqref{eq:grad_est_cont_1} in \autoref{alg:Forward-PF-infinite}.

\begin{algorithm}[h]
\caption{Online Score Function Estimation on
Path-Space}\label{alg:Forward-PF-infinite} 
\begin{enumerate}
\item{
For $i\in\{1,\dots,N\}$, sample  $u_0^i$ i.i.d.~from $\mathbb{W}(\cdot)\otimes \hat{p}_{\theta}(x_*,\cdot)$. The estimate of
$\nabla_{\theta}\log(\gamma_{1,\theta}(1))$
 is:
\begin{equation}\label{eq:grad_est_cont}
\widehat{\nabla_{\theta}\log(\gamma_{1,\theta}(1))} := \frac{\sum_{i=1}^N \exp\{\Phi_{0,\theta}^C(x_{*},u_{0}^i)\}\Lambda_{0,\theta}^C(x_{*},u_{0}^i)}{\sum_{i=1}^N \exp\{\Phi_{0,\theta}^C(x_{*},u_{0}^i)\}}.
\end{equation}
Set $k=1$ and for $i\in\{1,\dots,N\}$, $\check{u}_{-1}^i=x_{*}$.}
\item{For $i\in\{1,\dots,N\}$ sample $\check{u}_{k-1}^i$ from: 
$$
\sum_{i=1}^N\frac{\exp\{\Phi_{k-1,\theta}^C(\check{u}_{k-2}^i,u_{k-1}^i)\}}{\sum_{j=1}^N
\exp\{\Phi_{k-1,\theta}^C(\check{u}_{k-2}^j,u_{k-1}^j)\}}\delta_{\{u_{k-1}^i\}}(\cdot).
$$
If $k=1$, for $i\in\{1,\dots,N\}$, set $\widetilde{F}_{k-1,\theta}^{N}(\check{u}_{0}^i)=\Lambda_{0,\theta}^C(x_{*},\check{u}_{0}^i)$.
}
\item{For $i\in\{1,\dots,N\}$, sample $u_{k}^i$ from $\mathbb{W}(\cdot)\otimes\hat{p}_{\theta}^l(\check{x}_{k}^i,\cdot)$. For $i\in\{1,\dots,N\}$, compute: 
$$
\widetilde{F}_{k,\theta}^{N}(u_{k}^i) = 
\frac{\sum_{j=1}^N \exp\{\Phi_{k,\theta}^C(\check{u}_{k-1}^j,u_{k}^i)\}
\hat{p}_{\theta}(\check{x}_k^j,x_{k+1}^i)\{\widetilde{F}_{k-1,\theta}^{N}(\check{u}_{k-1}^j)
+ \Lambda_{k,\theta}^C(\check{u}_{k-1}^j,u_{k}^i)
\}}
{\sum_{j=1}^N 
\exp\{\Phi_{k,\theta}^C(\check{u}_{k-1}^j,u_{k}^i)\}
\hat{p}_{\theta}(\check{x}_k^j,x_{k+1}^i)
}.
$$
The estimate of
$\nabla_{\theta}\log(\gamma_{k+1,\theta}(1))$
is:
\begin{equation}\label{eq:grad_est_cont_1}
\widehat{\nabla_{\theta}\log(\gamma_{k+1,\theta}(1))} := 
\frac{\sum_{i=1}^N 
\exp\{\Phi_{k,\theta}^C(\check{u}_{k-1}^i,u_{k}^i)\}\widetilde{F}_{k,\theta}^{N}(u_{k}^i)}{\sum_{i=1}^N 
\exp\{\Phi_{k,\theta}^C(\check{u}_{k-1}^i,u_{k}^i)\}}.
\end{equation}
Set $k=k+1$ and return to the start of 2.}
\end{enumerate}
\end{algorithm}

\subsection{Time-Discretization}

Whilst conceptually important, path-space valued \autoref{alg:Forward-PF-infinite} can seldom be implemented directly in practice; unbiased methods e.g.~\cite{beskos2} may be possible,
but would be cumbersome.
We develop a time-discretization procedure, in a similar manner to that considered in Section \ref{sec:disc_grad_backward}.

We will discretize on the uniform grid with increment $\Delta_l=2^{-l}$. Let $k\in\mathbb{N}_0$ and define:
$$
u_{k,l} = (z_{k+\Delta_l},z_{k+2\Delta_l},\dots,z_{k+1-\Delta_l},x_{k+1})\in(\mathbb{R}^{d_x})^{\Delta_l^{-1}}=E_l,
$$
where $z_{k+\Delta_l},z_{k+2\Delta_l},\dots,z_{k+1-\Delta_l}$ will represent increments of Brownian motion and $u_{-1,l}=x_{*}$. Define the Markov kernel on $E_l$, for $k\in\mathbb{N}$:
$$
\widetilde{M}_{\theta}^l(u_{k-1,l},du_{k,l}) = \Big(\prod_{s=1}^{\Delta_l^{-1}-1}\phi_l(z_{k+s\Delta_l})dz_{k+s\Delta_l}\Big)
\hat{p}_{\theta}(x_{k},x_{k+1})dx_{k+1},
$$
where $\phi_l(z_{k+s\Delta_l})$ is the density associated to the $\mathcal{N}_{d_x}(0,\Delta_l I_{d_x})$ distribution. 
We denote the density of $\widetilde{M}_{\theta}^{l}$ as $\widetilde{Q}_{\theta}^{l}$.
Set: 
$$
\widetilde{\eta}_{0,\theta}^l(du_{0,l}) = \Big(\prod_{s=1}^{\Delta_l^{-1}-1}\phi_l(z_{s\Delta_l})dz_{s\Delta_l}\Big)
\hat{p}_{\theta}(x_{*},x_{1})dx_{1}.
$$
Now set,
for $(k,s)\in\{0,1,\dots,T-1\}\times\{0,1,\dots,\Delta_l^{-1}-2\}$:
\begin{align}
X_{(s+1)\Delta_l+k} = X_{s\Delta_l+k} + b_{\theta}^{\circ}(s\Delta_l, X_{s\Delta_l+k};x_{k+1})\Delta_l + \sigma(X_{s\Delta_l+k})Z_{(s+1)\Delta_l+k}. \label{eq:Euler_diff_bridge}
\end{align}
Define for $k\in\mathbb{N}_0$:
\begin{align*}
\Psi_{\theta}^l(u_{k-1,l},u_{k,l}) & =  \sum_{p=0}^{\Delta_l^{-1}-1} L_{\theta}(p\Delta_l, x_{k+p\Delta_l})\Delta_l;\\
J_{k,\theta}^l(u_{k-1,l},u_{k,l}) & =  \sum_{p=0}^{\Delta_l^{-1}-1} h_{\theta}(x_{k+p\Delta_l})^{*}(Y_{k+(s+1)\Delta_l}-Y_{k+s\Delta_l}) - \frac{1}{2}\sum_{p=0}^{\Delta_l^{-1}-1}h_{\theta}(x_{k+p\Delta_l})^{*}h_{\theta}(x_{k+p\Delta_l})\Delta_l;\\[0.2cm]
\Phi_{k,\theta}^l(u_{k-1,l},u_{k,l}) & =  J_{k,\theta}^l(u_{k-1,l},u_{k,l})+\Psi_{\theta}^l(u_{k-1,l},u_{k,l})+ \log\frac{\tilde{p}_{\theta,x_{k+1}}(x_k,x_{k+1})}{\hat{{p}}_{\theta}(x_{k},x_{k+1})};\\[0.3cm]
\widetilde{G}_{k,\theta}^l(u_{k-1,l},u_{k,l}) & =  \exp\{\Phi_{k,\theta}^l(u_{k-1,l},u_{k,l})\}.
\end{align*}
Now, for $k\in\mathbb{N}_0$:
\begin{align*}
\widetilde{\Lambda}_{k,\theta}^l(&u_{k-1,l},u_{k,l})  = 
\sum_{p=0}^{\Delta_l^{-1}-1}(\nabla_{\theta}b_{\theta}(x_{k+p\Delta_l}))^*a(x_{k+p\Delta_l})^{-1} (x_{k+(p+1)\Delta_l}-x_{k+p\Delta_l}-b_\theta(x_{k+p\Delta_l}) \Delta_l)  \\[-0.2cm]& 
+\sum_{p=0}^{\Delta_l^{-1}-1}(\nabla_{\theta}h_{\theta}(x_{k+p\Delta_l}))^* (Y_{k+(p+1)\Delta_l}-Y_{k+p\Delta_l}) - \sum_{p=0}^{\Delta_l^{-1}-1}(\nabla_{\theta}h_{\theta}(x_{k+p\Delta_l}))^*h_{\theta}(x_{k+p\Delta_l})\Delta_l.
\end{align*}
Set: 
$$
\widetilde{F}_{T,\theta}^l(u_{0,l},\dots,u_{T-1,l}) = \sum_{k=0}^{T-1}\widetilde{\Lambda}_{k,\theta}^l(u_{k-1,l},u_{k,l}).
$$
Writing expectations w.r.t.~$\widetilde{\eta}_{\theta}^l(du_{0,l})\prod_{k=1}^{T-1}\widetilde{M}_{\theta}^l(u_{k-1,l},du_{k,l})$ as 
$\widetilde{\mathbb{E}}^l_{\theta}[\,\cdot\,|\mathcal{Y}_T]$, 
our discretized approximation of $\nabla_{\theta}\log(\gamma_{T,\theta}(1))$ is:
$$
\nabla_{\theta}\log(\widetilde{\gamma}_{T,\theta}^l(1)) :=  \frac{\widetilde{\mathbb{E}}^l_{\theta}\,\big[\,
\widetilde{F}_{T,\theta}^l(U_{0,l},\dots,U_{T-1,l})
\prod_{k=0}^{T-1} \widetilde{G}_{k,\theta}^l(U_{k-1,l},U_{k,l})
\,\big|\,\mathcal{Y}_T\,\big]}
{\widetilde{\mathbb{E}}^l_{\theta}\,\big[\,
\prod_{k=0}^{T-1} \widetilde{G}_{k,\theta}^l(U_{k-1,l},U_{k,l})
\,\big|\,\mathcal{Y}_T\,\big]}.
$$
We note that, whilst terms $\nabla_{\theta}\log(\widetilde{\gamma}_{T,\theta}^l(1))$, $\nabla_{\theta}\log(\gamma_{T,\theta}^l(1))$ should both converge to $\nabla_{\theta}\log(\gamma_{T,\theta}(1))$, as $l\rightarrow\infty$, they will in general be different for any fixed $l$. 

One can also easily develop a discretized time reversal formula such as \eqref{eq:backward_grad} which will converge precisely to \eqref{eq:backward_grad_new} as $l\rightarrow\infty$. 
Define, for $k\in\mathbb{N}_0$:
$$
\widetilde{\pi}_{k,\theta}^l\big(d(u_{0,l},\dots,u_{k,l})\big) :=  \frac{\big(\prod_{p=0}^{k} \widetilde{G}_{p,\theta}^l(u_{p-1,l},u_{p,l})\big)\,\widetilde{\eta}_{0,\theta}^l(du_{0,l})\prod_{p=1}^k \widetilde{M}_{\theta}^{l}(u_{p-1,l},du_{p,l})}
{\int_{E_l^{k+1}}\Big(\prod_{p=0}^{k} \widetilde{G}_{p,\theta}^l(u_{p-1,l},u_{p,l})\Big)\widetilde{\eta}_{0,\theta}^l(du_{0,l})\prod_{p=1}^k \widetilde{M}_{\theta}^{l}(u_{p-1,l},du_{p,l})}.
$$
Then, we have that:
\begin{equation}
\label{eq:backward_grad_new_disc}
\nabla_{\theta}\log(\widetilde{\gamma}_{T,\theta}^l(1))
 = \int_{E_l^{T}}
\widetilde{F}_{T,\theta}^l(u_{0,l},\dots,u_{T-1,l})\, \widetilde{\mathbb{Q}}_{T-1,\theta}^l\big(d(u_{0,l},\dots,u_{T-1,l}) \big),
\end{equation}
where we have defined:
\begin{align*}
\widetilde{\mathbb{Q}}_{T-1,\theta}^l\big(d(u_{0,l},\dots,u_{T-1,l}) \big) := \widetilde{\pi}_{T-1,\theta}^l(du_{T-1,l})
\prod_{k=1}^{T-1}  \widetilde{B}_{k,\theta,\pi_{k-1,\theta}^l}^l(u_{k,l},du_{k-1,l})
\end{align*}
and:
\begin{align*}
\widetilde{B}_{k,\theta,\widetilde{\pi}_{k-1,\theta}^l}^l(u_{k,l}, &du_{k-1,l})  :=  \frac{\widetilde{\pi}_{k-1,\theta}^l(du_{k-1,l})
\widetilde{G}_{k,\theta}^l(u_{k-1,l},u_{k,l}) \widetilde{Q}_{\theta}^l(u_{k-1,l},u_{k,l})} 
{\widetilde{\pi}_{k-1,\theta}^l(
\widetilde{G}_{k,\theta}^l(\cdot,u_{k,l}) \widetilde{Q}_{\theta}^l(\cdot,u_{k,l}))}.
\end{align*}
We remark that, due to the structure of the model:
$$
\widetilde{B}_{k,\theta,\widetilde{\pi}_{k-1,\theta}^l}^l(u_{k,l},du_{k-1,l}) = \frac{\widetilde{\pi}_{k-1,\theta}^l(du_{k-1,l})\widetilde{G}_{k,\theta}^l(u_{k-1,l},u_{k,l})\hat{p}_{\theta}(x_k,x_{k+1})}
{\int_{E_l} \widetilde{\pi}_{k-1,\theta}^l(du_{k-1,l})\widetilde{G}_{k,\theta}^l(u_{k-1,l},u_{k,l})\hat{p}_{\theta}(x_k,x_{k+1})}.
$$
\begin{rem}
It is important to note that -- in contrast to \autoref{rem:cancel} -- there is no cancellation of terms of $\widetilde{G}_{k,\theta}^l(u_{k-1,l},u_{k,l})$ in the numerator and denominator of this backward kernel. This is precisely due to  recursion \eqref{eq:Euler_diff_bridge} which leads to a path-dependence of the future coordinates of the discretized bridge on the terminal position $x_{k+1}$. 
\end{rem}

\subsection{Particle Approximation}

Our online particle approximation of the gradient of the log-likelihood in \eqref{eq:backward_grad_new_disc}, for a 
given $l\in\mathbb{N}_0$ is presented in \autoref{alg:online_disc_new}. Our estimates are given in \eqref{eq:grad_est_cont_disc}
and \eqref{eq:grad_est_cont_disc_1} in \autoref{alg:online_disc_new}.

\autoref{alg:online_disc_new} is simply the time-discretization of the procedure presented in \autoref{alg:Forward-PF-infinite}.
A number of remarks are again of interest. Firstly, the cost of the algorithm per unit time
is now $\mathcal{O}(N^2\Delta_l^{-1})$. The increase in computational cost over \autoref{alg:online_disc} is the fact that
when computing $\widetilde{\Lambda}_{k,\theta}^l(\check{u}_{k-1,l}^j,u_{k,l}^i)$ in \eqref{eq:key_point}, one must solve the recursion \eqref{eq:Euler_diff_bridge} for each $(i,j)\in\{1,\dots,N\}^2$, which has a cost $\mathcal{O}(\Delta_l^{-1})$ and it is this cost that dominates. Secondly, following the discussion in Section \ref{sec:first_alg_disc}, we have proved in a companion work that, under appropriate assumptions, the MSE for $(k,N)\in\mathbb{N}^2$:
\begin{equation}\label{eq:mse_bound_new}
\mathbb{E}_{\theta}\Big[\,\Big\|\,\widehat{\nabla_{\theta}\log(\widetilde{\gamma}_{k,\theta}^l(1))}-
\nabla_{\theta}\log(\gamma_{k,\theta}(1))\,
\Big\|_2^2\,\Big] \leq C\Big(\frac{1}{N}+\Delta_l\Big),
\end{equation}
for constant $C$ that does not depend on $N$, $l$. To achieve an MSE of $\mathcal{O}(\epsilon^2)$ for some
$\epsilon>0$ given, one sets $l=\mathcal{O}(|\log(\epsilon)|)$ and $N=\mathcal{O}(\epsilon^{-2})$. The cost per unit time of doing this, is then $\mathcal{O}(\epsilon^{-6})$. This is significantly worse than the approach in \autoref{alg:online_disc},
but we remark that when discussing the cost of \autoref{alg:online_disc}, in the bound \eqref{eq:mse_bound}, we have assumed that the constant $C$ does not depend upon $l$. However, in a sequel work we will show that under assumptions that this afore-mentioned $C$ explodes exponentially in $l$. Conversely, $C$ in \eqref{eq:mse_bound_new}
can be proved to be independent of $l$, precisely due to the path-space development we have adopted. We remark, however, that one can use an MLMC method to reduce this cost of $\mathcal{O}(\epsilon^{-6})$ per unit time of \autoref{alg:online_disc_new} and this algorithm is presented in the next section.

\begin{algorithm}[h]
\caption{Modified Online Score Function Estimation for a given $l\in\mathbb{N}_0$.}\label{alg:online_disc_new}
\begin{enumerate}
\item{
For $i\in\{1,\dots,N\}$, sample  $u_{0,l}^i$ i.i.d.~from $\widetilde{\eta}_{0,\theta}^l(\cdot)$. The estimate of
$\nabla_{\theta}\log(\widetilde{\gamma}_{1,\theta}^l(1))$ 
is:
\begin{equation}\label{eq:grad_est_cont_disc}
\widehat{\nabla_{\theta}\log(\widetilde{\gamma}_{1,\theta}^l(1))}
:= \frac{\sum_{i=1}^N \widetilde{G}_{0,\theta}^l(x_{*},u_{0,l}^i)\widetilde{\Lambda}_{0,\theta}^l(x_{*},u_{0,l}^i)}{\sum_{i=1}^N 
\widetilde{G}_{0,\theta}^l(x_{*},u_{0,l}^i)}.
\end{equation}
Set $k=1$ and for $i\in\{1,\dots,N\}$, $\check{u}_{-1,l}^i=x_{*}$.}
\item{For $i\in\{1,\dots,N\}$ sample $\check{u}_{k-1,l}^i$ from:
$$
\sum_{i=1}^N\frac{\widetilde{G}_{k-1,\theta}^l(\check{u}_{k-2,l}^i,u_{k-1,l}^i)}{\sum_{j=1}^N
\widetilde{G}_{k-1,\theta}^l(\check{u}_{k-2,l}^j,u_{k-1,l}^j)}\delta_{\{u_{k-1,l}^i\}}(\cdot).
$$
If $k=1$, for $i\in\{1,\dots,N\}$, set $\widetilde{F}_{k-1,\theta}^{l,N}(\check{u}_{0,l}^i)=\widetilde{\Lambda}_{0,\theta}^l(x_{*},\check{u}_{0,l}^i)$.
}
\item{For $i\in\{1,\dots,N\}$, sample $u_{k,l}^i$ from $\widetilde{M}_{\theta}^l(\check{u}_{k-1,l}^i,\cdot)$. For $i\in\{1,\dots,N\}$, compute:
\begin{equation}
\label{eq:key_point}
\widetilde{F}_{k,\theta}^{l,N}(u_{k,l}^i) = 
\frac{\sum_{j=1}^N 
\widetilde{G}_{k,\theta}^l(\check{u}_{k-1,l}^j,u_{k,l}^i)
\hat{p}_{\theta}(\check{x}_k^j,x_{k+1}^i)\{\widetilde{F}_{k-1,\theta}^{l,N}(\check{u}_{k-1,l}^j)
+ \widetilde{\Lambda}_{k,\theta}^l(\check{u}_{k-1,l}^j,u_{k,l}^i)
\}}
{\sum_{j=1}^N 
\widetilde{G}_{k,\theta}^l(\check{u}_{k-1,l}^j,u_{k,l}^i)
\hat{p}_{\theta}(\check{x}_k^j,x_{k+1}^i)
}.
\end{equation}
The estimate of
$\nabla_{\theta}\log(\widetilde{\gamma}_{k+1,\theta}^l(1))$
 is:
\begin{equation}\label{eq:grad_est_cont_disc_1}
\widehat{\nabla_{\theta}\log(\widetilde{\gamma}_{k+1,\theta}^l(1))} = 
\frac{\sum_{i=1}^N 
\widetilde{G}_{k,\theta}^l(\check{u}_{k-1,l}^i,u_{k,l}^i)\widetilde{F}_{k,\theta}^{l,N}(u_{k,l}^i)}{\sum_{i=1}^N 
\widetilde{G}_{k,\theta}^l(\check{u}_{k-1,l}^i,u_{k,l}^i)}.
\end{equation}
Set $k=k+1$ and return to the start of 2..}
\end{enumerate}
\end{algorithm}

\subsection{Multilevel Particle Filter}\label{sec:ml}

We now present a new multilevel particle filter along with online estimation of the score-function. We fix $l\in\mathbb{N}$
for now and for $(k,s)\in\mathbb{N}_0\times\{l,l-1\}$ define:
$$
u_{k,s} = (z_{k+\Delta_l,s},z_{k+2\Delta_l,s},\dots,z_{k+1-\Delta_l,s},x_{k+1,s})\in(\mathbb{R}^{d_x})^{\Delta_s^{-1}}=E_s.
$$

\begin{algorithm}[!ht]
	\caption{Coupled Online Score Function Estimation for a given $l\in\mathbb{N}$.}\label{alg:online_disc_ml}
	\begin{enumerate}
		\item{
			\begin{itemize}
				\item{For $i\in\{1,\dots,N\}$, sample  $u_{0,l}^i$ i.i.d.~from $\widetilde{\eta}_{0,\theta}^l(\cdot)$.}
				\item{For $(i,p)\in\{1,\dots,N\}\times\{1,\dots,\Delta_{l-1}^{-1}-1\}$, set $z_{p\Delta_{l-1},l-1}^i=z_{p\Delta_{l-1},l}^i+z_{p\Delta_{l-1}-\Delta_l,l}^i$ and $x_{1,l-1}^i=x_{1,l}^i$.}
			\end{itemize}
			The estimate of
			$\nabla_{\theta}\log(\widetilde{\gamma}_{1,\theta}^l(1))-\nabla_{\theta}\log(\widetilde{\gamma}_{1,\theta}^{l-1}(1))$
			is:
			$$
			\widehat{\nabla_{\theta}\log(\widetilde{\gamma}_{1,\theta}^l(1))} - 
			\widehat{\nabla_{\theta}\log(\widetilde{\gamma}_{1,\theta}^{l-1}(1))}: \vspace{-0.3cm}
			=
			$$
			\begin{equation}\label{eq:grad_est_cont_disc_diff}
			\frac{\sum_{i=1}^N \widetilde{G}_{0,\theta}^l(x_{*},u_{0,l}^i)\widetilde{\Lambda}_{0,\theta}^l(x_{*},u_{0,l}^i)}{\sum_{i=1}^N 
				\widetilde{G}_{0,\theta}^l(x_{*},u_{0,l}^i)}-
			\frac{\sum_{i=1}^N \widetilde{G}_{0,\theta}^{l-1}(x_{*},u_{0,l-1}^i)\widetilde{\Lambda}_{0,\theta}^{l-1}(x_{*},u_{0,l-1}^i)}{\sum_{i=1}^N 
				\widetilde{G}_{0,\theta}^{l-1}(x_{*},u_{0,l-1}^i)}.
			\end{equation}
			Set $k=1$ and for $i\in\{1,\dots,N\}$, $\check{u}_{-1,l}^i=\check{u}_{-1,l-1}^i=x_{*}$.}
		\item{For $i\in\{1,\dots,N\}$, sample $(\alpha_l(i),\alpha_{l-1}(i))\in\{1,\dots,N\}^2$ from a coupling of:
			$$
			\sum_{i=1}^N\frac{\widetilde{G}_{k-1,\theta}^l(\check{u}_{k-2,l}^{\alpha_l(i)},u_{k-1,l}^{\alpha_l(i)})}{\sum_{j=1}^N
				\widetilde{G}_{k-1,\theta}^l(\check{u}_{k-2,l}^j,u_{k-1,l}^j)} \quad\quad\textrm{and}\quad\quad
			\sum_{i=1}^N\frac{\widetilde{G}_{k-1,\theta}^{l-1}(\check{u}_{k-2,l-1}^{\alpha_{l-1}(i)},u_{k-1,l}^{\alpha_{l-1}(i)})}{\sum_{j=1}^N
				\widetilde{G}_{k-1,\theta}^{l-1}(\check{u}_{k-2,l-1}^j,u_{k-1,l-1}^j)},
			$$
			and set $(\check{u}_{k-1,l}^i,\check{u}_{k-1,l-1}^i)=(u_{k-1,l}^{\alpha_l(i)},u_{k-1,l-1}^{\alpha_{l-1}(i)})$.
			If $k=1$, for $i\in\{1,\dots,N\}$, $s\in\{l,l-1\}$, set $\widetilde{F}_{k-1,\theta}^{s,N}(\check{u}_{0,s}^i)=\widetilde{\Lambda}_{0,\theta}^s(x_{*},\check{u}_{0,s}^i)$.
		}
		\item{
			\begin{itemize}
				\item{
					For $i\in\{1,\dots,N\}$, sample $(x_{k+1,l}^i,x_{k+1,l-1}^i)$ from a coupling of $\hat{p}_{\theta}(\check{x}_{k,l}^i,\cdot)$ and $\hat{p}_{\theta}(\check{x}_{k,l-1}^i,\cdot)$. }
				\item{For $i\in\{1,\dots,N\}$ sample $z_{k+\Delta_l,l}^i,\dots,z_{k+1-\Delta_l,l}^i$ i.i.d.~from $\prod_{s=1}^{\Delta_l^{-1}-1}\phi_l(\,\cdot\,)$.}
				\item{ For $(i,p)\in\{1,\dots,N\}\times\{1,\dots,\Delta_{l-1}^{-1}-1\}$ set $z_{k+p\Delta_{l-1},l-1}^i=z_{k+p\Delta_{l-1},l}^i+z_{k+p\Delta_{l-1}-\Delta_l,l}^i$.}
				\item{
					For $(i,s)\in\{1,\dots,N\}\times\{l.l-1\}$, compute: 
				$$\hspace{-18pt}	
					\widetilde{F}_{k,\theta}^{s,N}(u_{k,l}^i) = 
					\frac{\sum_{j=1}^N 
						\widetilde{G}_{k,\theta}^s(\check{u}_{k-1,s}^j,u_{k,s}^i)
						\hat{p}_{\theta}(\check{x}_{k,s}^j,x_{k+1,s}^i)\{\widetilde{F}_{k-1,\theta}^{s,N}(\check{u}_{k-1,s}^j)
						+ \widetilde{\Lambda}_{k,\theta}^s(\check{u}_{k-1,s}^j,u_{k,s}^i)
						\}}
					{\sum_{j=1}^N 
						\widetilde{G}_{k,\theta}^s(\check{u}_{k-1,s}^j,u_{k,s}^i)
						\hat{p}_{\theta}(\check{x}_{k,s}^j,x_{k+1,s}^i)
					}.
					$$}
			\end{itemize}
			The estimate of
			$\nabla_{\theta}\log(\widetilde{\gamma}_{k+1,\theta}^l(1))-\nabla_{\theta}\log(\widetilde{\gamma}_{k+1,\theta}^{l-1}(1))$
			is:
			$$
			\widehat{\nabla_{\theta}\log(\widetilde{\gamma}_{k+1,\theta}^l(1))} 
			- \widehat{\nabla_{\theta}\log(\widetilde{\gamma}_{k+1,\theta}^{l-1}(1))}
			:= \vspace{-0.3cm}
			$$
			\begin{equation}\label{eq:grad_est_cont_disc_diff1}
			\frac{\sum_{i=1}^N 
				\widetilde{G}_{k,\theta}^l(\check{u}_{k-1,l}^i,u_{k,l}^i)\widetilde{F}_{k,\theta}^{l,N}(u_{k,l}^i)}{\sum_{i=1}^N 
				\widetilde{G}_{k,\theta}^l(\check{u}_{k-1,l}^i,u_{k,l}^i)} - \frac{\sum_{i=1}^N 
				\widetilde{G}_{k,\theta}^{l-1}(\check{u}_{k-1,l-1}^i,u_{k,l-1}^i)\widetilde{F}_{k,\theta}^{l-1,N}(u_{k,{l-1}}^i)}{\sum_{i=1}^N 
				\widetilde{G}_{k,\theta}^{l-1}(\check{u}_{k-1,l-1}^i,u_{k,l-1}^i)}.
			\end{equation}
			Set $k=k+1$ and return to the start of 2..}
	\end{enumerate}
\end{algorithm}

We give the approach in \autoref{alg:online_disc_ml}. Before explaining how one can use \autoref{alg:online_disc_ml} to provide online estimates of the score function, several remarks are required to continue. 
The first is related to the
couplings mentioned in \autoref{alg:online_disc_ml} point 2.~and point 3.~bullet 1. The coupling in point 2, requires
a way to resample the indices of the particles so that they have the correct marginals. This topic has been investigated considerably in the literature, see  e.g.~\cite{cpf_clt,coup_pf}, and techniques that have been adopted include sampling maximal coupling, e.g.~\cite{mlpf}, or using the $\mathbb{L}_2$-Wasserstein optimal coupling \cite{wasser}; in general the latter is found to be better in terms of variance reduction, but can only be implemented when $d_x=1$. We rely upon the maximal coupling in this paper, which has a cost of $\mathcal{O}(N)$ per unit time. For point 3.~bullet 1, one again has a considerable degree of flexibility. In this article we sample the maximal coupling which can be achieved at a cost which is at most $\mathcal{O}(N)$ cost per-unit time using the algorithm of \cite{thor}. The second main remark of interest is that the basic filter that is sampled in 
\autoref{alg:online_disc_ml} is an entirely new coupled particle filter for diffusions (i.e.~different to \cite{mlpf,high_freq_ml}). The utility of the approach relative to \cite{high_freq_ml} is of great interest, in the context of filtering.

Set $(l_{*},L)\in\mathbb{N}^2$ with $l_{*}<L$. The idea is to run \autoref{alg:online_disc_ml}, independently, for $l\in\{l_{*},\dots,L\}$ each with $N_l$ particles and, independently, \autoref{alg:online_disc_new} for $l=l_{*}-1$ with $N_{l_*-1}$ particles. We then consider
the estimate, for $k\in\mathbb{N}$
$$
\widehat{\nabla_{\theta}\log(\widetilde{\gamma}_{k,\theta}^L(1))}_{ML} := \sum_{l=l_{*}}^L \Big\{\widehat{\nabla_{\theta}\log(\widetilde{\gamma}_{k,\theta}^l(1))} 
- \widehat{\nabla_{\theta}\log(\widetilde{\gamma}_{k,\theta}^{l-1}(1))}
\Big\} + \widehat{\nabla_{\theta}\log(\widetilde{\gamma}_{k,\theta}^{l_*-1}(1))},
$$
where the summands on the right hand side are defined in \eqref{eq:grad_est_cont_disc_diff} and \eqref{eq:grad_est_cont_disc_diff1}
and the last term on the right hand side is as either \eqref{eq:grad_est_cont_disc} or \eqref{eq:grad_est_cont_disc_1} (depending on $k$). Now, we show in an on-going companion work, under appropriate assumptions, one has the following result for $(k,N_{l_*-1:L})\in\mathbb{N}^{L-l_*+2}$:
\begin{equation}\label{eq:ml_var_thm}
\mathbb{E}_{\theta}\Big[\,\Big\|\,
\widehat{\nabla_{\theta}\log(\widetilde{\gamma}_{k,\theta}^L(1))}_{ML}
-
\nabla_{\theta}\log(\gamma_{k,\theta}(1))\,
\Big\|_2^2\,\Big] \leq C\Big(\sum_{l=l_*-1}^L\frac{\Delta_l^{\beta}}{N_l}+\Delta_L\Big),
\end{equation}
for constant  $C$  that does not depend on $N$, $l$; also, $\beta=1$ if $\sigma$ is a constant function, else $\beta=1/2$. Choose: i) $L$ so that $\Delta_L=\mathcal{O}( \epsilon^2)$, for $\epsilon>0$ given; ii) if $\beta=1$, $N_l=\mathcal{O}(\epsilon^{-2}\Delta_l^{1/2+\rho})$ for some $0<\rho<1/2$. These selections yield an MSE of $\mathcal{O}(\epsilon^2)$ for a cost of $\mathcal{O}(\epsilon^{-4})$. If $\beta=1/2$, one can set $N_l=\mathcal{O}(\epsilon^{-2}\Delta_l^{1/2+\rho}\Delta_L^{-\rho})$
for some $\rho>0$. This will yield an MSE of $\mathcal{O}(\epsilon^2)$ for a cost of $\mathcal{O}(\epsilon^{-4(1+\rho)})$. Such results are 
at least as good as the method in Section \ref{sec:first_alg_disc}, assuming that latter approach does not collapse with $l$.

We remark that it is possible to produce an almost-surely unbiased estimator of the score function, when $\theta$ is the true parameter, using a combination of the multilevel method that has been developed here and the approach in \cite{ub_filt}. This is left for future work.

\section{Numerical Results}\label{sec:numerics}

In this section, we consider four models to investigate the various properties of our algorithms. The score function is estimated using both Algorithms \ref{alg:online_disc} and \ref{alg:online_disc_new} for a fixed $\theta$. We will show, as expected, that they are equivalent for a large number of particles $N$ and a high level of descritization $l$. We then compare the cost of \autoref{alg:online_disc_new} and its multilevel version \autoref{alg:online_disc_ml}. As an application of our methods, we use Algorithms \ref{alg:online_disc} and \ref{alg:online_disc_ml} for parameter estimation via stochastic gradient. The code is written in MATLAB and it can be downloaded from \url{https://github.com/ruzayqat/score_based_par_est}.

We remark that we will not use \autoref{alg:online_disc_new} for parameter estimation because it is `slow' compared to the algorithms as illustrated in the previous section and in \autoref{fig:cost_old_vs_new}. In \autoref{fig:cost_old_vs_new} we consider the Model 1, as described in the next section, with $T=40$, $l=10$, $\theta = (-0.4, -0.5)$, $\kappa = 2$ and $x_*=0.2$. \autoref{fig:cost_old_vs_new} provides a comparison between the cost of Algorithms \ref{alg:online_disc} and \ref{alg:online_disc_new}, which is the average machine time measured in seconds needed per each simulation, versus the number of particles $N$. As predicted by our theoretical conjectures, we see that  the cost of \autoref{alg:online_disc} is significantly lower than that of \autoref{alg:online_disc_new}.

\begin{figure}[h]
\centering
 \includegraphics[height=0.3\textwidth]{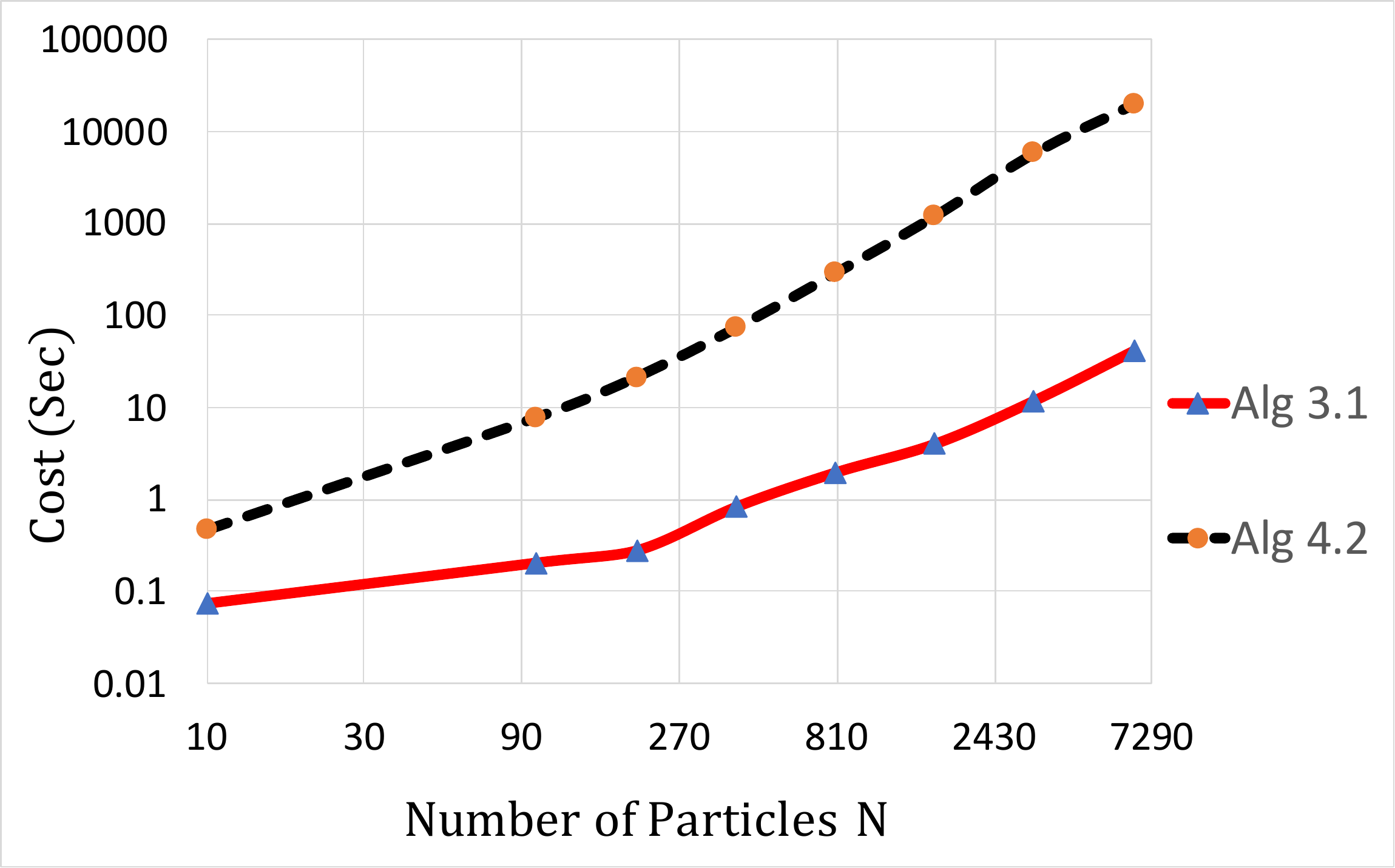} 
 \caption{Comparison between the cost of  Algorithms \ref{alg:online_disc} and \ref{alg:online_disc_new} per each simulation versus the number of particles $N$. We run both algorithms on Model 1.}
 \label{fig:cost_old_vs_new}
\end{figure}

\subsection{Models}

In the following, parameters $(\kappa,\sigma)$ are fixed.

\noindent \textbf{Model 1:} Let $d_x=d_y=1$, $d_\theta=2$ and consider the following linear SDEs:
\begin{align*}
dX_t &= \theta_1 X_t dt + \sigma dW_t;\\
dY_t &= \theta_2 (\kappa - X_t) dt + dB_t.
\end{align*}

\noindent \textbf{Model 2:} Let $d_x=d_y=1$, $d_\theta=3$ and consider a nonlinear diffusion process along with a linear diffusion process of observations:
\begin{align*}
dX_t &= \big( \tfrac{\theta_1}{X_t}+\theta_2 X_t \big) dt + \sigma dW_t;\\
dY_t &= \theta_3 (\kappa - X_t) dt + dB_t.
\end{align*}

\noindent \textbf{Model 3:} Let $d_x=d_y=1$, $d_\theta=3$ and consider a nonlinear signal  along with a nonlinear diffusion process of observations. The first SDE is a Cox-Ingersoll-Ross process after an 1-1 transform. Thus:
\begin{align*}
dX_t &= \frac{1}{2}  \big( \tfrac{\theta_1 \theta_2 - \sigma^2}{X_t}-\theta_2 X_t \big) dt + \sigma dW_t;\\
dY_t &= \theta_3 (\kappa - X_t^2) dt + dB_t.
\end{align*}
This model has a solution if and only if $\theta_1 \theta_2 > 2\sigma^2$.

\noindent \textbf{Model 4:} 
Let $d_x=d_y=1$, $d_\theta=3$ and consider a type of Black-Scholes model with a stochastic volatility:
\begin{align*}
dX_t &= \theta_1(\theta_2 - X_t) dt + \sigma(X_t) dW_t;\\
dY_t &=(\theta_3 - \frac{1}{2} X_t^2) dt + dB_t.
\end{align*}
where $\sigma(X_t)=\beta/\sqrt{X_t^2+1}$ and $\beta$ is fixed. In the hidden process, $\theta_1$ and $\theta_2$ are the speed and level of mean reversion and $\theta_3$ is
a mean type level for the observation process,
We will apply our methodology (see \autoref{fig:par_est_model4} later on) on the log mid-price of Tesla Inc. stock in 2018. The dataset shown in \autoref{fig:Tesla_stock} represents the log mid-price at every second during a trading day for a total of 250 trading days.

\begin{figure}[h]
\centering
 \includegraphics[height=0.4\textwidth]{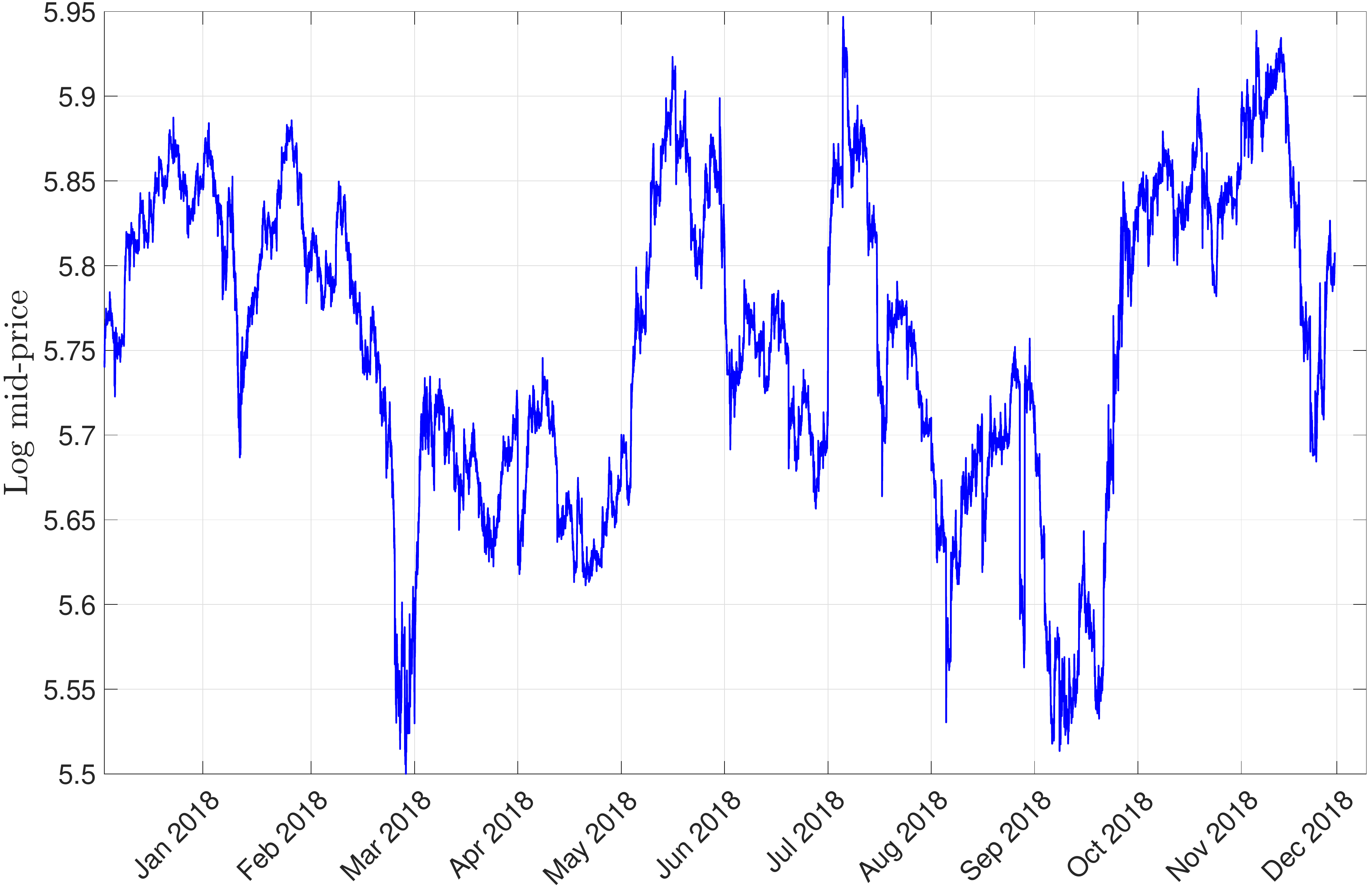} 
 \caption{The log of mid-price for every second of Tesla stock in 2018. The dataset contains $5.85 \times 10^6$ data points.}
 \label{fig:Tesla_stock}
\end{figure}

\subsection{Simulation Results}
In all our results data are generated from the model under the finest discretization considered except in model 4, where we use a real data. In \autoref{alg:online_disc_new}, we consider the auxiliary linear process following:
\begin{align*}
d\tilde{X}_t = \sigma dW_t
\end{align*}
in models 1-3 and in model 4 it follows:
\begin{align*}
d\tilde{X}_t = \sigma(x') dW_t
\end{align*}
 In models 1-3, $\tilde{p}_{\theta,x'}(x,t;x',1)=\mathcal{N}(x';x, (1-t)\sigma^2)$, hence $\tilde{p}_{\theta,x'}(x,x')=\mathcal{N}(x';x,\sigma^2)$, which is easy to sample $x'$ from, and therefore, $\hat{p}_\theta(x,x') = \tilde{p}_{\theta,x'}(x,x')$. But in model 4, $\tilde{p}_{\theta,x'}(x,x')=\mathcal{N}(x';x,\sigma^2(x'))$ which is not easy to sample $x'$ from. Therefore, we take $\hat{p}_\theta(x,x') = \mathcal{N}(x';x,\sigma^2(x))$.

\subsubsection{Estimation of the Score Function}

For each model, we fix parameter $\theta$ and estimate the score function using Algorithms \ref{alg:online_disc} and \ref{alg:online_disc_new}. 
In \autoref{alg:online_disc}, $N\in\{3000, 7000, 4000, 5000\}$ in the 1st, 2nd, 3rd \& 4th models, respectively. In \autoref{alg:online_disc_new}, $N\in\{1000, 2000, 1000, 1500\}$ in the 1st, 2nd, 3rd \& 4th models, respectively. In both algorithms, we set the discretization level to $l= 10$. In Models 1, 2 and 3, we set $\kappa = 2$, $2.2$, $1.5$, $x_*=0.2$, $1$, $2$ and $\sigma = 0.3$, $0.25$, $0.25$, respectively. While in model 4, we set $x_* = 1.3$ and $\beta=2$; $T=50$ for all 4 models ($T=50$ in model 4 corresponds to 14.22 hours of trading). 

\autoref{fig:score_func} summarizes the results of 56 replications of estimates of the score function for each model and for each unit time point. These simulations are implemented in parallel using 8 CPUs. The figure illustrates that both algorithms are equivalent for large $N$ and $l$ as one would expect.

\begin{figure}[h]
\centering
  \subfloat{\includegraphics[height=0.3\textwidth]{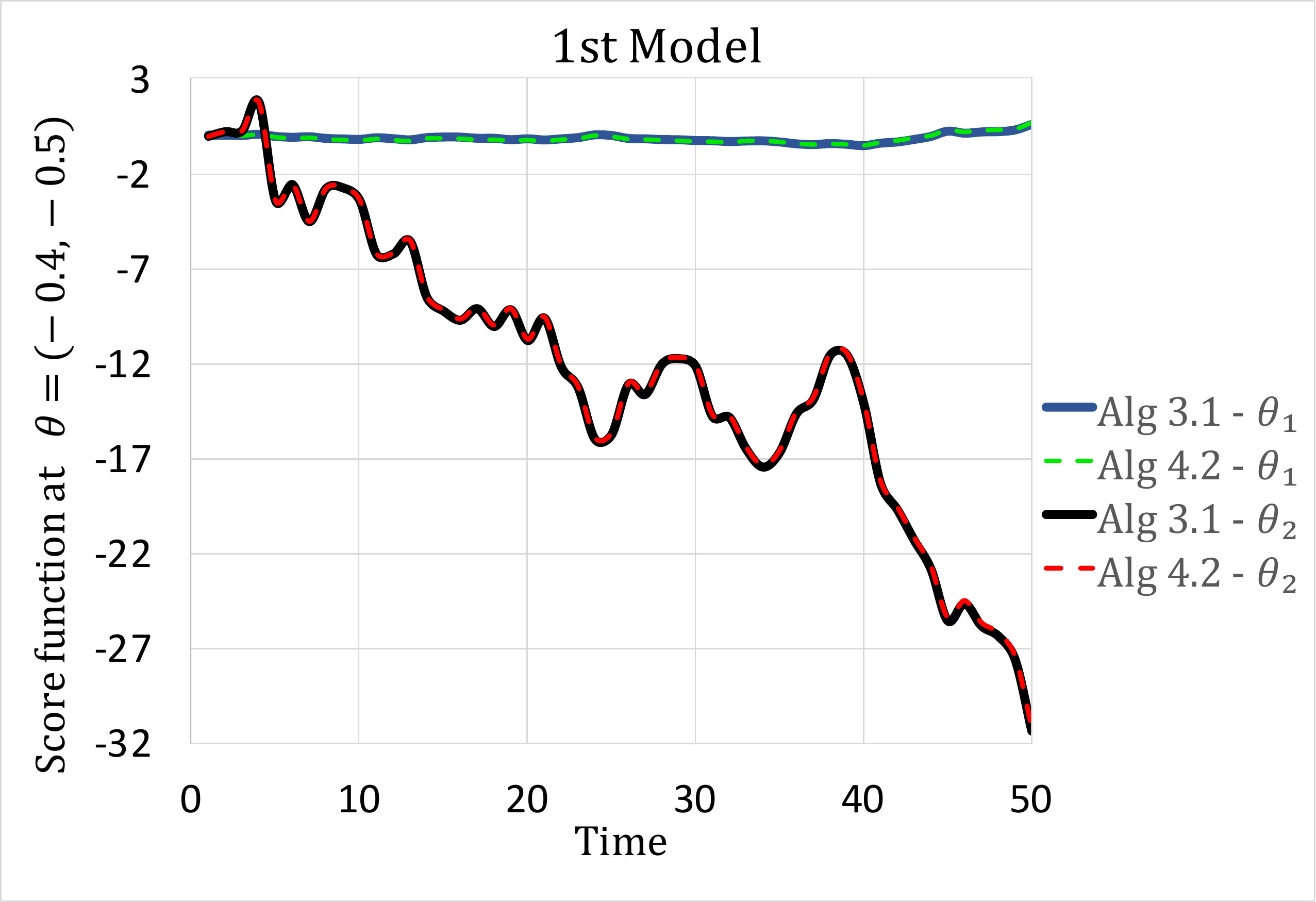}}\,
  \subfloat{\includegraphics[height=0.3\textwidth]{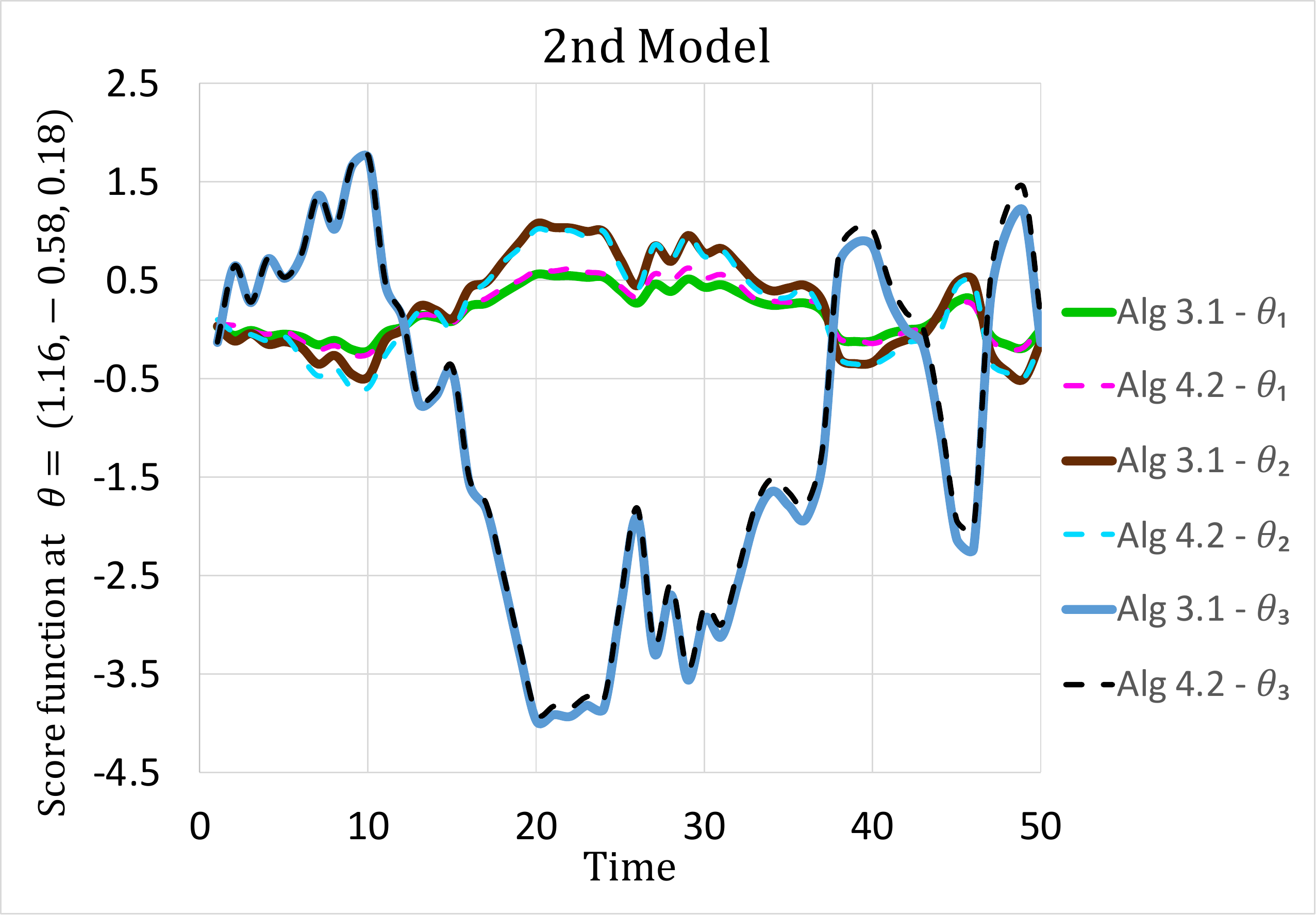}}\\
  \subfloat{\includegraphics[height=0.3\textwidth]{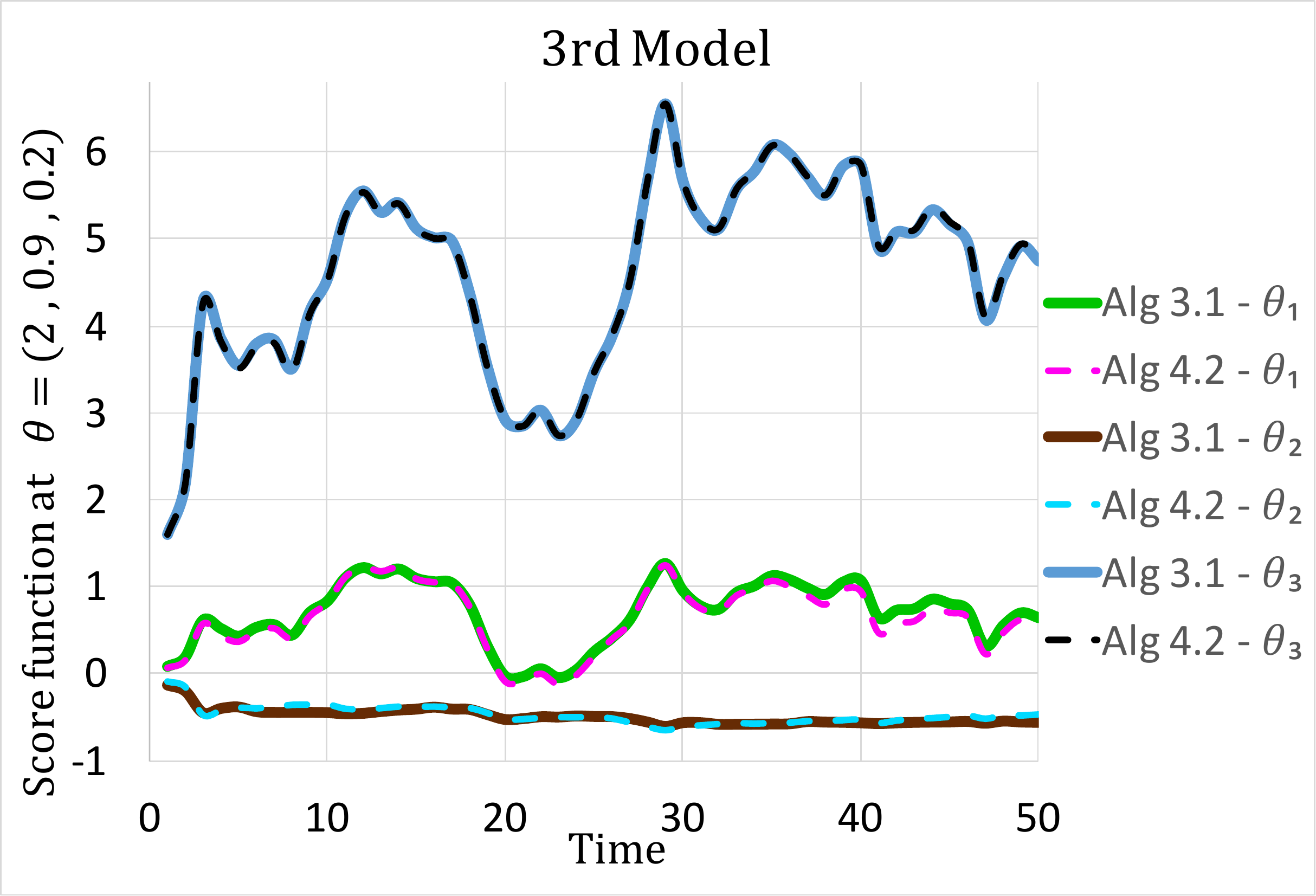}}\,
  \subfloat{\includegraphics[height=0.3\textwidth]{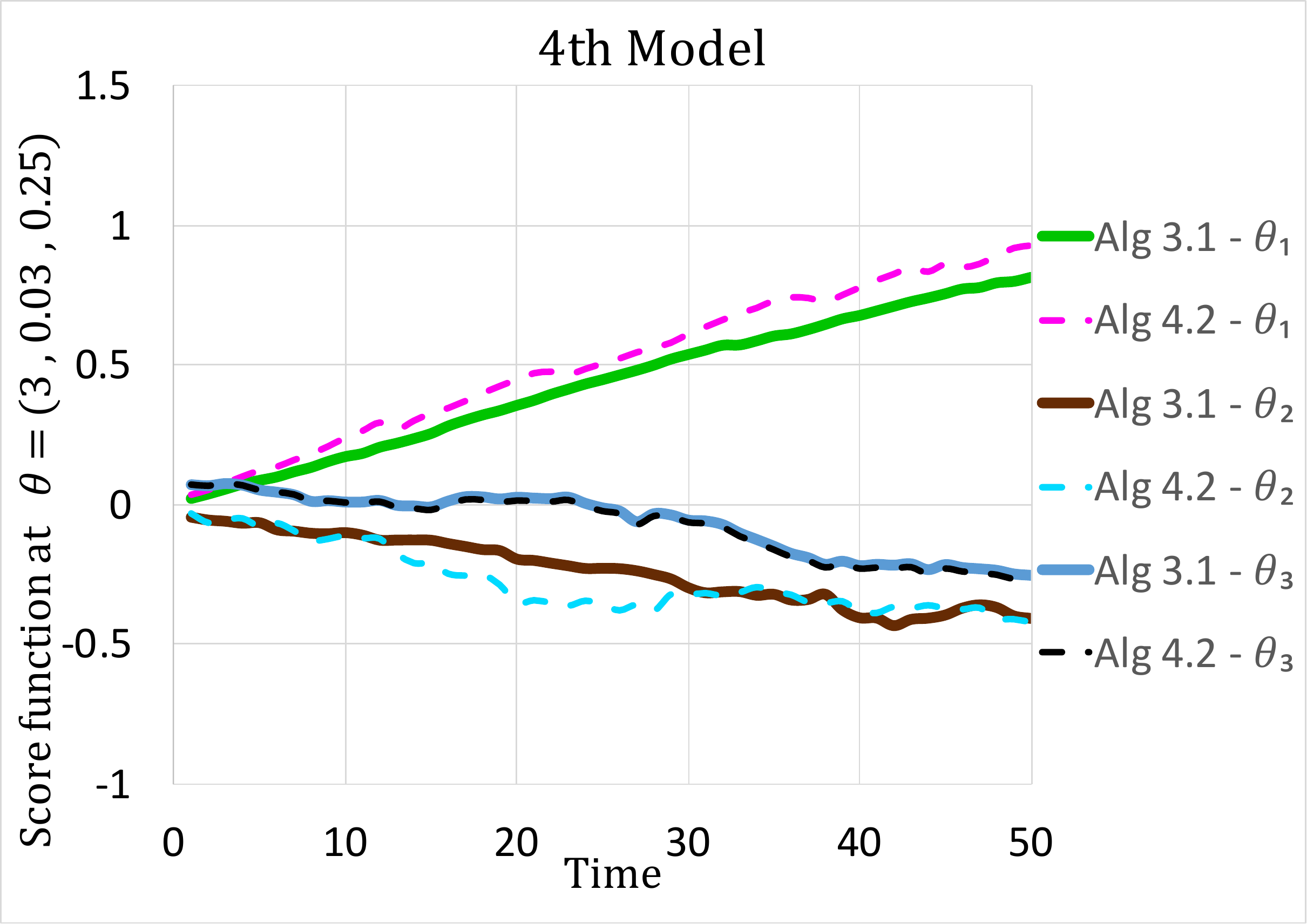}}
   \caption{Trajectories from the execution of Algorithms \ref{alg:online_disc} and \ref{alg:online_disc_new} for the estimation of the score function in Models  1-4.} 
     \label{fig:score_func}
\end{figure}


\subsubsection{Cost Comparison}

We now consider comparing the costs of Algorithms  \ref{alg:online_disc_new} and \ref{alg:online_disc_ml}. We take $l_*$ to be 7 in the models 1, 2 \& 3 and 8 in the 4th model. The parameters of the model are as in the previous section. The ground truth is computed at level 11 with $N=2000$ using \autoref{alg:online_disc_new}. We run 56 simulations of both algorithms for $L\in\{l_*-1,\cdots,10\}$. For each $L$, the number of particles are carefully chosen to give similar MSE values from both algorithms. Particularly, the number of particles in \autoref{alg:online_disc_new} is $N_L=\floor*{C_1 2^L}$  and for each level $l$ in \autoref{alg:online_disc_new}. In \autoref{fig:pf_ml_cost} the number of particles is $N_l=\floor*{C_2 2^L(L-l_*+2)\Delta_l^{1/2+\rho}}$ (in models 1 to 3) and $N_l=\floor*{C_2 2^L(L-l_*+2)\Delta_l^{1/2+\rho}\Delta_L^{-\rho}}$ (in model 4), where $C_1$ and $C_2$ are constants. In  we can observe the cost against MSE curve, that appear to follow our conjectures over algorithmic costs earlier in the article.


\begin{figure}[h] 
\centering
  \subfloat{\includegraphics[height=0.3\textwidth]{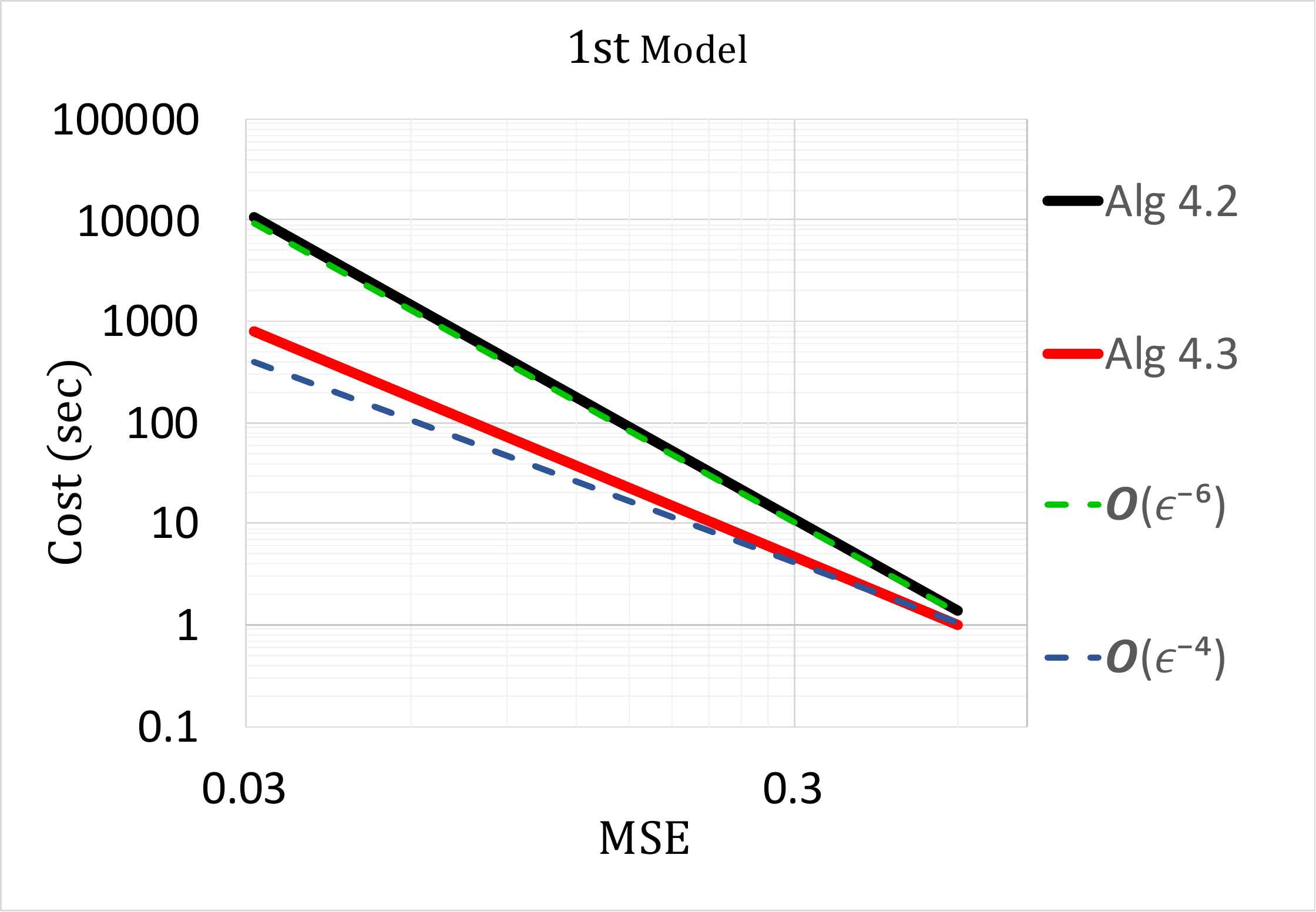}}\,\,\,
  \subfloat{\includegraphics[height=0.3\textwidth]{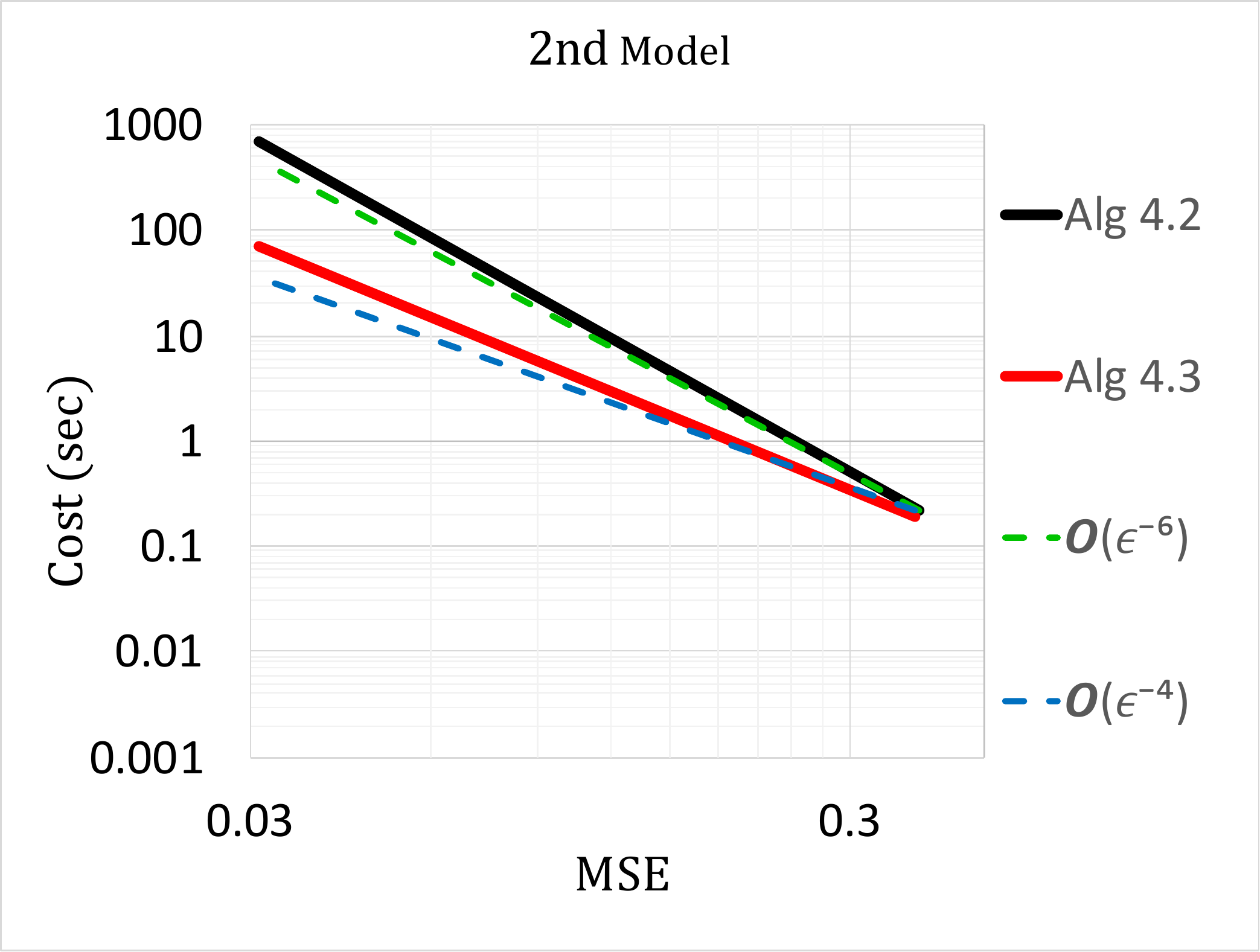}}\\
  \vspace{-8pt}
  \subfloat{\includegraphics[height=0.3\textwidth]{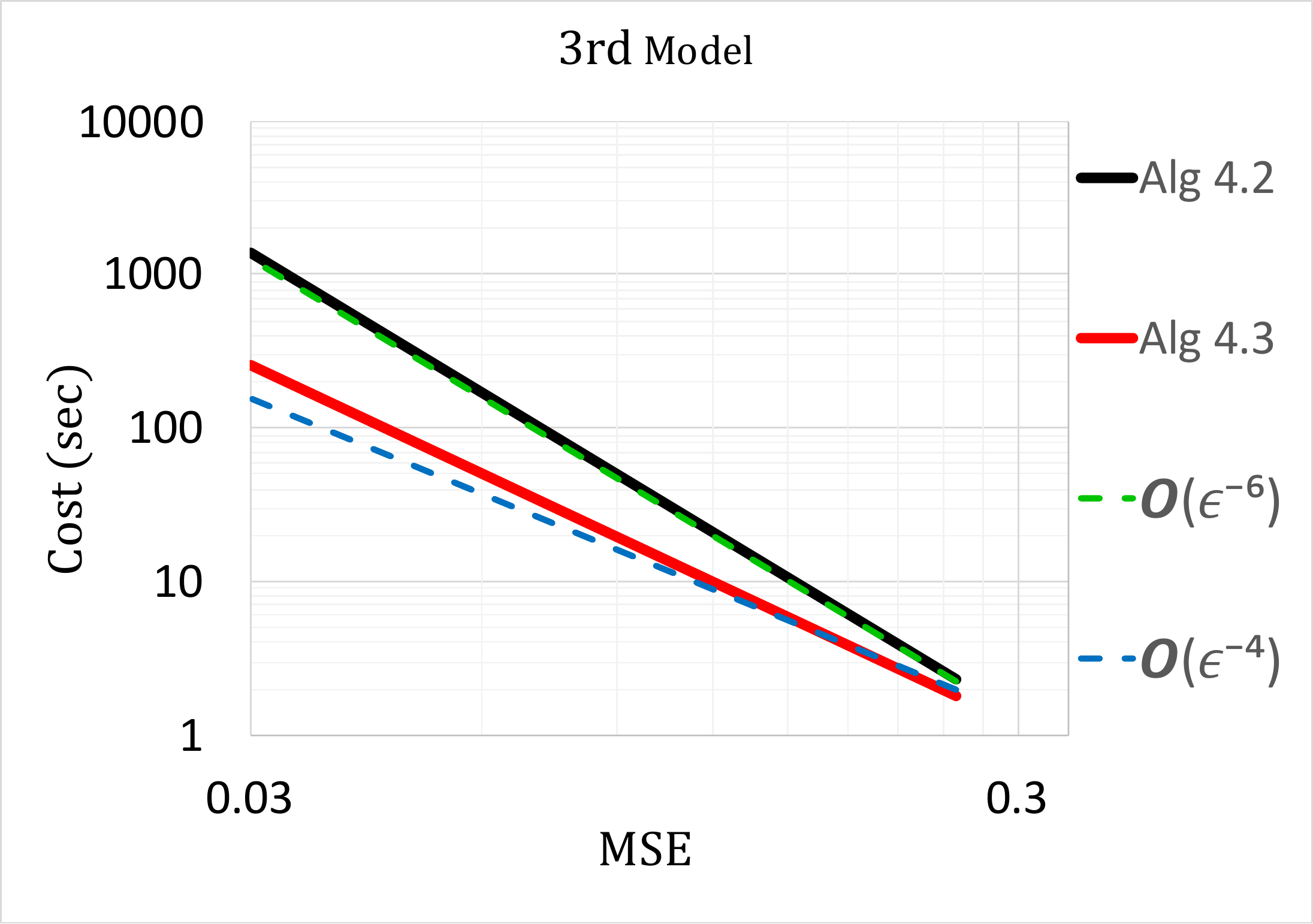}}\,\,\,
  \subfloat{\includegraphics[height=0.3\textwidth]{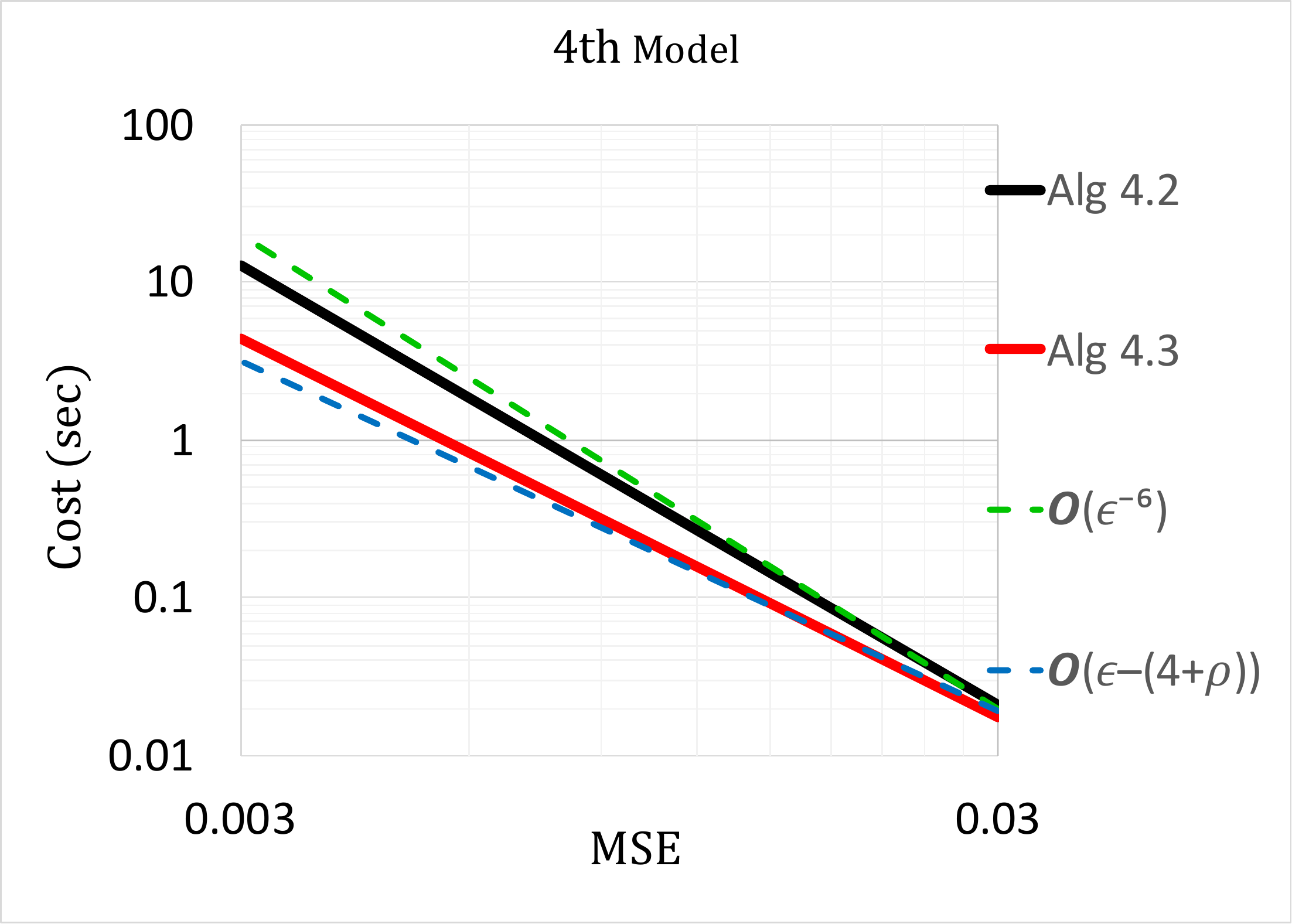}}
   \caption{Cost per each simulation versus MSE on a log-log scale for Algorithms \ref{alg:online_disc_new} and \ref{alg:online_disc_ml}. The dashed lines are for reference.}
    \label{fig:pf_ml_cost}
\end{figure}

\subsubsection{Parameter Estimation}
We use Algorithms \ref{alg:online_disc}, \ref{alg:online_disc_ml} to estimate the parameters in each model. In \autoref{alg:online_disc}, the level of discretization, $l$, is 10 for models 1-3 and 9 for model 4, and the number of particles, $N$, is 2,000 for models 1-3 and 2500 for model 4. In \autoref{alg:online_disc_ml}, we use $l_*=7$, $L=10$ and the number of particles on each level $l\in \{l_{*}-1,\cdots,L\}$ is $N_l= 2^L (L-l_*+2) \Delta_l^{1/2+\rho}$, where  $\rho\in\{0.14, \, 0.09, \,0.11\}$ in Models 1, 2 and 3, respectively. In model 4, $l_* = 8$, $L=9$ and the number of particles on each level $l\in \{l_{*}-1,\cdots,L\}$ is $N_l=1.4 \times 2^L (L-l_*+2) \Delta_l^{1/2+\rho} \Delta_L^{-\rho}$ where $\rho=0.1$.

\autoref{fig: par_est_model1} considers Model 1. We fix $x_*=0.2$, $\sigma = 0.3$, $\kappa= 2$, $T=20,000$.
The parameter values used to generate the data are $(\theta_1^\star,\theta_2^\star)=(-0.7,-0.5)$. For the stochastic gradient algorithm,
we used an initial value $(-0.05, -1.5)$ and step-size $\alpha_k=k^{-0.85}$. 
\autoref{fig:par_est_model2} considers Model 2. We fix $x_*=1.8$, $\sigma = 0.25$, $\kappa= 2.2$, $T=20,000$.
The parameter values used to generate the data are $(\theta_1^\star,\theta_2^\star,\theta_3^\star)=(1.3,-0.5,0.18)$. For the stochastic gradient algorithm,
we used initial value $(0.8,-1, 0.8)$ and step-size $\alpha_k=k^{-0.95}$. 
\autoref{fig:par_est_model3} considers Model 3. We fix $x_*=1.5$, $\sigma = 0.25$, $\kappa= 2$, $T=20,000$.
The parameter values used to generate the data are $(\theta_1^\star,\theta_2^\star,\theta_3^\star)=(2,1,0.45)$. For the stochastic gradient algorithm,
we used an initial value $(1.24, 0.6, 1.11)$ and step size $\alpha_k=k^{-0.9}$. 
\autoref{fig:par_est_model4} considers Model 4 applied to the data in \autoref{fig:Tesla_stock}. We fix $x_*=1.3$, $\beta = 2$, $T=11425$ (there is a rescaling of the time parameter). For the stochastic gradient algorithm,
we used an initial value $(2.4, 0.5, 0.4)$ and step size $\alpha_k=k^{-0.82}$.
In all cases considered (Figures \ref{fig: par_est_model1}-\ref{fig:par_est_model4}) our selected  settings allow for an accurate estimation of the parameter values over long time periods.

\begin{figure} 
\centering
    \subfloat{\includegraphics[height=0.2\textwidth]{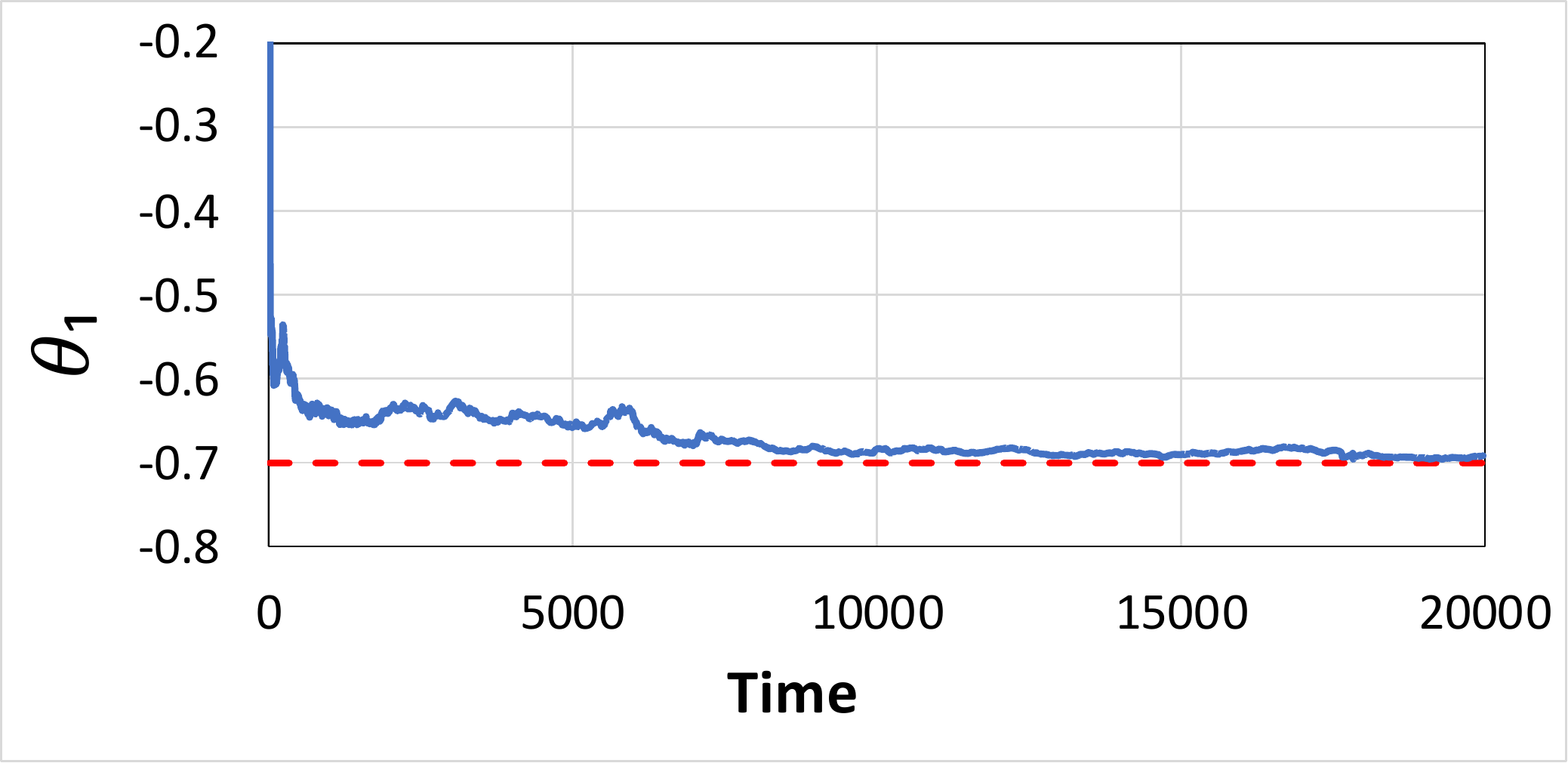} }
    \subfloat{\includegraphics[height=0.2\textwidth]{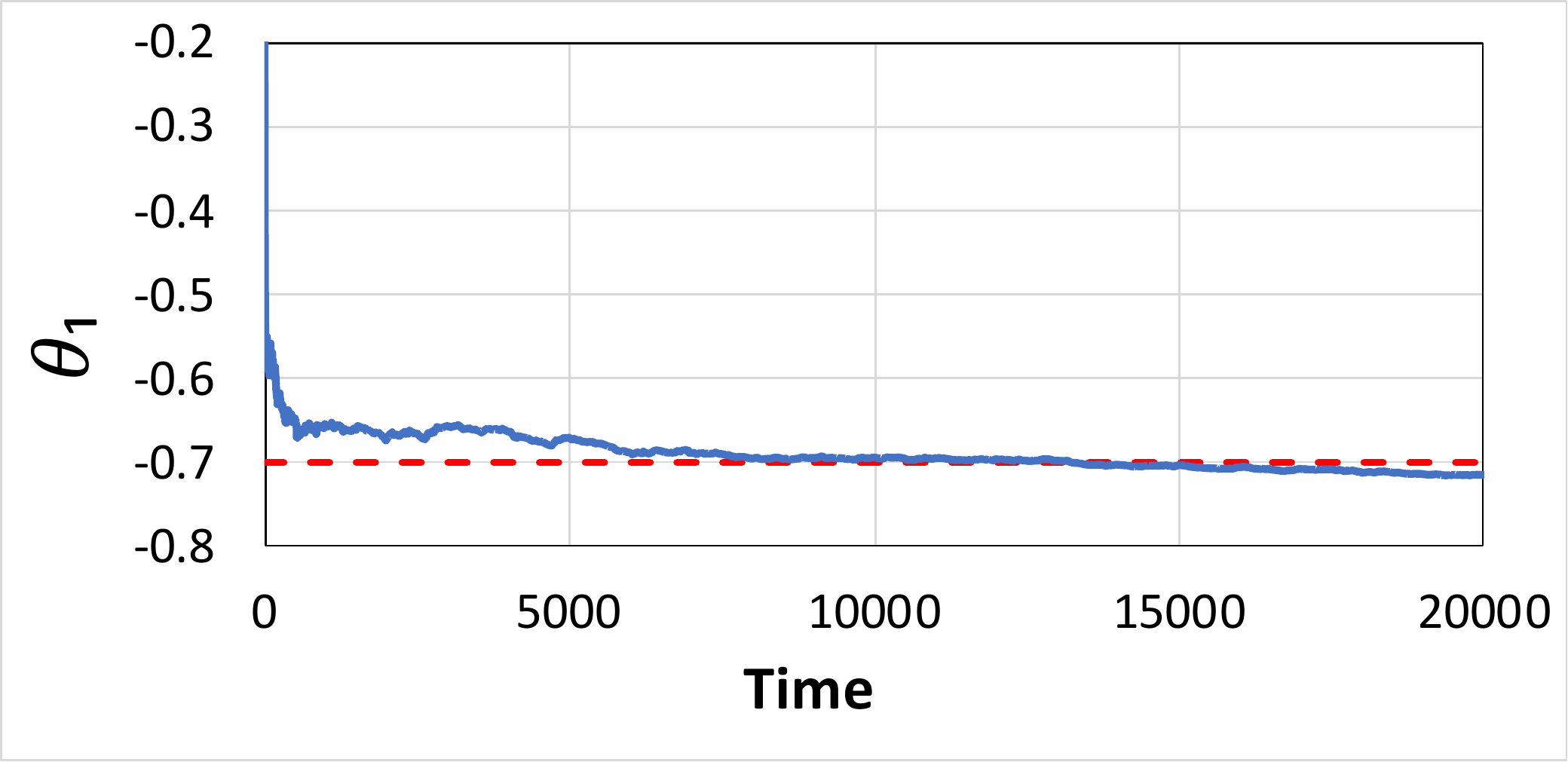} }\\ 
   \vspace{-10pt} 
   \subfloat{\includegraphics[height=0.2\textwidth]{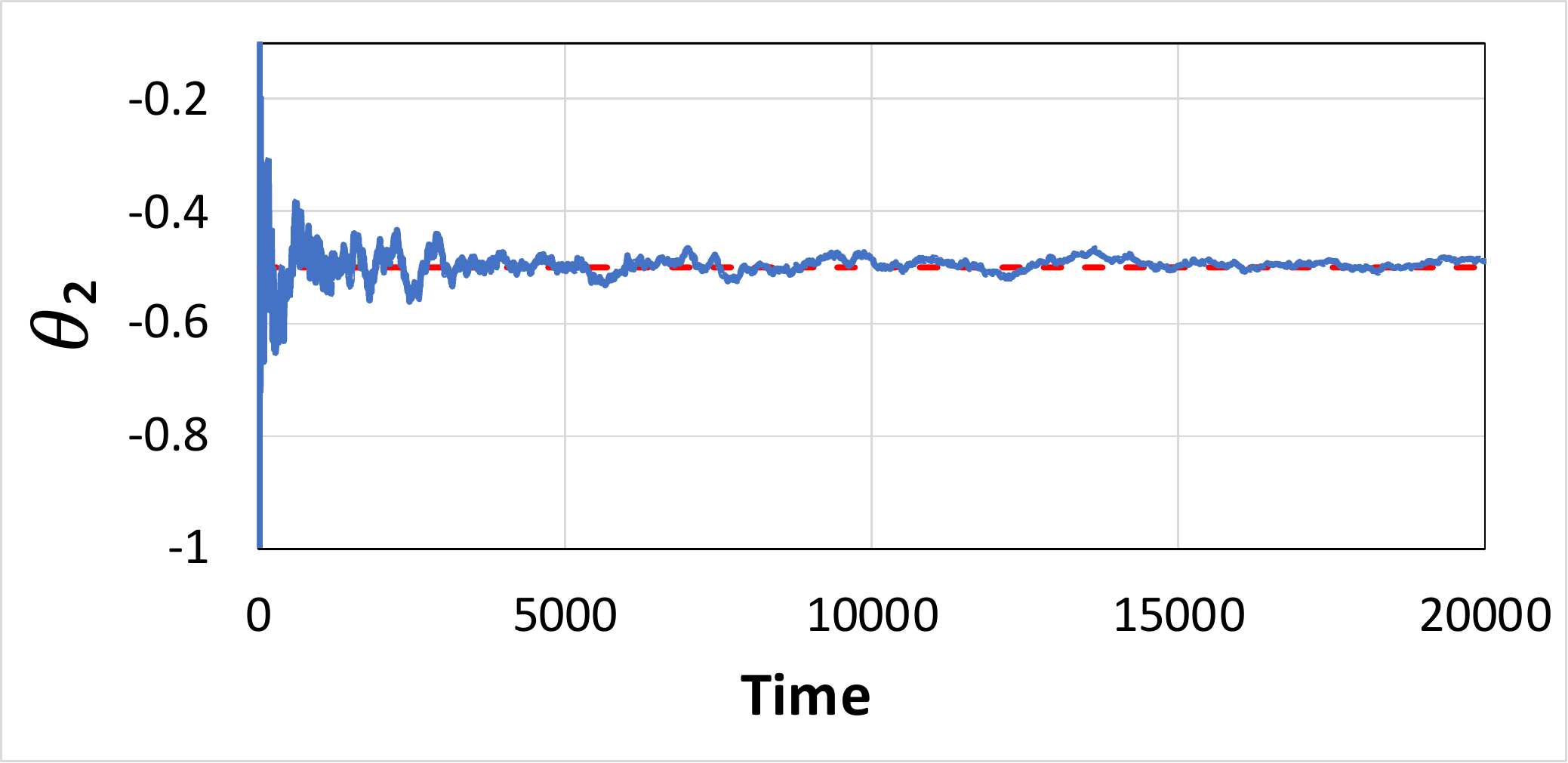}} \,
   \subfloat{\includegraphics[height=0.2\textwidth]{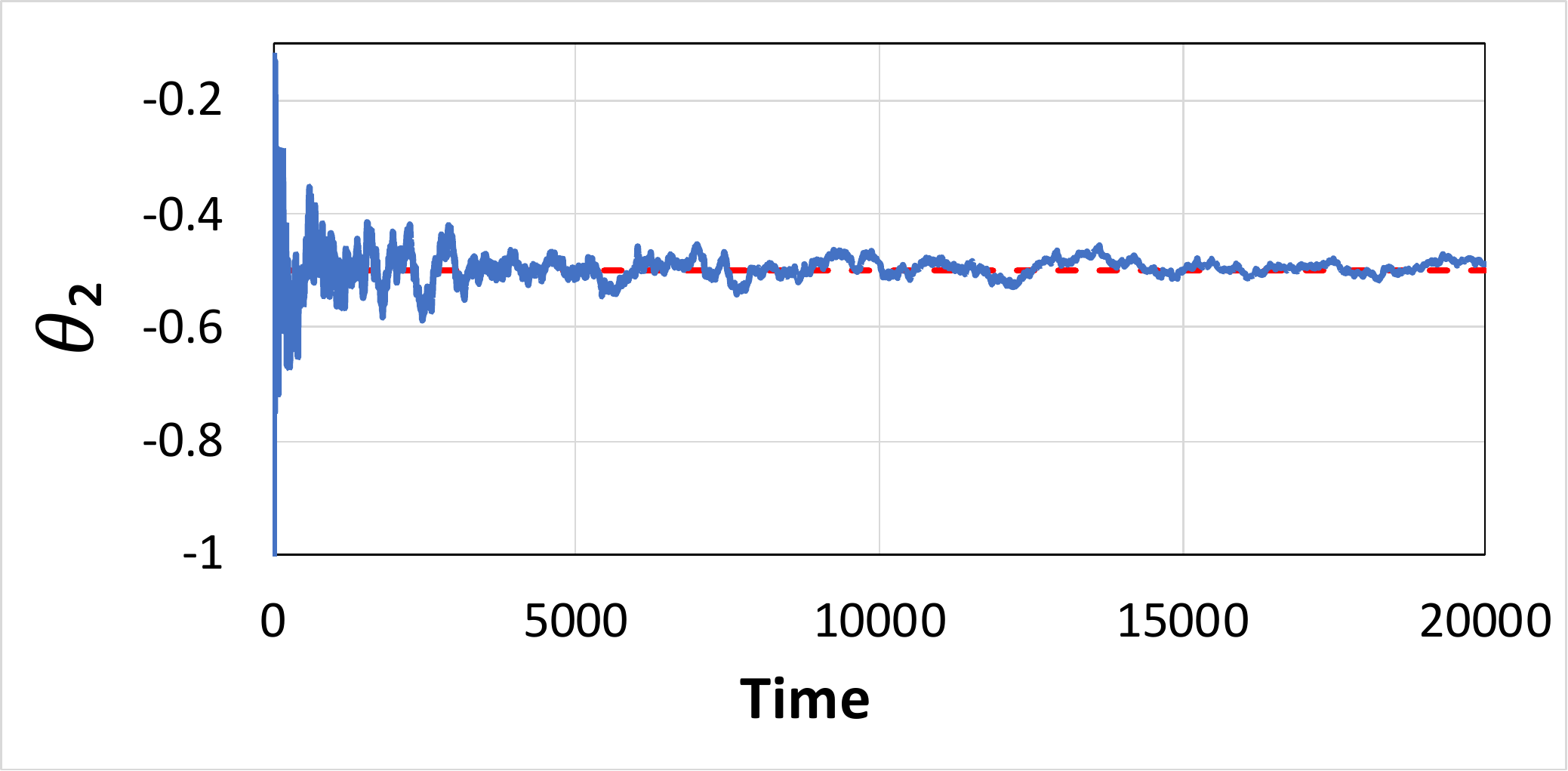} }
   \caption{Trajectories from the execution of \autoref{alg:online_disc} (left panel) and \autoref{alg:online_disc_ml} (right panel) for the estimation of $(\theta_1, \theta_2)$ from Model 1. The horizontal dashed lines in the plots show the true parameter values $(\theta_1^\star,\theta_2^\star)=(-0.7,-0.5)$.}
    \label{fig: par_est_model1}
\end{figure}

\begin{figure} 
\centering
    \subfloat{\includegraphics[height=0.2\textwidth]{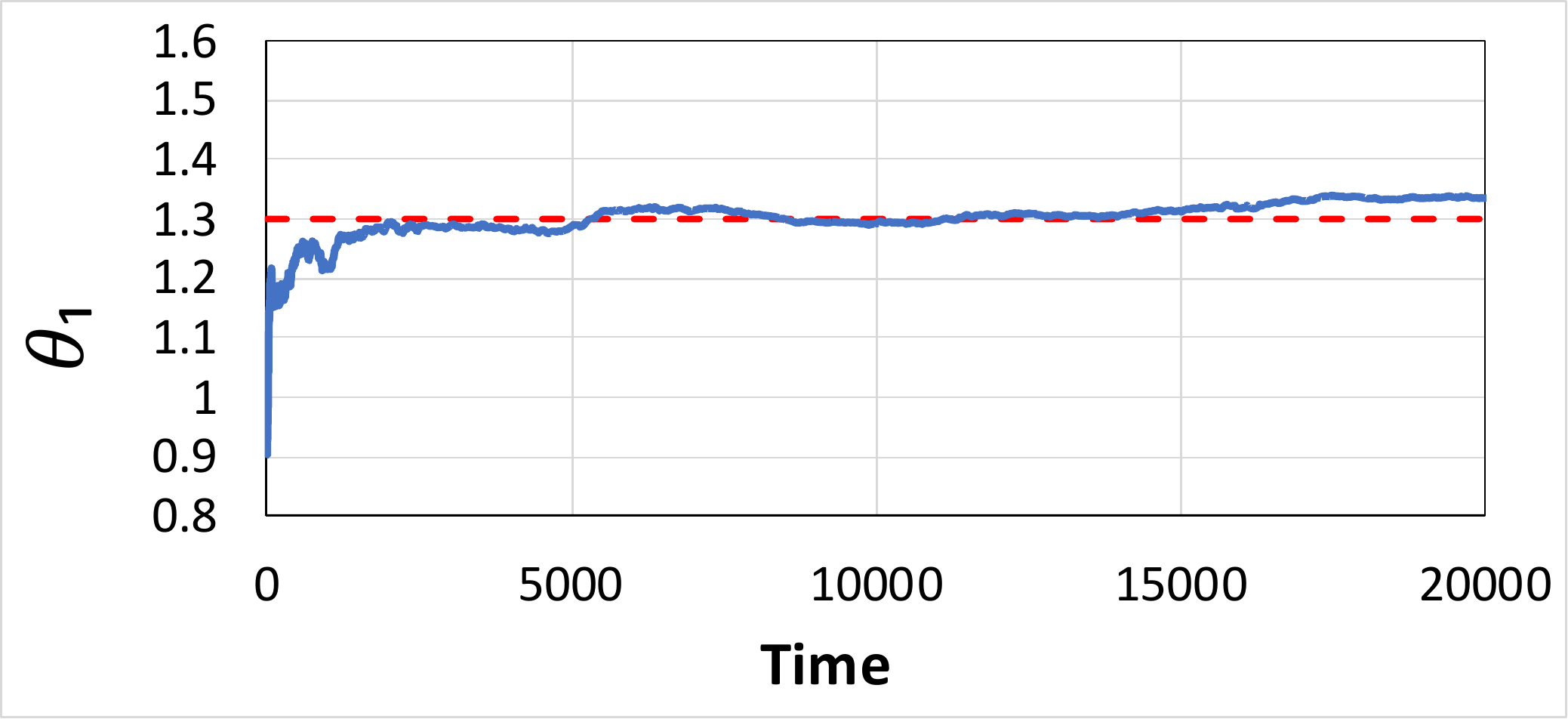} }
    \subfloat{\includegraphics[height=0.2\textwidth]{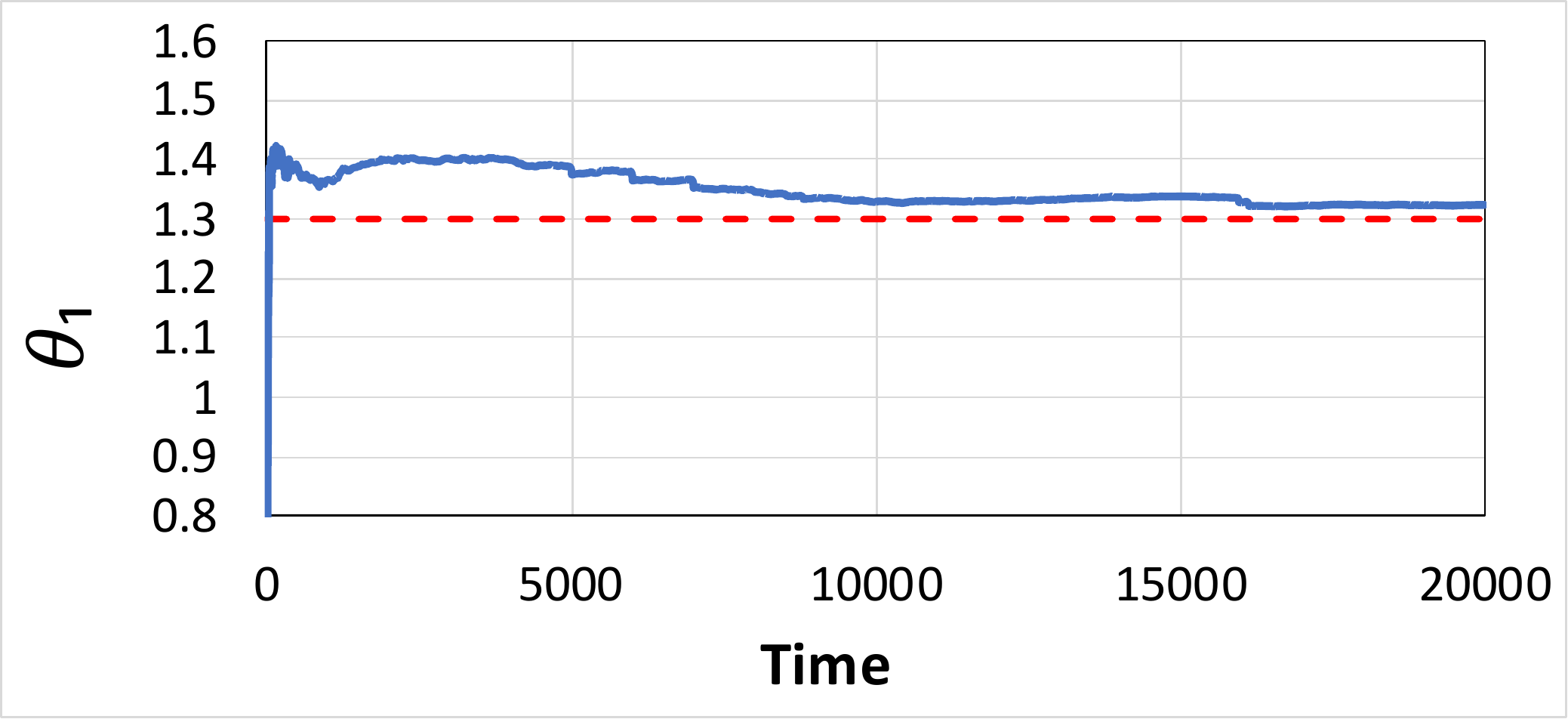} }\\
\vspace{-10pt}   \subfloat{\includegraphics[height=0.2\textwidth]{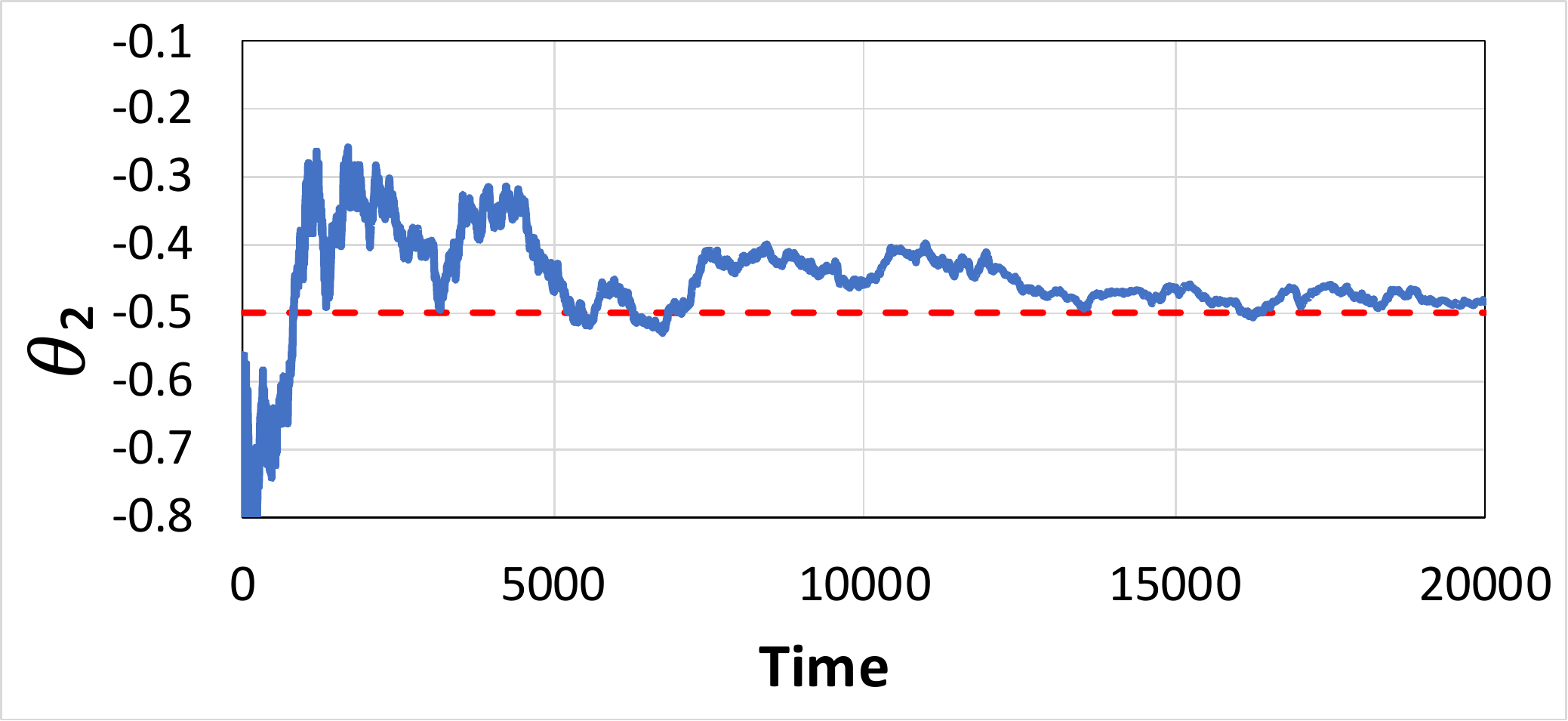} }
   \subfloat{\includegraphics[height=0.2\textwidth]{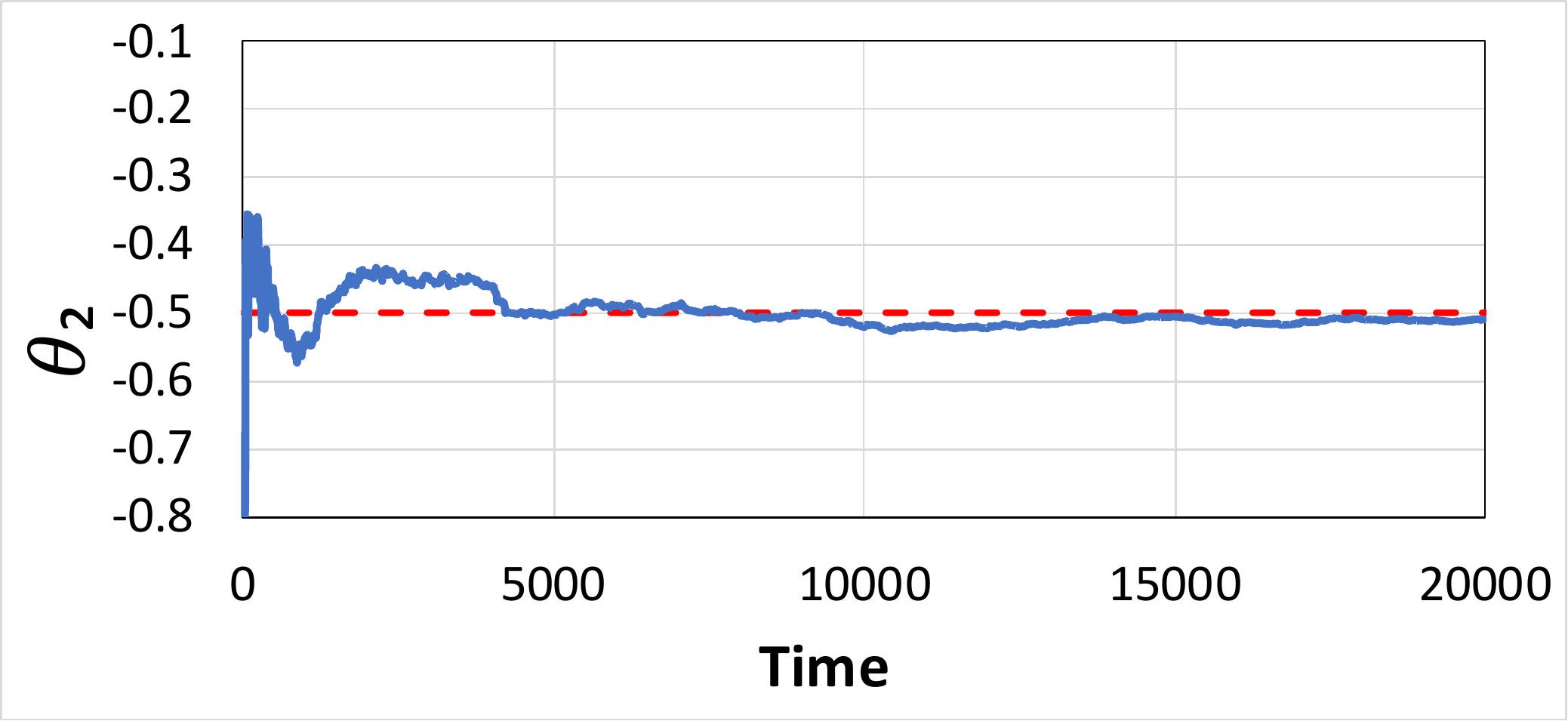} }\\
\vspace{-10pt}  
    \subfloat{\includegraphics[height=0.2\textwidth]{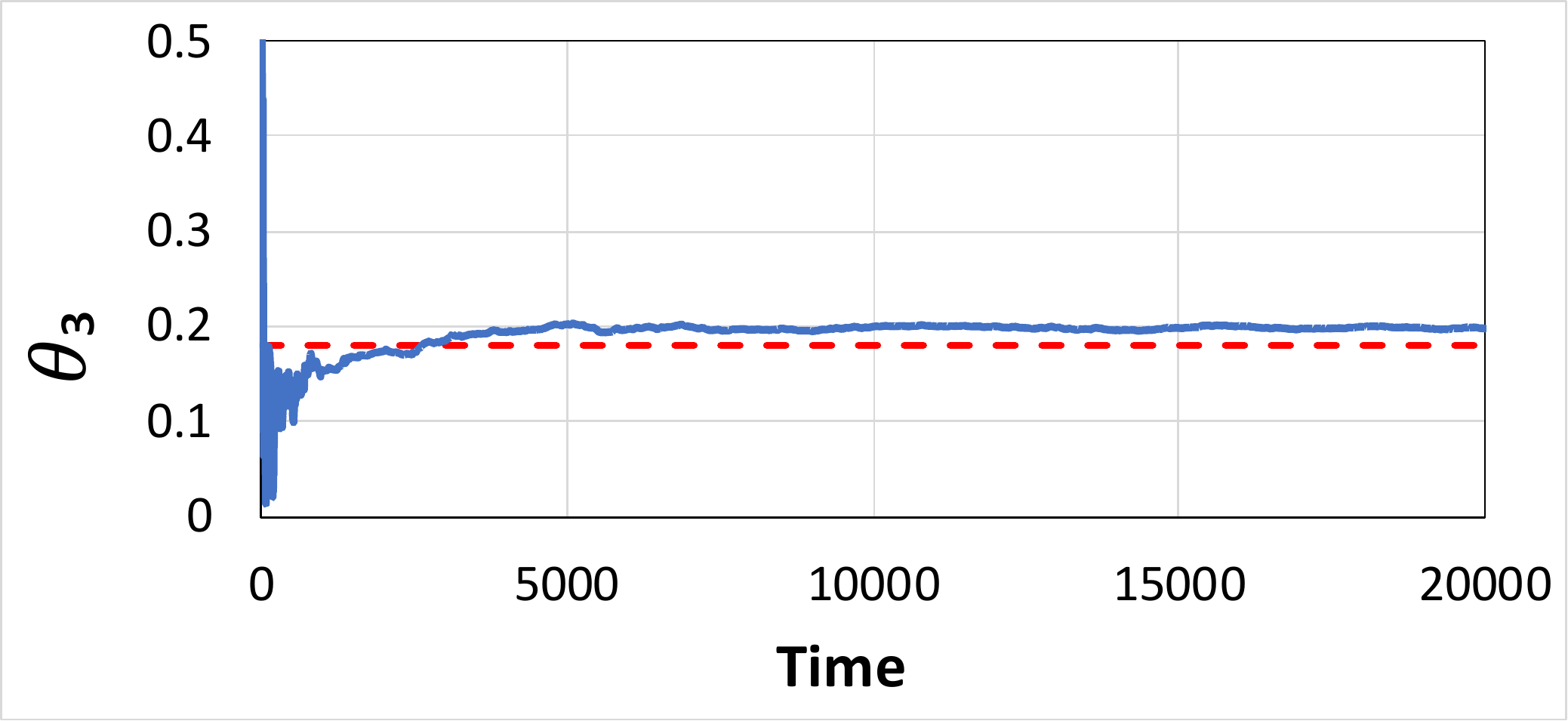}} \hspace{0.5pt}
    \subfloat{\includegraphics[height=0.2\textwidth]{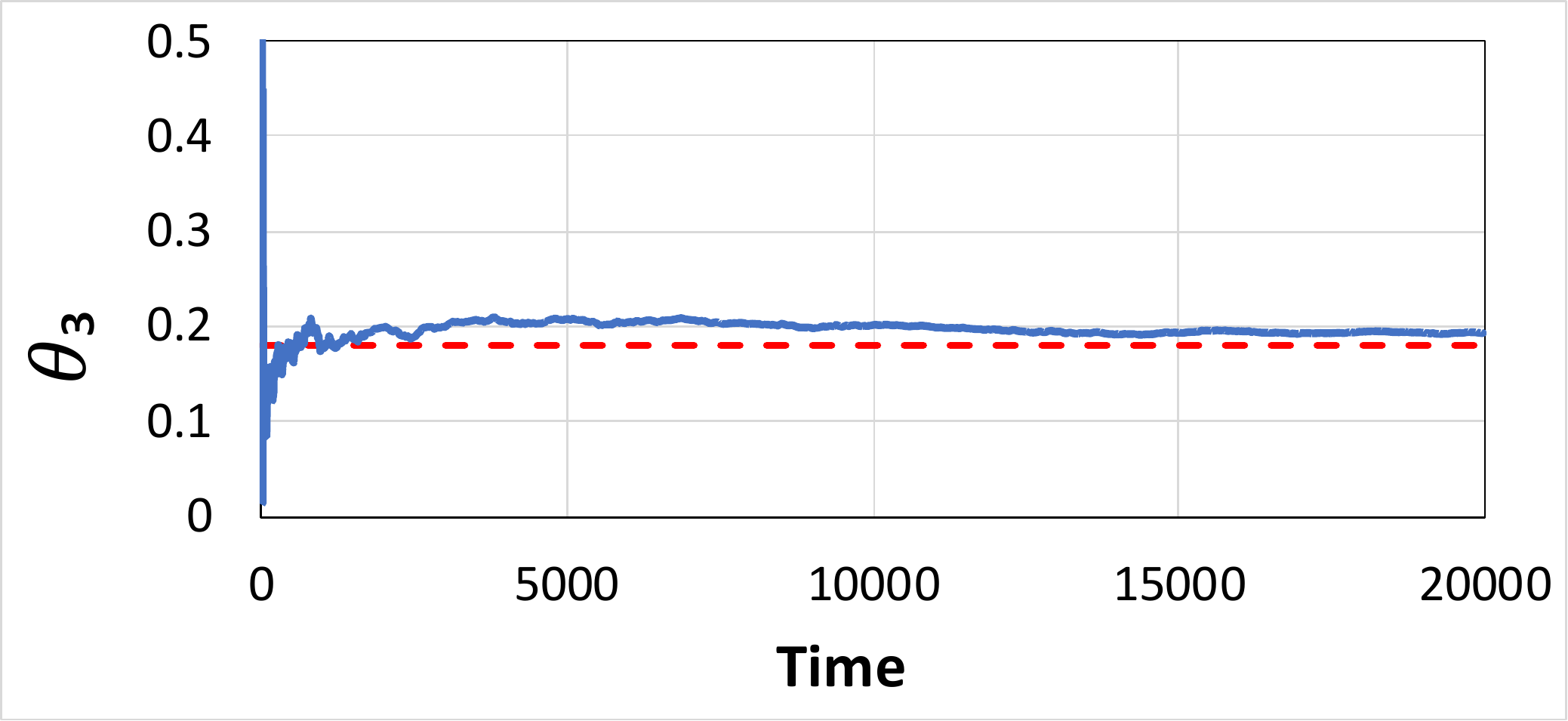} }
   \caption{Trajectories from the execution of \autoref{alg:online_disc} (left panel) and \autoref{alg:online_disc_ml} (right panel) for the estimation of $(\theta_1, \theta_2, \theta_3)$ of the second model. We used an initial value $(0.8,-1, 0.8)$. The horizontal dashed lines in the plots show the true parameter values $(\theta_1^\star,\theta_2^\star,\theta_3^\star)=(1.3,-0.5,0.18)$. }
    \label{fig:par_est_model2}
\end{figure}

\begin{figure}
\centering
    \subfloat{\includegraphics[height=0.2\textwidth]{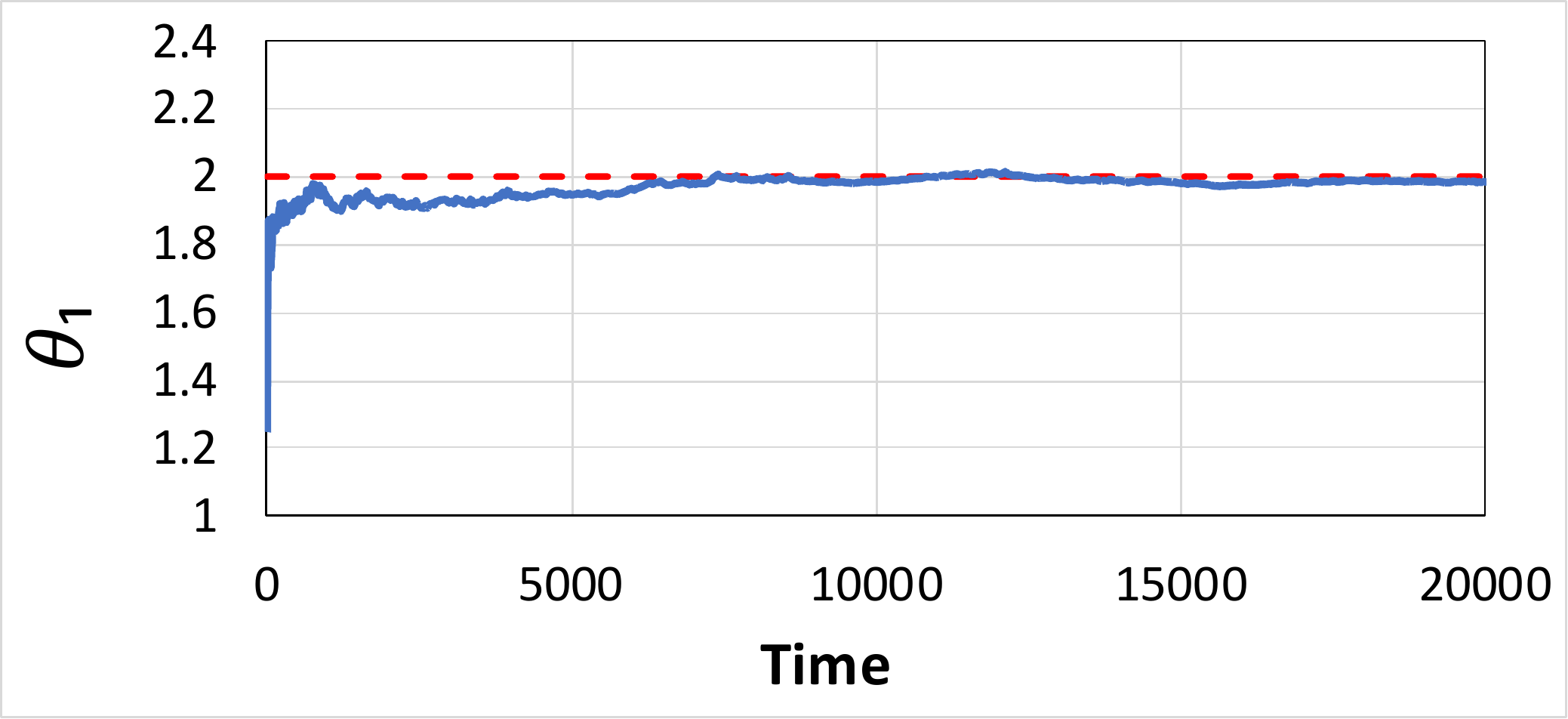} }
    \subfloat{\includegraphics[height=0.2\textwidth]{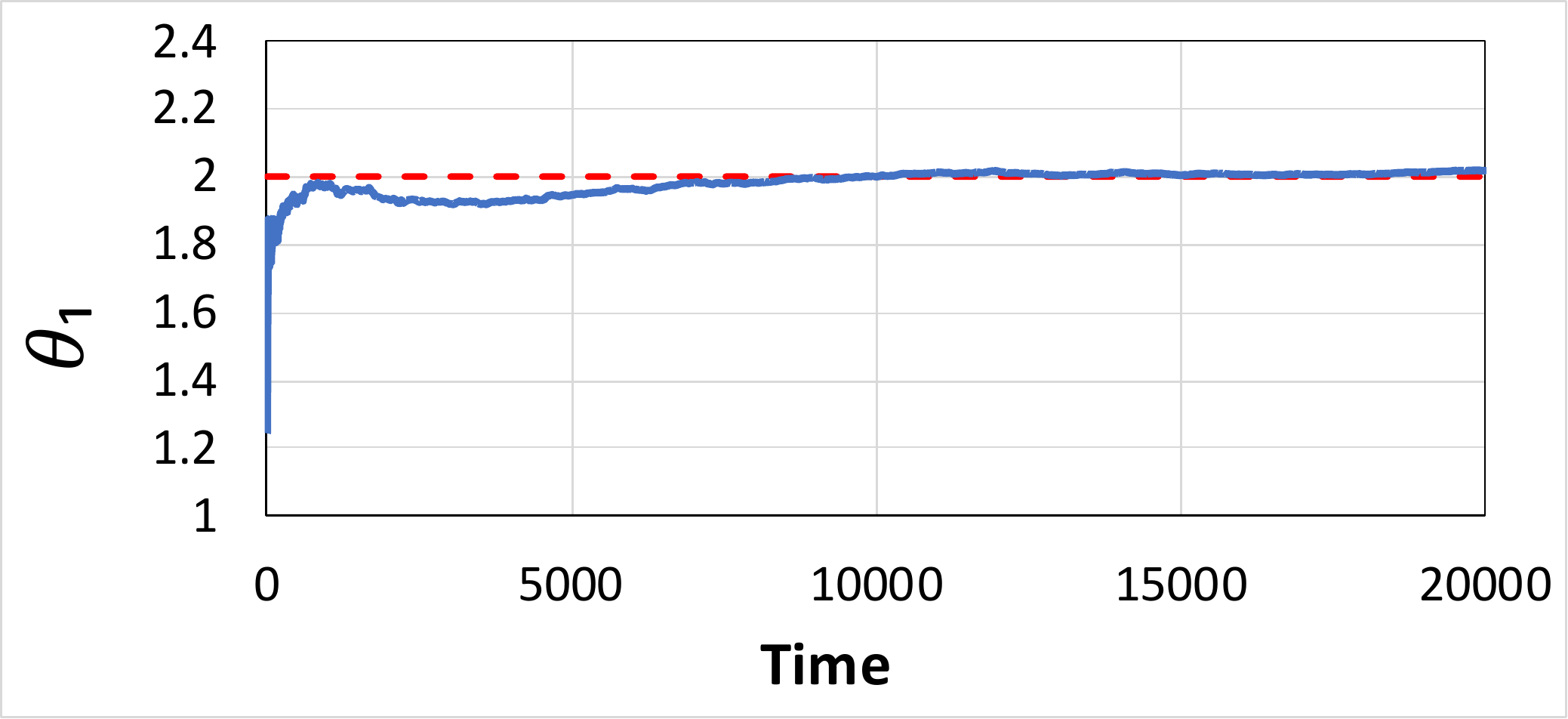} }\\
  \vspace{-10pt} \subfloat{\includegraphics[height=0.2\textwidth]{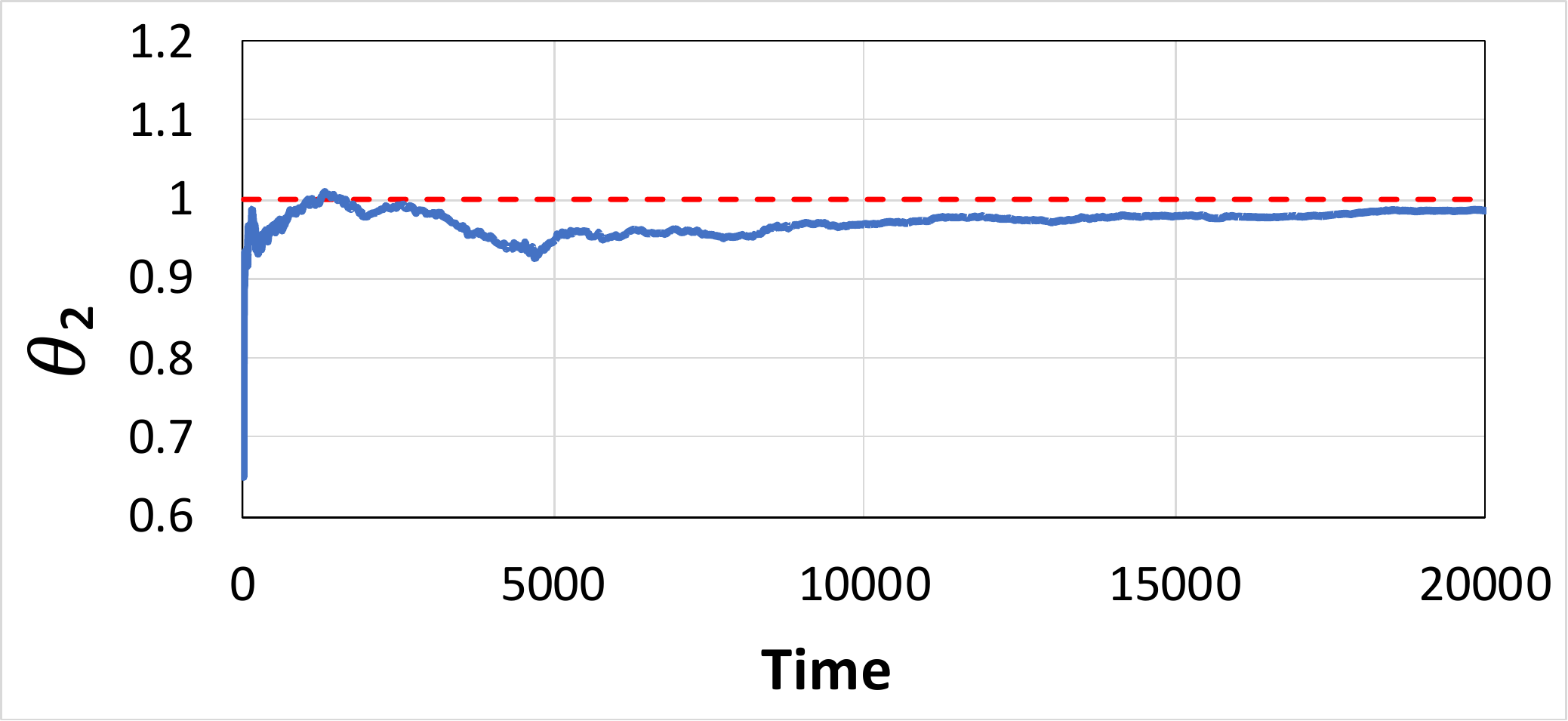}}\, 
   \subfloat{\includegraphics[height=0.2\textwidth]{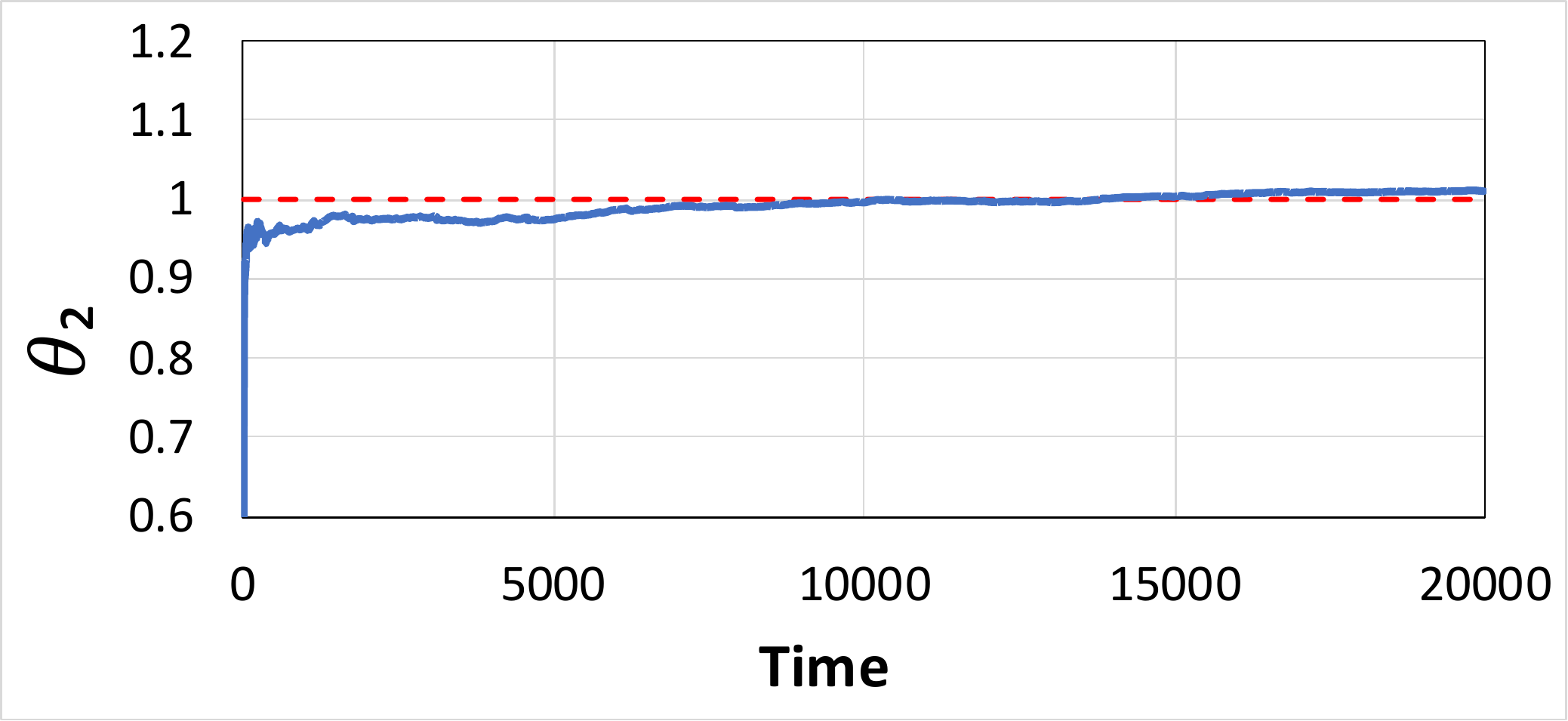} }\\
   \vspace{-10pt}
    \subfloat{\includegraphics[height=0.2\textwidth]{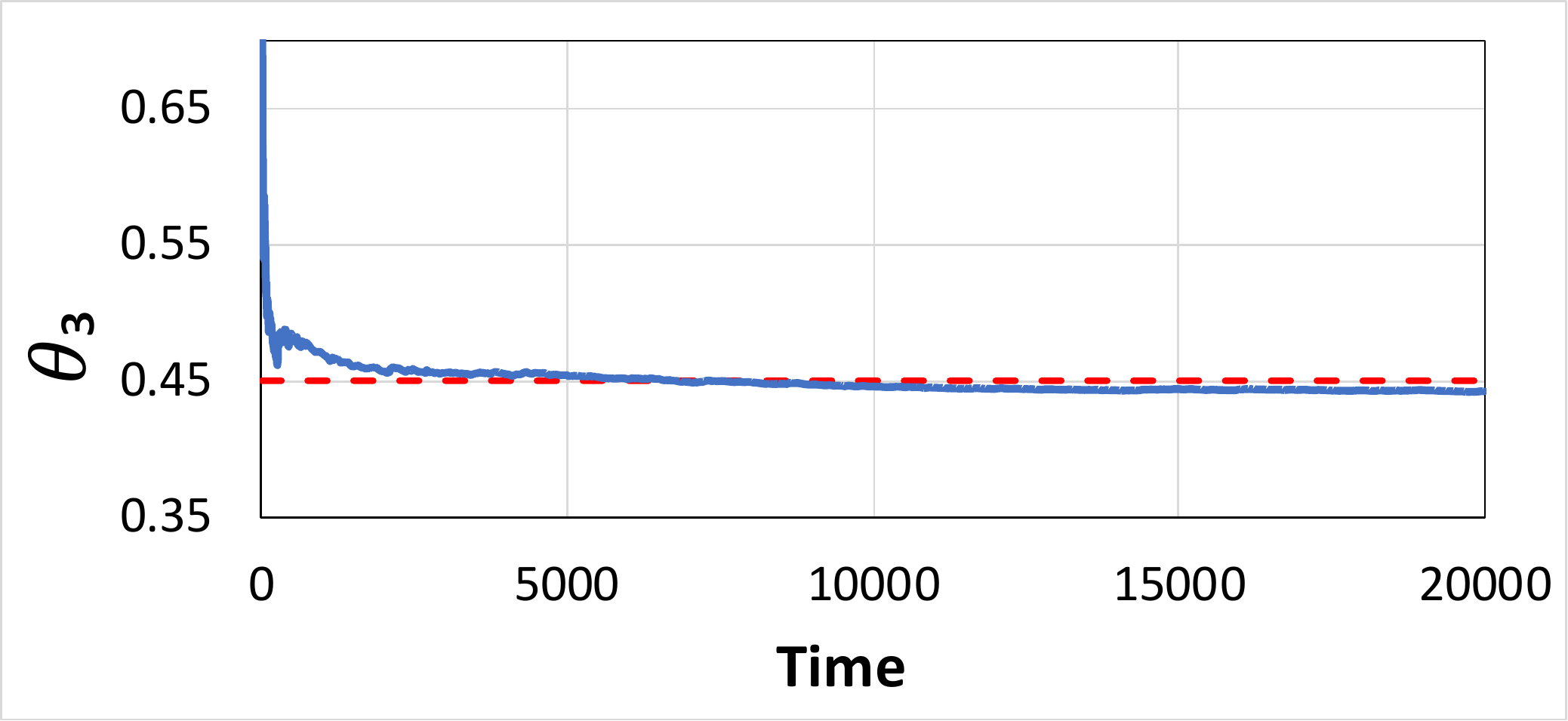}} \,
   \subfloat{\includegraphics[height=0.2\textwidth]{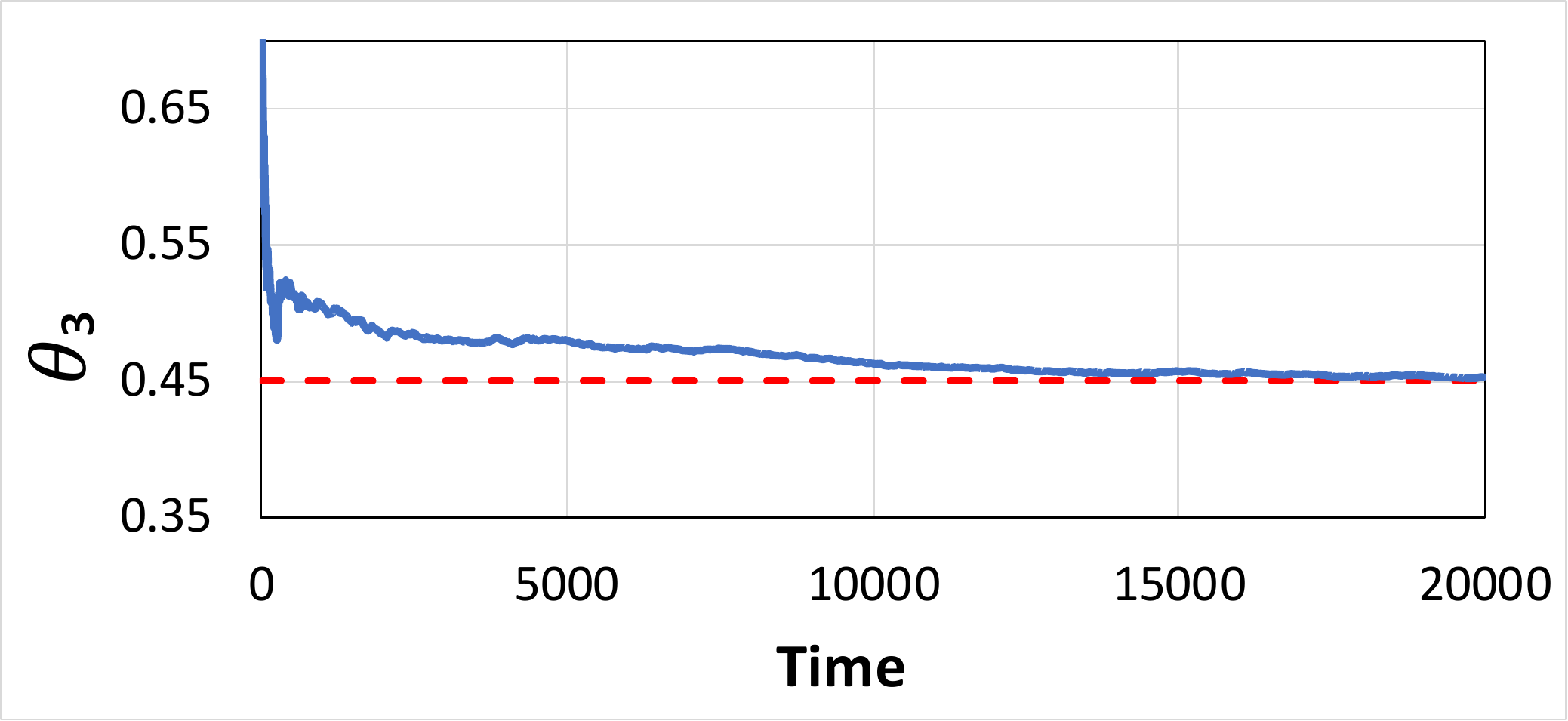} }
   \caption{Trajectories from the execution of \autoref{alg:online_disc} (left panel) and \autoref{alg:online_disc_ml} (right panel) for the estimation of $(\theta_1, \theta_2, \theta_3)$ from Model 3. The horizontal dashed lines in the plots show the true parameter values $(\theta_1^\star,\theta_2^\star,\theta_3^\star)=(2,1,0.45)$. }
    \label{fig:par_est_model3}
\end{figure}

\begin{figure}
\centering
    \subfloat{\includegraphics[height=0.2\textwidth]{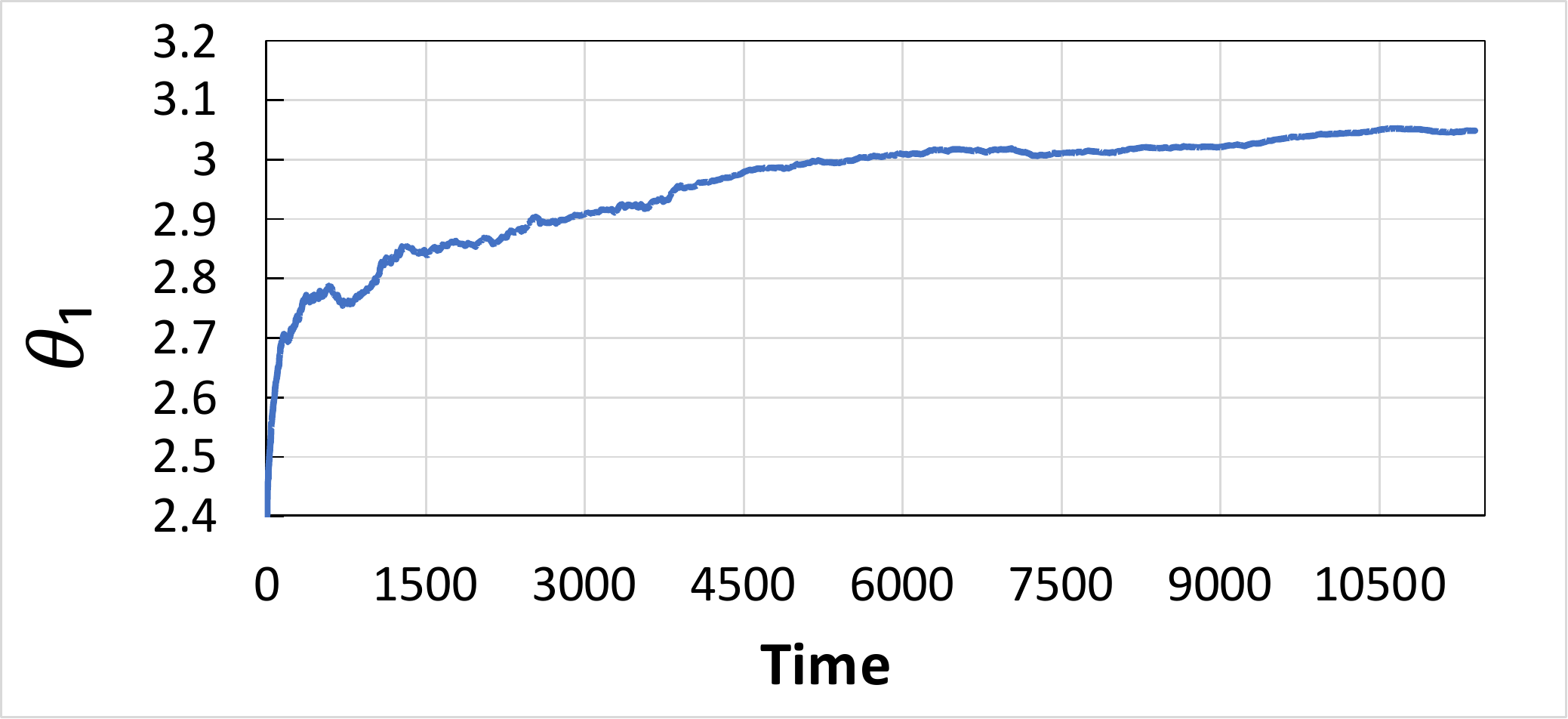} }
    \subfloat{\includegraphics[height=0.2\textwidth]{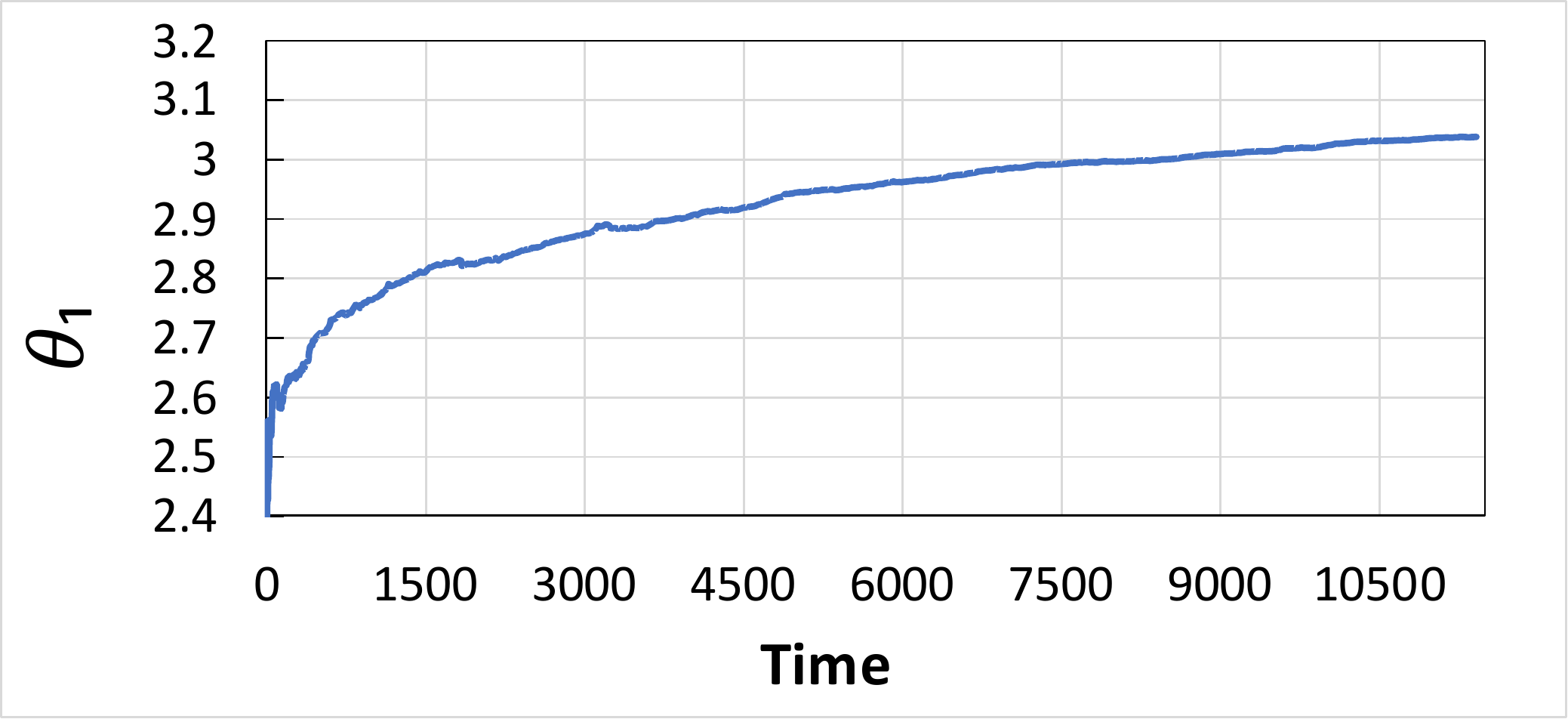} }\\
  \vspace{-10pt} \subfloat{\includegraphics[height=0.2\textwidth]{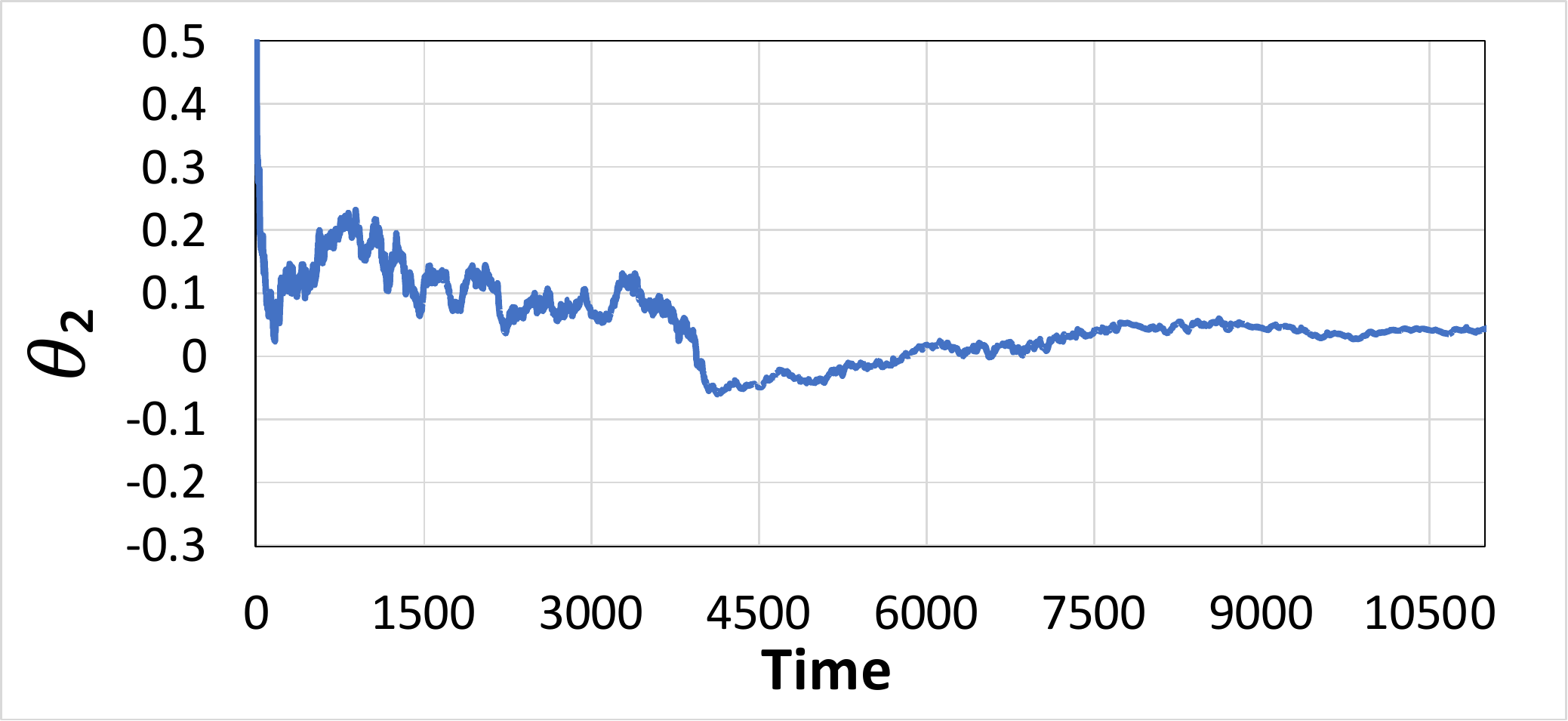}}\, 
   \subfloat{\includegraphics[height=0.2\textwidth]{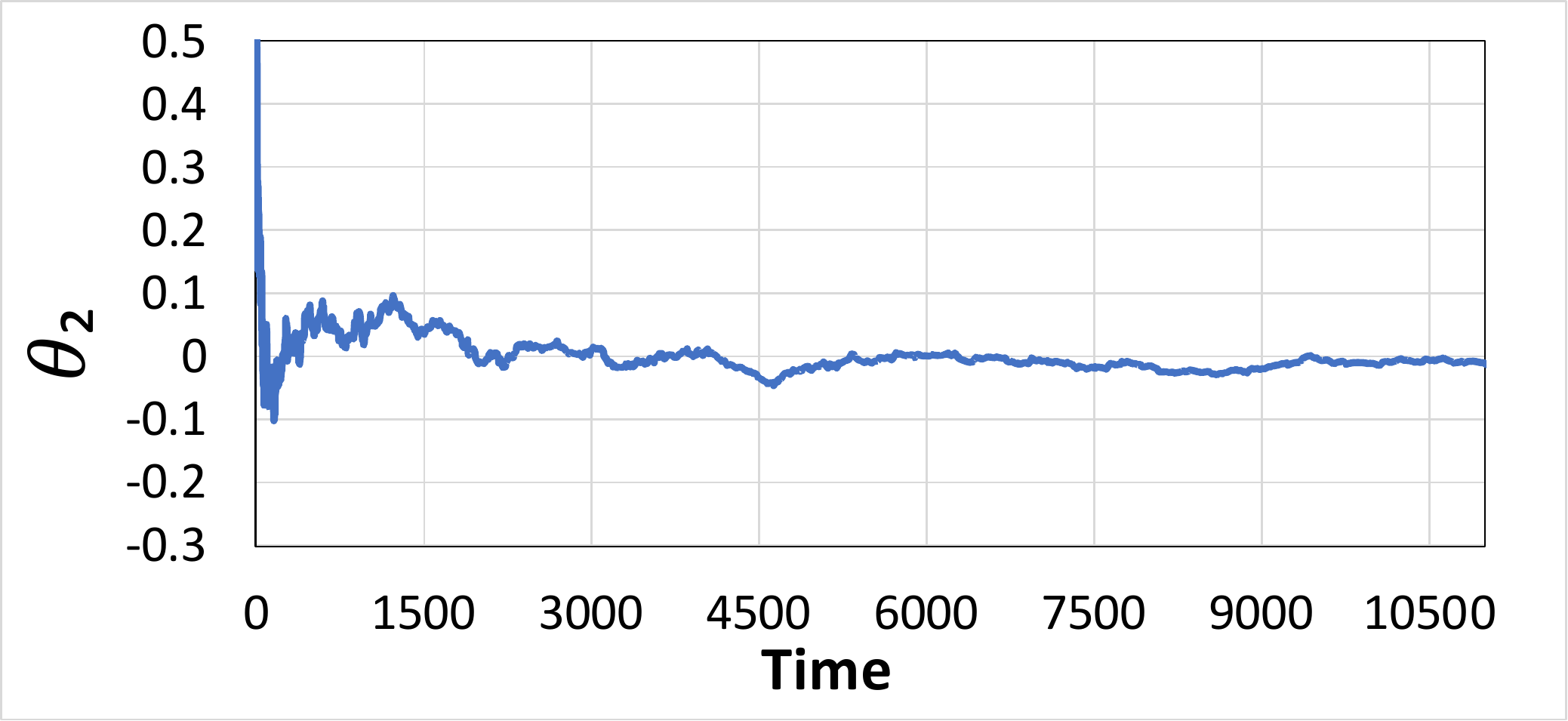} }\\
   \vspace{-10pt}
    \subfloat{\includegraphics[height=0.2\textwidth]{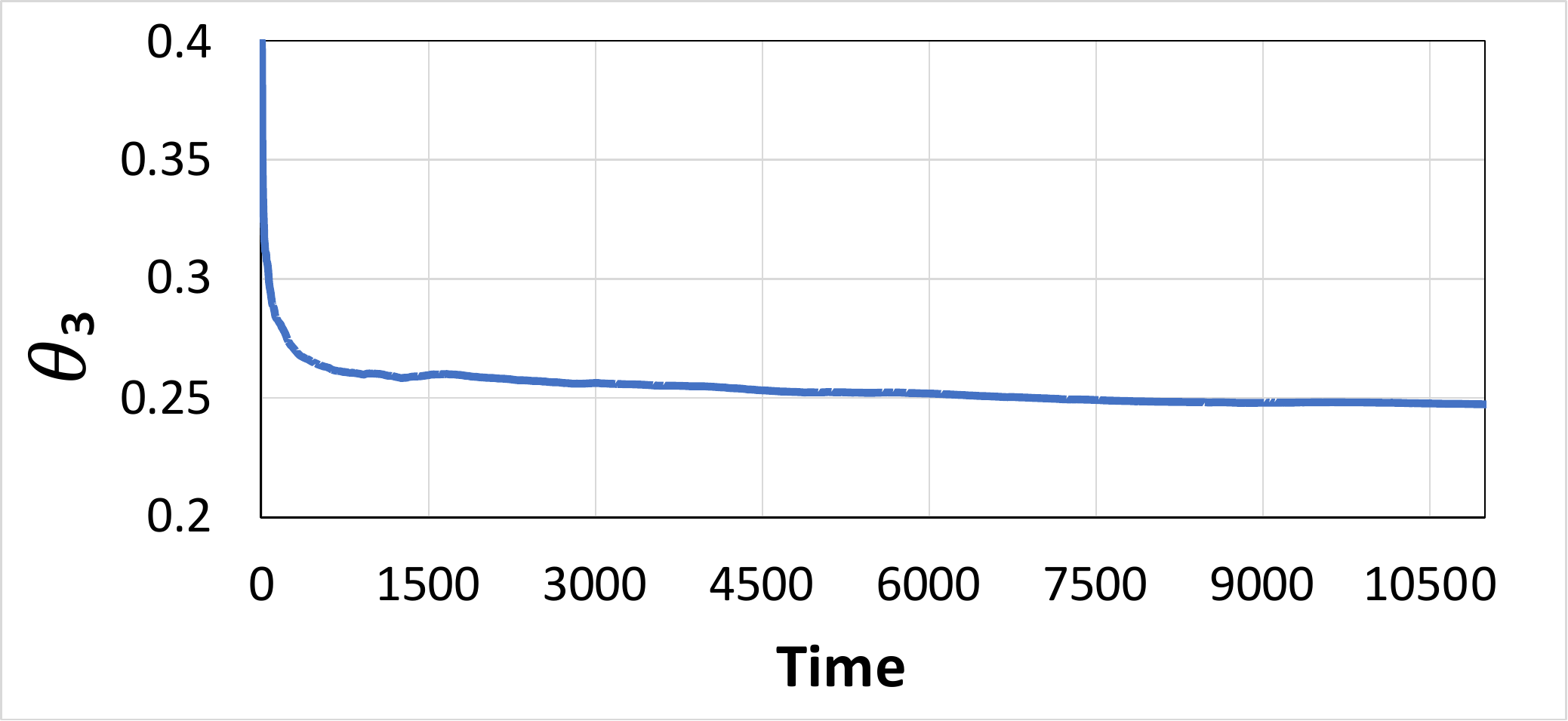}} \,
   \subfloat{\includegraphics[height=0.2\textwidth]{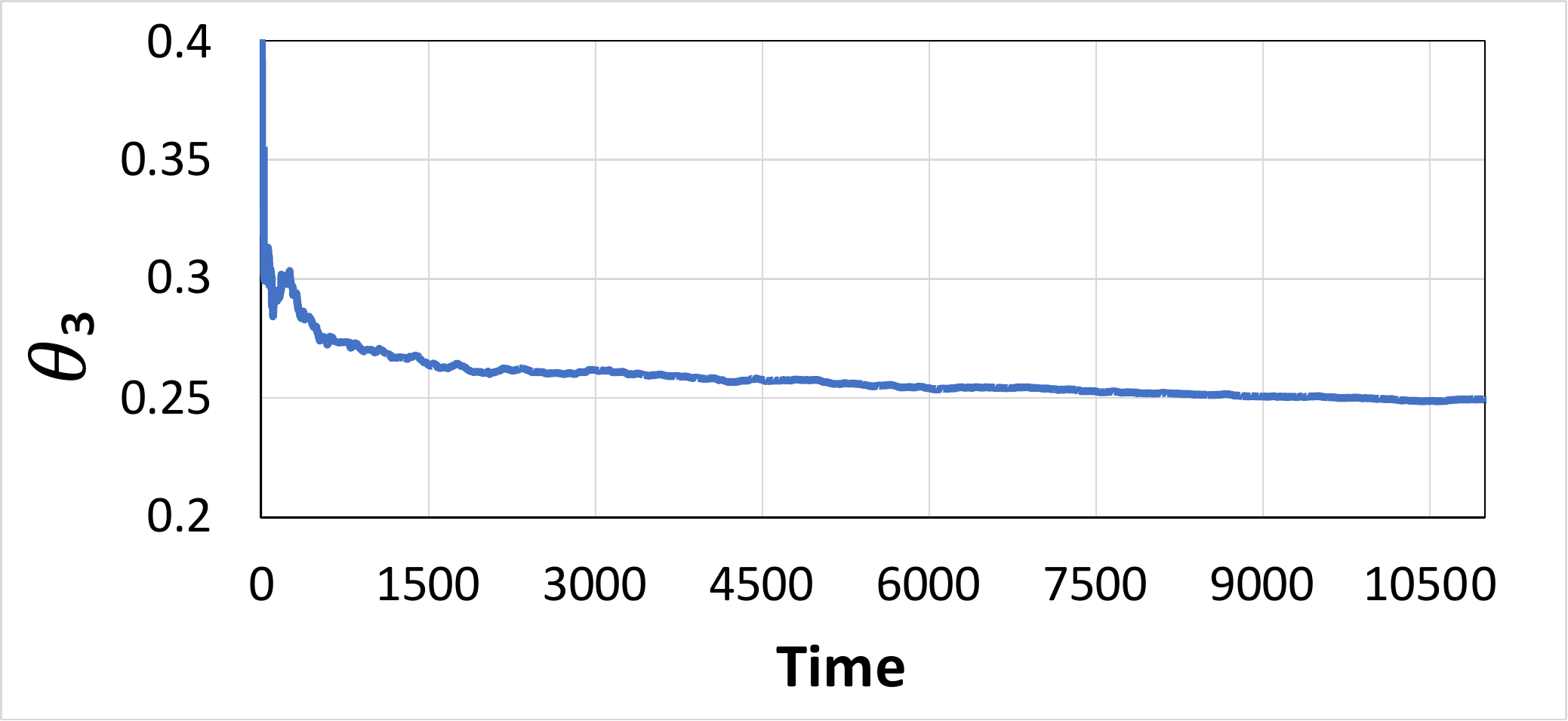} }
   \caption{Trajectories from the execution of \autoref{alg:online_disc} (left panel) and \autoref{alg:online_disc_ml} (right panel) for the estimation of $(\theta_1, \theta_2, \theta_3)$ from Model 4.  }
    \label{fig:par_est_model4}
\end{figure}

\subsubsection*{Acknowledgements}

AJ \& HR were supported by KAUST baseline funding. AB  acknowledges support from a Leverhulme Trust Prize. DC was partially supported by EU Synergy project STUOD - DLV-856408.
NK acknowledges funding by a JP Morgan A.I. Faculty award.

\appendix

\section{Derivation of \eqref{eq:gll}}
\label{sec:Deriv_Eq4}
Recall that  under $\overline{\mathbb{P}}_{\theta}$ the original processes (2.1)-(2.2) have dynamics:
\begin{align*}
dY_t  = dB_t; \quad dX_t  = b_\theta(X_t)dt+ \sigma(X_t)dW_t.
\end{align*}
We consider the processes:
\begin{align*}
dY_t  = dB_t;\quad 
dX_t  = \sigma(X_t)dW_t ,
\end{align*}
so that if $ \mathbb{P}$ denotes their law, 
then we have the Radon-Nikodym derivative: 
%
\begin{align*}
\frac{d\overline{\mathbb{P}}_{\theta}}{d \mathbb{P}} = \exp\Big\{\int_{0}^T b_{\theta}(X_s)^* a(X_s)^{-1} dX_s - \tfrac{1}{2}\int_{0}^T b_{\theta}(X_s)^*  a(X_s)^{-1} b_{\theta}(X_s)ds\Big\}.
\end{align*}
The log-likelihood is $\log(\gamma_{T,\theta}(1)) = \log \overline{\mathbb{E}}_\theta \left[Z_{T,\theta} | \mathcal{Y}_T \right] $, with 
$Z_{T,\theta}=d\mathbb{P}_{\theta}/d\overline{\mathbb{P}}_\theta$,
thus:
\begin{align*}
\nabla_{\theta}\log(\gamma_{T,\theta}(1)) 
= \frac{1}{\overline{\mathbb{E}}_\theta \left[Z_{T,\theta} | \mathcal{Y}_T \right]} \nabla_{\theta} \overline{\mathbb{E}}_\theta \big[ d\mathbb{P}_{\theta}/d\overline{\mathbb{P}}_\theta \big| \mathcal{Y}_T \big].
\end{align*}
For convenience we use the notation $\overline{\mathbb{E}}_{\theta,X}$, $\mathbb{E}_{X}$
for the marginal expectations w.r.t.~the original process $X$ and the $\theta$-free process $X$ defined above, respectively.
Notice  that we can write: 
\begin{align*}
\nabla_{\theta} \overline{\mathbb{E}}_\theta \big[ d\mathbb{P}_{\theta}/d\overline{\mathbb{P}}_\theta \big| \mathcal{Y}_T \big]  &=
\nabla_{\theta} \overline{\mathbb{E}}_{\theta,X} \big[ d\mathbb{P}_{\theta}/d\overline{\mathbb{P}}_\theta\big] \\
&= \nabla_{\theta} \mathbb{E}_{X} \big[ (d\mathbb{P}_{\theta}/d\overline{\mathbb{P}}_\theta) (d\overline{\mathbb{P}}_\theta/d\mathbb{P})\big] \\
& = \mathbb{E}_{X} \big[ Z_{T,\theta}\nabla_{\theta}\log(Z_{T,\theta}) (d\overline{\mathbb{P}}_\theta/d\mathbb{P})\big] +
  \mathbb{E}_{X} \big[ Z_{T,\theta}
  \nabla_{\theta}\log(d\overline{\mathbb{P}}_\theta/d\mathbb{P})
   (d\overline{\mathbb{P}}_\theta/d\mathbb{P})\big] \\
   &= 
  \overline{\mathbb{E}}_\theta[ (\nabla_{\theta}\log(Z_{T,\theta})+  \nabla_{\theta}\log(d\overline{\mathbb{P}}_\theta/d\mathbb{P})) Z_{T,\theta}]  . 
\end{align*}
One can now verify that, for $\lambda_{T,\theta}$ as defined in the main text:
\begin{align*}
\nabla_{\theta}\log(Z_{T,\theta})+  \nabla_{\theta}\log(d\overline{\mathbb{P}}_\theta/d\mathbb{P}) = \lambda_{T,\theta}.
\end{align*}

\section{$\mathbb{L}_r$ Bound for the Discretization Error}
\label{proof_thm_bias}

\subsection{Formulation}

We now consider proving a bound on ($\|\cdot\|_2$ is the $L_2-$norm for vectors)
$$
\mathbb{E}\left[\left\|\nabla_{\theta}\log(\gamma_{T,\theta}(1))  - \nabla_{\theta}\log(\gamma_{T,\theta}^l(1))\right\|_2^r\right]^{1/r}
$$
where
\begin{eqnarray*}
\nabla_{\theta}\log(\gamma_{T,\theta}(1)) & = & \frac{\overline{\mathbb{E}}_{\theta}\,[\,\lambda_{T,\theta}Z_{T,\theta}\,|\,\mathcal{Y}_T\,]}{\overline{\mathbb{E}}_{\theta}\,[\,Z_{T,\theta}\,|\,\mathcal{Y}_T\,]}\\
\nabla_{\theta}\log(\gamma_{T,\theta}^l(1)) & = & \frac{\overline{\mathbb{E}}_{\theta}\,[\,\lambda_{T,\theta}^l(\widetilde{X}_0,\widetilde{X}_{\Delta_l},\dots,\widetilde{X}_T)
\,Z_{T,\theta}^l(\widetilde{X}_0,\widetilde{X}_{\Delta_l},\dots,\widetilde{X}_T)\,|\,\mathcal{Y}_T\,]}{\overline{\mathbb{E}}_{\theta}\,[\,Z_{T,\theta}^l(\widetilde{X}_0,\widetilde{X}_{\Delta_l},\dots,\widetilde{X}_T)\,|\,\mathcal{Y}_T\,]}.
\end{eqnarray*}
We begin by noting that for any random vector $X$ of dimension $d$ and finite $r-$moments that
$$
\mathbb{E}[\|X\|_2^r]^{1/r} \leq \left(\sum_{i=1}^d\mathbb{E}[|X^{(i)}|^r]^{2/r}\right)^{1/2}.
$$
As a result, it will suffice to control for each $i\in\{1,\dots,d_{\theta}\}$
$$
\mathbb{E}_{\theta}\left[\left|\nabla_{\theta}\log(\gamma_{T,\theta}(1))^{(i)}  - \nabla_{\theta}\log(\gamma_{T,\theta}^l(1))^{(i)}\right|^r\right]^{1/r}.
$$
Throughout all of our proofs, $C$ is a deterministic constant whose value will change upon each appearance. In addition we supress any dependencies on $\theta$ below.

\subsection{Technical Results}

\begin{lem}\label{lem:diff1}
Assume (D\ref{hyp_diff:1}). Then for any $(r,T)\in[1,\infty)\times\mathbb{N}$ there exists a $C<+\infty$ such that 
$$
\overline{\mathbb{E}}_{\theta}\left[\frac{1}{\overline{\mathbb{E}}_{\theta}\,[\,Z_{T,\theta}\,|\,\mathcal{Y}_T\,]^r}\right] \leq C.
$$
\end{lem}

\begin{proof}
By the conditional Jensen inequality
$$
\overline{\mathbb{E}}_{\theta}\left[\frac{1}{\overline{\mathbb{E}}_{\theta}\,[\,Z_{T,\theta}\,|\,\mathcal{Y}_T\,]^r}\right] \leq 
\overline{\mathbb{E}}_{\theta}\left[Z_{T,\theta}^{-r}\right].
$$
It is now simple to use the properties of the process under study to deduce the result.
\end{proof}

\begin{rem}\label{rem:z_bd}
A standard result is that $\overline{\mathbb{E}}_{\theta}\left[Z_{T,\theta}^{r}\right]\leq C$ for any fixed $r\in\mathbb{R}$.
\end{rem}

\begin{lem}\label{lem:diff2}
Assume (D\ref{hyp_diff:1}). Then for any $(r,T)\in[1,\infty)\times\mathbb{N}$ there exists a $C<+\infty$ such that for any $l\in\mathbb{N}_0$
$$
\overline{\mathbb{E}}_{\theta}\left[\frac{1}{\overline{\mathbb{E}}_{\theta}\,[\,Z_{T,\theta}^l(\widetilde{X}_0,\widetilde{X}_{\Delta_l},\dots,\widetilde{X}_T)\,|\,\mathcal{Y}_T\,]^r}\right] \leq C.
$$
\end{lem}
\begin{proof}
By the conditional Jensen inequality
$$
\overline{\mathbb{E}}_{\theta}\left[\frac{1}{\overline{\mathbb{E}}_{\theta}\,[\,Z_{T,\theta}^l(\widetilde{X}_0,\widetilde{X}_{\Delta_l},\dots,\widetilde{X}_T)\,|\,\mathcal{Y}_T\,]^r}\right] \leq 
\overline{\mathbb{E}}_{\theta}\left[Z_{T,\theta}^l(\widetilde{X}_0,\widetilde{X}_{\Delta_l},\dots,\widetilde{X}_T)^{-r}\right].
$$
The result now follows by the arguments stated in \cite[eq.~(20)-(21)]{high_freq_ml}.
\end{proof}

\begin{rem}\label{rem:z_lbd}
By the arguments stated in \cite[eq.~(20)-(21)]{high_freq_ml}  $\overline{\mathbb{E}}_{\theta}\left[Z_{T,\theta}^l(\widetilde{X}_0,\widetilde{X}_{\Delta_l},\dots,\widetilde{X}_T)^{r}\right]\leq C$ for any fixed $r\in\mathbb{R}$ and $C$ does not depend on $l$.
\end{rem}

\begin{lem}\label{lem:diff3}
Assume (D\ref{hyp_diff:1}). Then for any $(r,T)\in[1,\infty)\times\mathbb{N}$ there exists a $C<+\infty$ such that for any $(l,i)\in\mathbb{N}_0\times\{1,\dots,d_{\theta}\}$
$$
\overline{\mathbb{E}}_{\theta}\left[\left|\overline{\mathbb{E}}_{\theta}\,[\,\lambda_{T,\theta}^l(\widetilde{X}_0,\widetilde{X}_{\Delta_l},\dots,\widetilde{X}_T)^{(i)}
\,Z_{T,\theta}^l(\widetilde{X}_0,\widetilde{X}_{\Delta_l},\dots,\widetilde{X}_T)\,|\,\mathcal{Y}_T\,]\right|^r\right] \leq C.
$$
\end{lem}

\begin{proof}
Using the conditional Jensen inequality, Cauchy-Schwarz and Remark \ref{rem:z_lbd} it suffices to bound
$$
\overline{\mathbb{E}}_{\theta}\left[\left|\overline{\mathbb{E}}_{\theta}\,[\,\lambda_{T,\theta}^l(\widetilde{X}_0,\widetilde{X}_{\Delta_l},\dots,\widetilde{X}_T)^{(i)}\right|^{2r}\right].
$$
Via the $C_{2r}-inequality$ one can then simply focus on the 3 terms
\begin{eqnarray*}
T_1 & = & \overline{\mathbb{E}}_{\theta}\left[\left|\sum_{k=0}^{T/\Delta_l-1}\,\left(\frac{\partial}{\partial\theta_i} b_{\theta}(x_{k\Delta_l})\right)^*a(x_{k\Delta_l})^{-1}\sigma(x_{k\Delta_l})(W_{(k+1)\Delta_l}-W_{k\Delta_l})\right|^{2r}\right] \\
T_2 & = & \overline{\mathbb{E}}_{\theta}\left[\left|\sum_{k=0}^{T/\Delta_l-1}
\left(\frac{\partial}{\partial\theta_i}h_{\theta}(x_{k\Delta_l})\right)^*(Y_{(k+1)\Delta_l}-Y_{k\Delta_l})
\right|^{2r}\right]\\
T_3 & = & \overline{\mathbb{E}}_{\theta}\left[\left|\sum_{k=0}^{T/\Delta_l-1}
\left(\frac{\partial}{\partial\theta_i} h_{\theta}(x_{k\Delta_l})\right)^*h_{\theta}(x_{k\Delta_l})\Delta_l
\right|^{2r}\right].
\end{eqnarray*}
To bound $T_1$ and $T_2$ one can simply apply the Burkholder-Gundy-Davis inequality and combine this with the boundedness of the terms which are functions of $x$; this is a standard argument in the literature. The bound on $T_3$ is immediate by the boundedness of the summands. This concludes the proof.
\end{proof}

\begin{rem}\label{rem:lam_lbd}
It is more-or-less the same argument as in the proof of Lemma \ref{lem:diff3} to deduce that $\overline{\mathbb{E}}_{\theta}\left[\left|\lambda_{T,\theta}^{(i)}\right|^r\right]\leq C$.
\end{rem}

\begin{lem}\label{lem:diff4}
Assume (D\ref{hyp_diff:1}). Then for any $(r,T)\in[1,\infty)\times\mathbb{N}$ there exists a $C<+\infty$ such that for any $(l,i)\in\mathbb{N}_0\times\{1,\dots,d_{\theta}\}$
$$
\overline{\mathbb{E}}_{\theta}\left[\left|
\overline{\mathbb{E}}_{\theta}\,[\,\lambda_{T,\theta}^{(i)}Z_{T,\theta}\,|\,\mathcal{Y}_T\,]-
\overline{\mathbb{E}}_{\theta}\,[\,\lambda_{T,\theta}^l(\widetilde{X}_0,\widetilde{X}_{\Delta_l},\dots,\widetilde{X}_T)^{(i)}
\,Z_{T,\theta}^l(\widetilde{X}_0,\widetilde{X}_{\Delta_l},\dots,\widetilde{X}_T)\,|\,\mathcal{Y}_T\,]\right|^r\right]^{1/r} \leq C\Delta_l^{1/2}.
$$
\end{lem}

\begin{proof}
Applying conditional Jensen and Minkowski we have the upper-bound
$$
\overline{\mathbb{E}}_{\theta}\left[\left|
\overline{\mathbb{E}}_{\theta}\,[\,\lambda_{T,\theta}^{(i)}Z_{T,\theta}\,|\,\mathcal{Y}_T\,]-
\overline{\mathbb{E}}_{\theta}\,[\,\lambda_{T,\theta}^l(\widetilde{X}_0,\widetilde{X}_{\Delta_l},\dots,\widetilde{X}_T)^{(i)}
\,Z_{T,\theta}^l(\widetilde{X}_0,\widetilde{X}_{\Delta_l},\dots,\widetilde{X}_T)\,|\,\mathcal{Y}_T\,]\right|^r\right]^{1/r}\leq 
T_1 + T_2
$$
where
\begin{eqnarray*}
T_1 & = & \overline{\mathbb{E}}_{\theta}\left[\left|\left\{\lambda_{T,\theta}^l(\widetilde{X}_0,\widetilde{X}_{\Delta_l},\dots,\widetilde{X}_T)^{(i)}-\lambda_{T,\theta}^{(i)}\right\}
\,Z_{T,\theta}^l(\widetilde{X}_0,\widetilde{X}_{\Delta_l},\dots,\widetilde{X}_T)\right|^r\right]^{1/r}\\
T_2 & = & \overline{\mathbb{E}}_{\theta}\left[\left|
\left\{Z_{T,\theta}^l(\widetilde{X}_0,\widetilde{X}_{\Delta_l},\dots,\widetilde{X}_T)-Z_{T,\theta}\right\}\lambda_{T,\theta}^{(i)}
\right|^r\right]^{1/r}.
\end{eqnarray*}
For $T_1$ one can use Cauchy-Schwarz and the result in Remark \ref{rem:z_lbd} to deduce the upper-bound
$$
T_1 \leq C\overline{\mathbb{E}}_{\theta}\left[\left|\lambda_{T,\theta}^l(\widetilde{X}_0,\widetilde{X}_{\Delta_l},\dots,\widetilde{X}_T)^{(i)}-\lambda_{T,\theta}^{(i)}\right|^{2r}\right]^{1/(2r)}
$$
Then the term on the R.H.S.~can be dealt with by using standard results in the discretization of Riemann-integrals coupled with the Burkholder-Gundy-Davis inequality, Lipschitz properties of the various functions and Euler discretizations. As the results are almost identical to the calculations in \cite[pp.589]{crisan} they are omitted. That is, one can deduce that
$$
T_1\leq C\Delta_l^{1/2}.
$$
For $T_2$, again, using Cauchy-Schwarz and the result in Remark \ref{rem:lam_lbd} we have
$$
T_2 \leq C \overline{\mathbb{E}}_{\theta}\left[\left|Z_{T,\theta}^l(\widetilde{X}_0,\widetilde{X}_{\Delta_l},\dots,\widetilde{X}_T)-Z_{T,\theta}
\right|^{2r}\right]^{1/(2r)}.
$$
Then by \cite[Lemma A.5.]{high_freq_ml}
$$
T_2\leq C\Delta_l^{1/2}.
$$
The end of the proof is now clear.
\end{proof}

\subsection{Proof of \autoref{thm:bias} }
\begin{proof}
We need only to bound
$$
T:=\mathbb{E}_{\theta}\left[\left|
\frac{\overline{\mathbb{E}}_{\theta}\,[\,\lambda_{T,\theta}^{(i)}Z_{T,\theta}\,|\,\mathcal{Y}_T\,]}{\overline{\mathbb{E}}_{\theta}\,[\,Z_{T,\theta}\,|\,\mathcal{Y}_T\,]}-
\frac{\overline{\mathbb{E}}_{\theta}\,[\,\lambda_{T,\theta}^l(\widetilde{X}_0,\widetilde{X}_{\Delta_l},\dots,\widetilde{X}_T)^{(i)}
\,Z_{T,\theta}^l(\widetilde{X}_0,\widetilde{X}_{\Delta_l},\dots,\widetilde{X}_T)\,|\,\mathcal{Y}_T\,]}{\overline{\mathbb{E}}_{\theta}\,[\,Z_{T,\theta}^l(\widetilde{X}_0,\widetilde{X}_{\Delta_l},\dots,\widetilde{X}_T)\,|\,\mathcal{Y}_T\,]}
\right|^r\right]^{1/r}.
$$
Now we have by Cauchy-Schwarz
$$
T \leq T_1T_2
$$
where
\begin{eqnarray*}
T_1 & = & \overline{\mathbb{E}}_{\theta}[Z_{T,\theta}^{2r}]^{1/(2r)}\\
T_2 & = & \overline{\mathbb{E}}_{\theta}\left[\left|
\frac{\overline{\mathbb{E}}_{\theta}\,[\,\lambda_{T,\theta}^{(i)}Z_{T,\theta}\,|\,\mathcal{Y}_T\,]}{\overline{\mathbb{E}}_{\theta}\,[\,Z_{T,\theta}\,|\,\mathcal{Y}_T\,]}-
\frac{\overline{\mathbb{E}}_{\theta}\,[\,\lambda_{T,\theta}^l(\widetilde{X}_0,\widetilde{X}_{\Delta_l},\dots,\widetilde{X}_T)^{(i)}
\,Z_{T,\theta}^l(\widetilde{X}_0,\widetilde{X}_{\Delta_l},\dots,\widetilde{X}_T)\,|\,\mathcal{Y}_T\,]}{\overline{\mathbb{E}}_{\theta}\,[\,Z_{T,\theta}^l(\widetilde{X}_0,\widetilde{X}_{\Delta_l},\dots,\widetilde{X}_T)\,|\,\mathcal{Y}_T\,]}
\right|^{2r}\right]^{1/(2r)}.
\end{eqnarray*}
By the result in Remark \ref{rem:z_bd} $T_1\leq C$, so we need only consider $T_2$. We have by using a standard decomposition and the Minkowski inequality that
$$
T_2 \leq T_3 + T_4
$$
where
\begin{eqnarray*}
T_3 & = & \overline{\mathbb{E}}_{\theta}\Bigg[\Bigg|\frac{\overline{\mathbb{E}}_{\theta}\,[\,\lambda_{T,\theta}^l(\widetilde{X}_0,\widetilde{X}_{\Delta_l},\dots,\widetilde{X}_T)^{(i)}
\,Z_{T,\theta}^l(\widetilde{X}_0,\widetilde{X}_{\Delta_l},\dots,\widetilde{X}_T)\,|\,\mathcal{Y}_T\,]}{\overline{\mathbb{E}}_{\theta}\,[\,Z_{T,\theta}\,|\,\mathcal{Y}_T\,]\overline{\mathbb{E}}_{\theta}\,[\,Z_{T,\theta}^l(\widetilde{X}_0,\widetilde{X}_{\Delta_l},\dots,\widetilde{X}_T)\,|\,\mathcal{Y}_T\,]}
\times \\ & &
\left\{\overline{\mathbb{E}}_{\theta}[Z_{T,\theta}^l(\widetilde{X}_0,\widetilde{X}_{\Delta_l},\dots,\widetilde{X}_T)|\mathcal{Y}_t]-\overline{\mathbb{E}}_{\theta}[Z_{T,\theta}|\mathcal{Y}_T]\right\}\Bigg|^{2r}\Bigg]^{1/(2r)}\\
T_4 & = & \overline{\mathbb{E}}_{\theta}\Bigg[\Bigg|\frac{1}{\overline{\mathbb{E}}_{\theta}\,[\,Z_{T,\theta}\,|\,\mathcal{Y}_T\,]}\Big\{
\overline{\mathbb{E}}_{\theta}\,[\,\lambda_{T,\theta}^{(i)}Z_{T,\theta}\,|\,\mathcal{Y}_T\,]-\\ & &
\overline{\mathbb{E}}_{\theta}\,[\,\lambda_{T,\theta}^l(\widetilde{X}_0,\widetilde{X}_{\Delta_l},\dots,\widetilde{X}_T)^{(i)}Z_{T,\theta}^l(\widetilde{X}_0,\widetilde{X}_{\Delta_l},\dots,\widetilde{X}_T)\,|\,\mathcal{Y}_T\,]
\Big\}\Bigg|^{2r}\Bigg]^{1/(2r)}.
\end{eqnarray*}
For $T_3$ one can apply Cauchy-Schwarz and \cite[Lemma A.5.]{high_freq_ml} we have the upper-bound
$$
T_3 \leq C\Delta_l^{1/2}\overline{\mathbb{E}}_{\theta}\Bigg[\Bigg|\frac{\overline{\mathbb{E}}_{\theta}\,[\,\lambda_{T,\theta}^l(\widetilde{X}_0,\widetilde{X}_{\Delta_l},\dots,\widetilde{X}_T)^{(i)}
\,Z_{T,\theta}^l(\widetilde{X}_0,\widetilde{X}_{\Delta_l},\dots,\widetilde{X}_T)\,|\,\mathcal{Y}_T\,]}{\overline{\mathbb{E}}_{\theta}\,[\,Z_{T,\theta}\,|\,\mathcal{Y}_T\,]\overline{\mathbb{E}}_{\theta}\,[\,Z_{T,\theta}^l(\widetilde{X}_0,\widetilde{X}_{\Delta_l},\dots,\widetilde{X}_T)\,|\,\mathcal{Y}_T\,]}\Bigg|^{4r}\Bigg]^{1/(4r)}.
$$
Now for the expectation on the R.H.S.~one can use the H\"older inequality along with Lemmata \ref{lem:diff1}-\ref{lem:diff3} to deduce that
$$
T_3 \leq C\Delta_l^{1/2}.
$$
For $T_4$ using Cauchy-Schwarz and Lemma \ref{lem:diff4} we have the upper-bound
$$
T_4 \leq C\Delta_l^{1/2}\overline{\mathbb{E}}_{\theta}\Bigg[\Bigg|\frac{1}{\overline{\mathbb{E}}_{\theta}\,[\,Z_{T,\theta}\,|\,\mathcal{Y}_T\,]}\Bigg|^{4r}\Bigg]^{1/(4r)}.
$$
Applying Lemma \ref{lem:diff1} we have
$$
T_4 \leq C\Delta_l^{1/2}
$$
from which we conclude.
\end{proof}


\begin{thebibliography}{99}

\bibitem{crisan_bain}
{\sc Bain}, A.~\& {\sc Crisan}, D.~(2009). \emph{Fundamentals of Stochastic Filtering}. Springer: New York.

\bibitem{wasser}
{\sc Ballesio}, M., {\sc Jasra}, A., {\sc von Schwerin}, E. \& {\sc Tempone}, R.~(2020).
A Wasserstein coupled particle filter for multilevel estimation. arXiv:2004.03981.

\bibitem{BMP90}{\textsc{Benveniste}, A., \textsc{M\'etivier}, M. \& \textsc{Priouret}, P. (1990).
\emph{Adaptive Algorithms and Stochastic Approximation}. New York:
Springer-Verlag.}

\bibitem{beskos2}
{\sc Beskos}, A., {\sc Papaspiliopoulos}, O., {\sc Roberts}, G., {\sc Fearnhead}, P.~(2006).
Exact and computationally efficient likelihood-based estimation for discretely observed diffusion processes (with discussion).
\emph{J. R. Statist. Soc. Ser. B}, {\bf 68}, 333-382. 


\bibitem{bierkens2020}
{\sc Bierkens}, J., {\sc Van Der Meulen}, F.,  {\sc Schauer}, M.~(2020).
Simulation of elliptic and hypo-elliptic conditional diffusions,
\emph{Advances in Applied Probability}, {\bf 52}, 173-212.


\bibitem{botha_sde_pmcmc}
\textsc{Botha}, I., \textsc{Kohn}, R., \& \textsc{Drovandi}, C. (2020). Particle methods for stochastic differential equation mixed effects models. \emph{Bayes. Anal.}  (to appear).


\bibitem{campillo}
{\sc Campillo}, F. \& {\sc Le Gland}, F.~(1989). Maximum likelihood estimation for partially observed diffusions: Direct Maximization vs The EM algorithm.
\emph{Stoch. Proc. Appl.}, {\bf 33}, 245--274.


\bibitem{lan_mu}
{\sc Cliffe}, K. A., {\sc Giles}, M. B., {\sc Scheichl}, R., \& {\sc Teckentrup}, A. L.~(2011). Multilevel Monte Carlo methods and applications to elliptic PDEs with random coefficients. 
\emph{Comput. Vis. Sci.}, {\bf 14}, 3--15.


\bibitem{crisan}
{\sc Crisan}, D.~(2011). Discretizing the continuous-time filtering problem: Order of Convergence. In \cite{handbook}, 572--597. 

\bibitem{handbook}
\textsc{Crisan}, D., \& \textsc{Rozovskii}, B. (2011).\emph{ The Oxford Handbook of Nonlinear Filtering}. Oxford University Press.

\bibitem{FK}
{\sc Del Moral}, P.~(2004).\emph{  Feynman-Kac Formulae}. Springer.

\bibitem{backward}
{\sc Del Moral}, P., {\sc Doucet}, A., \& {\sc Singh} S. S.~(2010). A backward particle interpretation
of Feynman-Kac formuale. \emph{M2AN}, {\bf 44}, 947--975.

\bibitem{forward_smoothing}
{\sc Del Moral}, P., {\sc Doucet}, A., \& {\sc Singh} S. S.~(2010). Forward smoothing using sequential Monte Carlo, arXiv:1012.5390 

\bibitem{etienne_bwd}
\textsc{Etienne}, M. P., \textsc{Gloaguen}, P., \textsc{Corff}, S. L., \& \textsc{Olsson}, J. (2020). Backward importance sampling for partially observed diffusion processes. arXiv:2002.05438.


\bibitem{glo}
{\sc Gloaguen}, P., {\sc Etienne}, M. P.  \& {\sc Le Corff}, S.~(2018).
Online sequential Monte Carlo smoother for partially observed diffusion processes.
\emph{EURASIP J. Adv. Sig. Proc}, article 9. 

\bibitem{mlpf}
{\sc Jasra}, A., {\sc Kamatani}, K., {\sc Law} K. J. H. \& {\sc Zhou}, Y.~(2017). 
Multilevel particle filters. \emph{SIAM J. Numer. Anal.}, {\bf 55}, 3068-3096.


\bibitem{ub_filt}
{\sc Jasra}, A., {\sc Law}, K. J. H., \& {\sc Yu}, F.~(2020). Unbiased filtering of a class of partially observed diffusions. arXiv:2002.03747.

\bibitem{cpf_clt}
{\sc Jasra}, A., \& {\sc Yu}, F.~(2020). Central limit theorems for coupled particle filters.
\emph{Adv. Appl. Probab.} {\bf 52}, 942--1001.

\bibitem{high_freq_ml}
{\sc Jasra}, A., {\sc Yu}, F. \& {\sc Heng}, J.~(2020). Multilevel particle filters for the nonlinear filtering problem in continuous time. \emph{Stat. Comp.} {\bf 30}, 1381--1402.

\bibitem{legland1997}\textsc{Le Gland}, F. and \textsc{Mevel}, M. (1997). Recursive
identification in hidden Markov models. \emph{Proc. 36th IEEE Conf.
Dec. Contr.}, 3468-3473.


\bibitem{vd_meulen_smoothing}
{\sc Mider}, M., \textsc{Schauer}, M. \&
\textsc{van der Meulen}, F.~(2020). Continuous-discrete smoothing of diffusions. arXiv:1712.03807.


\bibitem{paris}
{\sc Olsson}, J. \& {\sc Westerborn}, J.~(2017).
Efficient particle-based online smoothing in general hidden Markov models: The PaRIS algorithm. \emph{Bernoulli}, {\bf 23}, 1951-1996.


\bibitem{omiros_DA}
\textsc{Papaspiliopoulos}, O., \textsc{Roberts}, G. O., \& \textsc{Stramer}, O. (2013). Data augmentation for diffusions. \emph{J. Comp. Graph. Stat.}, {\bf 22}, 665-688.



\bibitem{omiros_ISdiffusions}
{\sc Papaspiliopoulos}, O. \& {\sc Roberts}, G.~(2012).
Importance sampling techniques for estimation of diffusion models. \emph{Stat. Meth. Stoch. Diff. Eq.}, {\bf 124}, 311-340.

\bibitem{picard}
\textsc{Picard}, J. (1984). Approximation of nonlinear filtering problems and order of convergence. In \emph{Filtering and control of random processes}, 219-236, Springer, Berlin, Heidelberg.


\bibitem{poyia}
\textsc{Poyiadjis}, G., \textsc{Doucet}, A., \& \textsc{Singh}, S. S. (2011). Particle approximations of the score and observed information matrix in state space models with application to parameter estimation. \textit{Biometrika}, {\bf 98}, 65-80.



\bibitem{sarkka_pf_girsanov}
\textsc{S{\"a}rkk{\"a}}, S., \& \textsc{Sottinen}, T. (2008). Application of Girsanov theorem to particle filtering of discretely observed continuous-time non-linear systems. \emph{Bayes. Anal.}, {\bf 3}, 555-584.


\bibitem{schauer}
{\sc Schauer}, M., {\sc van der Meulen}, F. \& {\sc van Zanten}, H.~(2017). Guided proposals for simulating multi-dimensional diffusion bridges. \emph{Bernoulli}, {\bf 23}, 2917--2950.

\bibitem{coup_pf}
{\sc Sen}, D., {\sc Thiery},  A., {\sc Jasra}, A.~(2018). On coupling particle filters. \emph{Statist. Comp.}, {\bf 28}, 461-475.


\bibitem{surace18} 
{\sc Surace}, S. C., \& {\sc Pfister}, J. P. (2018). Online Maximum-Likelihood Estimation of the Parameters of Partially Observed Diffusion Processes. \emph{IEEE Transactions on Automatic Control}, 64(7), 2814-2829.


\bibitem{talay}
\textsc{Talay}, D. (1984). Efficient numerical schemes for the approximation of expectations of functionals of the solution of a SDE, and applications. In  \emph{Filtering and control of random processes}, 294-313, Springer, Berlin, Heidelberg.


\bibitem{thor}
{\sc Thorisson}, H.~(2000). \emph{Coupling, stationarity, and regeneration}. Springer:New York.


\bibitem{vd_meulen_guided_mcmc}
\textsc{van der Meulen}, F., \& \textsc{Schauer}, M. (2017). Bayesian estimation of discretely observed multi-dimensional diffusion processes using guided proposals. \emph{Elec. J. Stat.}, {\bf 11}, 2358-2396. 


\bibitem{whitaker_guided}
\textsc{Whitaker}, G. A., \textsc{Golightly}, A., \textsc{Boys}, R. J., \& \textsc{Sherlock}, C. (2017). Improved bridge constructs for stochastic differential equations. \emph{Stat. Comp.}, {\bf 27}, 885-900.

\bibitem{shouto}
{\sc Yonekura}, S. \& {\sc Beskos}, A. ~(2020). Online smoothing for diffusion processes observed with noise. arXiv: 200312247.


\end{thebibliography}
\end{document}